\documentclass[11pt]{article}
\usepackage[a4paper,margin=28mm]{geometry}
\usepackage[T1]{fontenc}
\usepackage{lmodern,microtype}
\usepackage{amsmath,amssymb,amsthm,mathtools,mathrsfs}
\DeclarePairedDelimiterX{\KetDelimiters}[1]{\lvert}{\rangle}{\mathstrut #1}
\DeclarePairedDelimiterX{\BraDelimiters}[1]{\langle}{\rvert}{\mathstrut #1}
\DeclarePairedDelimiterX{\BraketDelimiters}[2]{\langle}{\rangle}{%
  \mathstrut #1\,\delimsize\vert\,#2}
\newcommand{\ket}[1]{\KetDelimiters*{#1}}

\usepackage{bm}
\usepackage{graphicx,booktabs,array,flafter}
\usepackage{tikz}
\usetikzlibrary{arrows.meta}
\tikzset{every picture/.append style={execute at begin picture={\sffamily}},
  rolpanel/.style={font=\small\sffamily\bfseries}}
\definecolor{linkteal}{RGB}{22,123,136}
\definecolor{defectorange}{RGB}{185,102,22}
\definecolor{linkgray}{RGB}{98,118,133}
\usepackage[colorlinks=true,linkcolor=blue,citecolor=blue,urlcolor=blue]{hyperref}
\usepackage{aliascnt}
\usepackage[nameinlink,capitalize]{cleveref}
\newtheorem{proposition}{Proposition}[section]
\newaliascnt{theorem}{proposition}
\newtheorem{theorem}[theorem]{Theorem}
\aliascntresetthe{theorem}
\newaliascnt{lemma}{proposition}
\newtheorem{lemma}[lemma]{Lemma}
\aliascntresetthe{lemma}
\newaliascnt{corollary}{proposition}
\newtheorem{corollary}[corollary]{Corollary}
\aliascntresetthe{corollary}
\crefname{proposition}{Proposition}{Propositions}
\Crefname{proposition}{Proposition}{Propositions}
\crefname{theorem}{Theorem}{Theorems}
\Crefname{theorem}{Theorem}{Theorems}
\crefname{corollary}{Corollary}{Corollaries}
\Crefname{corollary}{Corollary}{Corollaries}
\crefname{lemma}{Lemma}{Lemmas}
\Crefname{lemma}{Lemma}{Lemmas}
\theoremstyle{definition}
\newaliascnt{definition}{proposition}
\newtheorem{definition}[definition]{Definition}
\aliascntresetthe{definition}
\crefname{definition}{Definition}{Definitions}
\Crefname{definition}{Definition}{Definitions}
\newaliascnt{remark}{proposition}
\newtheorem{remark}[remark]{Remark}
\aliascntresetthe{remark}
\crefname{remark}{Remark}{Remarks}
\Crefname{remark}{Remark}{Remarks}
\newtheorem*{informaltheorem}{Theorem}
\usepackage{tcolorbox}
\tcbuselibrary{breakable,skins}
\tcbset{rolstatement/.style={
  enhanced jigsaw,breakable,sharp corners,
  colframe=linkgray!50,colback=linkgray!4,boxrule=.4pt,
  boxsep=0pt,left=9pt,right=9pt,top=6pt,bottom=6pt,
  before skip=9pt,after skip=9pt
}}
\newenvironment{conceptbox}{\begin{tcolorbox}[rolstatement]}{\end{tcolorbox}}
\newif\ifROLstatementboxes
\ROLstatementboxestrue
\newcommand{\ROLboxenvironment}[1]{%
  \AddToHook{env/#1/before}{\ifROLstatementboxes\begin{tcolorbox}[rolstatement]\fi}%
  \AddToHook{env/#1/after}{\ifROLstatementboxes\end{tcolorbox}\fi}%
}
\ROLboxenvironment{theorem}
\ROLboxenvironment{proposition}
\ROLboxenvironment{lemma}
\ROLboxenvironment{corollary}
\ROLboxenvironment{informaltheorem}
\newcommand{\HH}{\mathcal H}
\newcommand{\II}{\mathcal I}
\newcommand{\YY}{\mathcal Y}
\newcommand{\Energy}{\mathcal E}
\newcommand{\Eff}{\mathcal E^{\mathrm{eff}}}
\DeclareMathOperator{\Av}{Av}
\DeclareMathOperator{\Var}{Var}
\newcommand{\GG}{\mathcal G}
\newcommand{\LOp}{\bm{L}}
\newcommand{\DOp}{\bm{D}}
\newcommand{\AOp}{\bm{A}}

\newcommand{\TOp}{\bm{T}}
\newcommand{\LowerLOp}{\underline{\LOp}}
\newcommand{\LowerDOp}{\underline{\DOp}}
\newcommand{\QOp}{\bm{Q}}

\newcommand{\lk}{\operatorname{lk}}
\newcommand{\id}{\mathrm{id}}
\newcommand{\C}{\mathbb C}
\DeclareMathOperator{\ran}{ran}

\title{\LARGE
Recursive Operator Localization and Spectral Gap Descent\\
via the Geometry of Local Spectral Expansion
}
\author{\large Ryu Hayakawa\thanks{%
\begin{tabular}[t]{@{}l@{}}
Yukawa Institute for Theoretical Physics \& The Hakubi Center, Kyoto University, Japan\\
\href{mailto:ryu.hayakawa@yukawa.kyoto-u.ac.jp}{\texttt{ryu.hayakawa@yukawa.kyoto-u.ac.jp}}
\end{tabular}}}
\date{}

\begin{document}

\noindent
\hspace{\fill} YITP-26-136
\begingroup
\let\newpage\relax
\maketitle
\endgroup

\begin{abstract}
We introduce Recursive Operator Localization (ROL), a framework for constructing hierarchies of local energy problems linked by linear maps between Hilbert spaces. To address spectral-gap problems, including those arising in quantum systems, we study a realization of ROL by nested subspaces indexed by a weighted simplicial complex. Using spectral gaps of auxiliary classical random walks on links, we propagate geometric estimates upward and use them to turn local spectral gaps at the bottom of the hierarchy into a spectral-gap bound for the global operator. For weighted scalar-valued functions on facets, the ROL decomposition becomes the law of total variance, recovering classical spectral-gap bounds from local link expansion (Alev and Lau, STOC 2020). Moreover, for positive binary spin distributions, conditional influence describes the overlap of the corresponding subspaces, yielding known heat-bath gap bounds from spectral independence (Anari, Liu, and Oveis Gharan, FOCS 2020). The framework also offers a route toward quantum Gibbs sampling, with applications requiring model-specific geometric estimates and energy comparisons.
\end{abstract}

\section{Introduction}\label{sec:introduction}
Lower bounds on spectral gaps provide quantitative control of directions
outside the kernel of a positive operator. In a quantum Hamiltonian,
the gap above the ground-state energy measures the separation between
the ground-state space and excited states. The generator gap of a
reversible quantum Markov semigroup instead controls relaxation toward
the stationary subspace in the corresponding Hilbert-space norm.
However, obtaining quantitative gap bounds in many-body systems is
difficult, especially when their dependence on system size and model
parameters must be controlled.

Local-to-global methods address this difficulty by combining tractable
local estimates with control of how the local problems fit together
to obtain a global spectral-gap bound.
Here locality need not mean spatial
proximity: local problems may arise from restrictions to spatial regions,
conditioning on selected variables, passage to links of graphs or hypergraphs, or decompositions
of a state space or Hilbert space.

A useful perspective on organizing local spectral information comes
from geometric approaches to cohomological vanishing.
Bochner's method uses curvature identities to control harmonic forms
and establish vanishing under suitable positivity
conditions~\cite{bochner1948curvature}.
Garland's method provides a combinatorial spectral counterpart,
deriving global spectral bounds and cohomological vanishing from spectral
conditions on simplicial links~\cite{garland1973padic}.
For our purposes, the guiding idea is to combine local energy estimates
into global spectral bounds. Such bounds can establish cohomological
vanishing or control the positive spectrum.
Oppenheim developed a method for propagating spectral estimates across successive levels of links, known as spectral descent~\cite{oppenheim2017local}.
In modern high-dimensional expansion~\cite{gotlib2023nowhere}, related local-to-global
estimates control higher-order random walks~\cite{alev2020improved}
and support classical sampling algorithms, with spectral independence
providing a source of local spectral control~\cite{anari2020spectral,chen2025spectral}.

Beyond this link-based approach, other relevant local-to-global methods
include finite-size criteria and
detectability-lemma methods for frustration-free
Hamiltonians~\cite{anshu2019improved,anshu2016detectability};
subspace geometry and spectral criteria for property~\textup{(T)}~\cite{ershov2010property,kassabov2011subspace};
classical variance and entropy
factorization methods~\cite{caputo2021block};
decomposition of Markov chains~\cite{jerrum2004elementary};
and coupling and curvature methods~\cite{blanca2021mixing,ollivier2007ricci}.
For quantum Gibbs sampling, correlation-decay conditions yield
spectral-gap bounds for commuting lattice systems
\cite{kastoryano2016quantum,lucia2025spectral}.
Complementary high-temperature approaches include quantum
Dobrushin conditions for rapid mixing~\cite{bakshi2026dobrushin}
and size-independent sampler gaps for bounded-strength all-to-all
systems~\cite{bergamaschi2026fast}.

\paragraph{Our perspective.}
Motivated by the high-dimensional expansion viewpoint described above,
we ask how a hierarchy of local energy problems can be organized and
analyzed in an operator setting beyond scalar-valued functions.
We begin by separating the hierarchy of spaces from the energy
problems defined on them.

We consider a hierarchy of spaces $\II_\tau$, whose elements represent
the data kept explicit in each local energy problem. We call these
the retained data. Starting from the root space $\II_\varnothing$,
localization maps pass these data to successive nodes of the hierarchy;
the local energy problems on these spaces remain to be specified.
\Cref{fig:localization-intuition} shows a selected path within such
a hierarchy and the retained spaces along it.

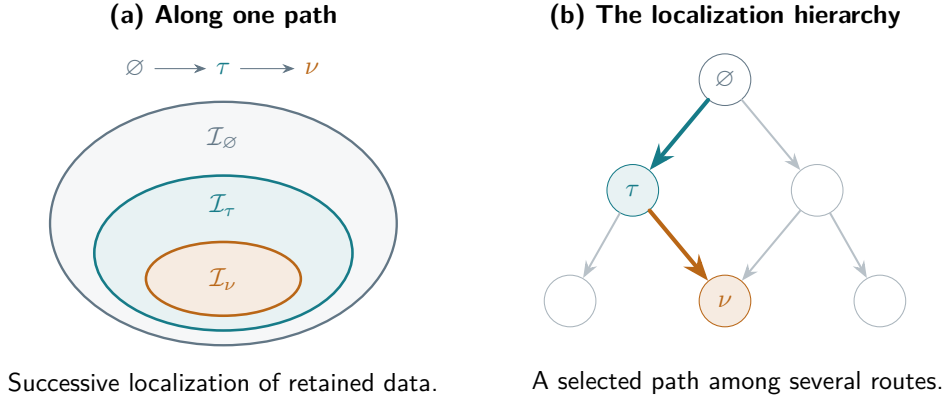
\begin{figure}[htbp]
\centering
\resizebox{.82\linewidth}{!}{%
\begin{tikzpicture}[x=1cm,y=1cm,font=\small,
  line cap=round,line join=round,>=Stealth]
\node[rolpanel] at (0,4.5) {(a) Along one path};
\node[text=linkgray] (intuition-root) at (-1.2,3.8) {$\varnothing$};
\node[text=linkteal] (intuition-tau) at (0,3.8) {$\tau$};
\node[text=defectorange] (intuition-nu) at (1.2,3.8) {$\nu$};
\draw[->,linkgray] (intuition-root)--(intuition-tau);
\draw[->,linkgray] (intuition-tau)--(intuition-nu);
\draw[draw=linkgray,fill=linkgray!6,line width=.9pt]
  (0,1.7) ellipse (2.35cm and 1.65cm);
\draw[draw=linkteal,fill=linkteal!10,line width=1pt]
  (0,1.3) ellipse (1.75cm and 1.05cm);
\draw[draw=defectorange,fill=defectorange!13,line width=1pt]
  (0,.95) ellipse (1.05cm and .5cm);
\node[text=linkgray] at (0,2.85) {$\mathcal I_\varnothing$};
\node[text=linkteal] at (0,1.95) {$\mathcal I_\tau$};
\node[text=defectorange] at (0,.95) {$\mathcal I_\nu$};
\node at (0,-.45) {Successive localization of retained data.};

\begin{scope}[shift={(6.8,0)}]
\node[rolpanel] at (0,4.5) {(b) The localization hierarchy};
\tikzset{hiernode/.style={circle,draw=linkgray!55,fill=white,
  minimum size=7mm,inner sep=2pt},
  hieredge/.style={->,draw=linkgray!45,line width=.8pt}}
\node[hiernode,text=linkgray,draw=linkgray] (hroot) at (0,3.65) {$\varnothing$};
\node[hiernode,draw=linkteal,text=linkteal,fill=linkteal!10]
  (htau) at (-1.25,2.15) {$\tau$};
\node[hiernode] (hother) at (1.25,2.15) {};
\node[hiernode] (hleafleft) at (-2.1,.65) {};
\node[hiernode,draw=defectorange,text=defectorange,fill=defectorange!13]
  (hnu) at (0,.65) {$\nu$};
\node[hiernode] (hleafright) at (2.1,.65) {};
\draw[hieredge] (hroot)--(hother);
\draw[hieredge] (htau)--(hleafleft);
\draw[hieredge] (hother)--(hnu);
\draw[hieredge] (hother)--(hleafright);
\draw[->,linkteal,line width=1.6pt] (hroot)--(htau);
\draw[->,defectorange,line width=1.6pt] (htau)--(hnu);
\node at (0,-.45) {\ \ \ A selected path among several routes.};
\end{scope}
\end{tikzpicture}%
}
\caption{Retained spaces along a localization path.
Left: along one selected path $\varnothing\to\tau\to\nu$ in the
hierarchy, retained data are successively localized in the spaces
$\mathcal I_\varnothing$, $\mathcal I_\tau$, and $\mathcal I_\nu$.
The nested regions illustrate the intuition of a subspace realization;
general ROL spaces need not be subspaces of a common
ambient space, nor is strict dimension decrease required at every step.
Right: the same path is highlighted within a branching hierarchy.
In a hierarchy indexed by faces of a simplicial complex, a child has
the form $\nu=\tau i:=\tau\cup\{i\}$, where $i$ is a vertex of the
link of $\tau$.
}
\label{fig:localization-intuition}
\end{figure}

Along a step $a:\tau\to\nu$ in \Cref{fig:localization-intuition},
retained data $b\in\II_\tau$ are mapped to $b'=R_ab\in\II_\nu$.
At $\nu$, we define the effective energy by holding $b'$ fixed and
minimizing the local energy over the remaining degrees of freedom.
This gives a lower bound on the energy of any full variable realizing $b'$.

We would like to repeat this construction across a hierarchy of local
problems. 
Minimization defines the effective energy at each stage, but does not by itself ensure that the resulting local problems fit together consistently across the hierarchy. 
The spaces, metrics, and
maps connecting them must be specified, and independently minimizing
local problems need not produce choices that arise from a common full
variable. Thus, the minimum energies alone may omit part of the cost of
assembling the local problems. We therefore ask what information must
accompany these energies to support compatible recursive localization.
A second, quantitative question is how estimates along this hierarchy
yield a global spectral-gap bound in the original metric, including
control of norm loss under localization.

\subsection{Our results}\label{sec:intro-our-results}
To address these questions, we develop \emph{Recursive Operator
Localization} (ROL), a framework for organizing effective energy problems
across a hierarchy and tracking the information needed to assemble them.
Our contributions are threefold.
\begin{enumerate}
\item \emph{Tracking the cost of reconciling local minimizers.}
Local minimum energies alone can miss the cost of making local choices
consistent. ROL isolates this cost as a nonnegative correction and
tracks it through the hierarchy, giving an exact energy decomposition
as well as a lower bound obtained by dropping the correction
(\cref{lem:rol-one-step,prop:rol-finite-depth}).

\item \emph{Geometric conditions for turning local gaps into a global gap.}
For nested subspaces, we separate local energy bounds from control of
the directions lost under localization. Link expansion propagates
this geometric control through the hierarchy, providing verifiable
conditions under which local gap bounds yield a global one
(\cref{thm:weighted-link-gap}).

\item \emph{Connecting classical methods to the operator framework.}
We identify variance decomposition and, for full-support binary laws,
conditional influence with the framework's energy and subspace
structures, respectively. These correspondences explain how classical
link-expansion and spectral-independence gap arguments fit into the
broader operator perspective
(\cref{prop:classical-variance-rol,prop:classical-si-gram}).

\end{enumerate}

Thus, our framework combines hierarchies of local energy problems with the geometry of local spectral expansion, extending classical local-to-global methods to an operator setting and providing a foundation for applications beyond scalar-valued functions.

\vspace{1em}

We now describe these results in more detail.

\paragraph{Recursive Operator Localization (ROL).}
The recursive structure consists of retained spaces $\II_\tau$
linked by localization maps $R_a:\II_\tau\to\II_\nu$ for steps
$a:\tau\to\nu$. Each space carries an effective energy $\Eff_\tau$;
ROL relates this energy to those at the next nodes, evaluated on
the localized data $R_ab$.

To construct these energies, each local problem has a full-variable
space $\HH_\tau$ and a readout $\Gamma_\tau:\HH_\tau\to\II_\tau$.
The effective energy $\Eff_\tau(b)$ is obtained by minimizing its
energy over full variables with readout $b$.
The maps $r_a:\HH_\tau\to\HH_\nu$ connect these local problems
consistently with the retained-data hierarchy: localizing a full
variable and then reading out its data must agree with applying $R_a$
to its original readout. This compatibility is expressed by
\begin{equation}\label{eq:intro-compatible-square}
\begin{tikzpicture}[baseline=(current bounding box.center),
  every node/.style={inner sep=3pt},
  compatmap/.style={-{Stealth[length=1.7mm]},line width=.6pt}]
\node (full-start) at (0,1.2) {$\HH_\tau$};
\node (full-end) at (3,1.2) {$\HH_\nu$};
\node (retained-start) at (0,0) {$\II_\tau$};
\node (retained-end) at (3,0) {$\II_\nu$};
\draw[compatmap] (full-start)--node[above] {$r_a$} (full-end);
\draw[compatmap] (retained-start)--node[below] {$R_a$} (retained-end);
\draw[compatmap] (full-start)--node[left] {$\Gamma_\tau$} (retained-start);
\draw[compatmap] (full-end)--node[right] {$\Gamma_\nu$} (retained-end);
\end{tikzpicture}
\end{equation}
The lower row is a step of the retained-data hierarchy; the upper row
and vertical readouts describe its realization by local energy problems.
The readouts themselves do not perform minimization.

Let $\omega_a>0$ be the weights used to assemble the full energy from the
localized energies. Under the compatibility and assembly assumptions formally stated in \cref{lem:rol-one-step}, we show that there exists a nonnegative
quadratic form $\Delta_\tau$ on $\II_\tau$ such that, for every
$b\in\II_\tau$,
\[
\Eff_\tau(b)
=\underbrace{\sum_{a:\tau\to\nu} \omega_a\Eff_\nu(R_ab)}
_{\text{energies after one-step localization}}
+\underbrace{\Delta_\tau(b)}_{\text{compatibility cost}},
\qquad \Delta_\tau(b)\geq0.
\]
The sum includes every specified step out of $\tau$.
The nonnegative quadratic form $\Delta_\tau$ is the gluing defect: the
extra energy required to use a common full variable instead of
minimizing the localized problems independently.

Write $t(p)$ for the endpoint of a path $p$. Along a path from the root,
the retained maps compose and the weights multiply:
\[
p=(\varnothing\xrightarrow{a_1}\cdots\xrightarrow{a_k}t(p)),
\qquad R_p=R_{a_k}\cdots R_{a_1},
\qquad \omega_p=\prod_{j=1}^k \omega_{a_j}.
\]
For the empty path, $R_p=\id$ and $\omega_p=1$.
Choose a finite stopping rule: at each visited path, either stop or
continue along every specified outgoing step. Let $\mathcal P_{\mathrm{stop}}$
be the stopped paths and $\mathcal P_{\mathrm{int}}$ their distinct proper
prefixes, including the empty path when the root is not a stopping point.

\begin{informaltheorem}[Compatible localization and finite-depth recursion, informal]
If the one-step hypotheses above hold at every nonterminal visit,
then, for every root input $b\in\II_\varnothing$,
\[
\Eff_\varnothing(b)
=\underbrace{\sum_{p\in\mathcal P_{\mathrm{stop}}}
 \omega_p\Eff_{t(p)}(R_pb)}_{\text{terminal energies}}
+\underbrace{\sum_{p\in\mathcal P_{\mathrm{int}}}
 \omega_p \Delta_{t(p)}(R_pb)}_{\text{accumulated defects}}.
\]
Thus the root effective energy is exactly the sum of the energies
remaining at the stopping points and the nonnegative defect costs
accumulated along the way. 
\end{informaltheorem}

See \cref{lem:rol-one-step,prop:rol-finite-depth} for the precise
statements. Distinct paths remain distinct even when they have the same
endpoint.

ROL is an abstract framework that admits different realizations through
the choice of spaces, localization maps, and energies.
\Cref{sec:rol-toolkit} illustrates two complementary realizations:
\begin{itemize}
\item \emph{Functions on links}: retaining facet components gives
weighted variance energies and a defect measuring disagreement of
local means; conditioning on site values is a scalar specialization
(\cref{sec:toolkit-links}).
\item \emph{Observable residuals}: enlarging a site set removes
newly representable observable components in a fixed metric, both
classically and for quantum observables (\cref{sec:toolkit-observables}).
\end{itemize}
The first makes the ROL decomposition explicit; the second introduces
the geometry used for descent below. Concrete applications should determine the particular
spaces, energies, and quantitative estimates.

The descent argument below discards the nonnegative defect.
\Cref{sec:gluing-response-cost} gives conditional one-step gap bounds
that use this term to control otherwise uncontrolled directions.
Incorporating these bounds into recursive descent remains a direction
for future work.

\paragraph{Quantitative descent through subspace geometry.}
The decomposition above alone is not a gap bound: local energies must also
control the norm in which the global gap is sought.
\label{par:global-problem}
Let $\AOp_\varnothing\succeq0$ be the operator to be estimated on a
root Hilbert space $\HH_\varnothing$ with a specified inner product,
and set $\mathcal N_\varnothing^{\mathrm{full}}=\ker\AOp_\varnothing$. The root label
$\varnothing$ denotes the starting problem. Our goal is to obtain
a quantitative lower bound $\gamma>0$ such that
\begin{equation}\label{eq:global-target}
\langle f,\AOp_\varnothing f\rangle_\varnothing
\geq\gamma\,\operatorname{dist}_\varnothing(f,\mathcal N_\varnothing^{\mathrm{full}})^2
\qquad(f\in\HH_\varnothing).
\end{equation}
A gap bound for the effective energy controls the retained data.
To obtain a gap bound on the full space, one must also control the
directions eliminated in defining that effective energy; see
\cref{sec:forms-metrics-gaps,sec:rol}.

We now specialize to retained spaces that are nested subspaces of a
common Hilbert space, with localization given by orthogonal projection.
In this setting, the geometry of the directions removed along different
branches determines how much of a vector's norm survives localization,
providing the control needed to turn local energy bounds into a global
gap estimate.

The descent mechanism couples subspace geometry with the spectra of auxiliary
classical link walks, then uses the resulting norm-retention bounds to transport energy.
\Cref{fig:intro-upward-descent-informal} illustrates the local geometry
and the upward propagation of bounds.

\begin{figure}[tbp]
\centering
\resizebox{.85\linewidth}{!}{%
\begin{tikzpicture}[x=1cm,y=1cm,>=Stealth,
  font=\small\sffamily,line cap=round,line join=round,
  concept/.style={rounded corners=5pt,align=center,
    minimum width=3.1cm,minimum height=.8cm}]
\node[rolpanel] at (-5.5,5.75) {(a) One parent--child pair};
\node[rolpanel] at (1,5.75) {(b) Upward propagation};
\begin{scope}[yshift=-.6cm]
\node[draw=linkgray,rounded corners=3pt,fill=linkgray!4,
  minimum width=1.2cm] (parent) at (-7.25,4.8) {Parent $\tau$};
\node[draw=linkgray,rounded corners=3pt,fill=linkgray!4,
  minimum width=1.2cm] (child) at (-3.75,4.8) {Child $\tau i$};
\draw[->,linkgray,thick] (parent)--node[above=5pt,font=\footnotesize\sffamily]
  {localization} (child);
\draw[linkgray!65,fill=linkgray!3] (-5.5,2.3) ellipse (2.25cm and 1.65cm);
\node[text=linkgray] at (-6.85,3.25) {$\II_\tau$};
\draw[defectorange,fill=defectorange!7,rotate around={22:(-5.35,1.9)}]
  (-5.35,1.9) ellipse (1.8cm and .55cm);
\coordinate (origin) at (-6.15,1.55);
\coordinate (original) at (-5.8,3.55);
\coordinate (removed) at (-4.7,2.14);
\draw[->,linkgray,very thick] (origin)--(original)
  node[above=-3pt] {$b$};
\draw[linkgray!65,dashed] (original)--(removed);
\draw[->,defectorange,very thick] (origin)--(removed)
  node[above=-3pt,xshift=10pt] {$E_{\tau,i}b$};
\fill[linkgray] (origin) circle (2pt);
\node[text=defectorange] at (-5.05,1.15) {$\mathcal U_{\tau,i}$};
\node[align=center,font=\footnotesize\sffamily] at (-5.5,-.3)
  {Component removed\\when passing to the child};
\end{scope}
\begin{scope}[shift={(.8,0)}]
\node[concept,fill=linkteal!12] (gap) at (0,4.8)
  {Global spectral gap};
\node[concept,fill=defectorange!10] (loss) at (0,2.5)
  {Control of average\\norm loss};
\draw[->,linkteal,very thick] (loss.north)--
  node[left,font=\footnotesize\sffamily,align=right] {energy\\transport} (gap.south);
\node[align=left,text=linkteal,font=\footnotesize\sffamily] (energy) at (3,3.65)
  {Local energy\\bounds};
\draw[->,linkteal,thick] (energy.west)--(0,3.65);
\foreach \x/\angle in {-2.2/15,0/30,2.2/45}{
  \begin{scope}[shift={(\x,0)}]
    \draw[linkgray!45,fill=linkgray!3] (0,0) ellipse (.85cm and .52cm);
    \draw[defectorange,fill=defectorange!7,rotate around={\angle:(0,-.12)}]
      (0,-.12) ellipse (.48cm and .13cm);
    \coordinate (small-origin) at ({-.3*cos(\angle)},{-.12-.3*sin(\angle)});
    \coordinate (small-removed) at ({.4*cos(\angle)},{-.12+.4*sin(\angle)});
    \coordinate (small-original) at (-.1,.39);
    \draw[linkgray!65,dashed] (small-original)--(small-removed);
    \draw[->,linkgray,thick] (small-origin)--(small-original);
    \draw[->,defectorange,thick] (small-origin)--(small-removed);
    \fill[linkgray] (small-origin) circle (1.2pt);
  \end{scope}
  \draw[->,defectorange!85,thick] (\x,.65)--(loss.south);
}
\node[fill=white,inner sep=4pt,text=linkgray,font=\footnotesize\sffamily]
  at (0,1.15) {via link expansion};
\node at (0,-.9) {Geometry across branches};
\node[text=linkteal,font=\footnotesize\sffamily] at (0,-1.55)
  {Terminal local gaps};
\node[text=linkteal] at (0,-2.05)
  {$\LOp_\sigma\succeq\delta_\sigma\id\quad
    (\sigma\in\mathscr T_{\mathrm{stop}})$};
\end{scope}
\end{tikzpicture}}
\caption{An overview of descent. Left: for one parent--child pair,
$E_{\tau,i}b$ is the component removed by localization, lying in
$\mathcal U_{\tau,i}\subseteq\II_\tau$.
Right: geometry across branches and spectral expansion of the classical
link walks (illustrated in \Cref{fig:link-up-down-walk}) control
average norm loss through the hierarchy. This allows local energy bounds
to be transported upward to a global spectral gap under the descent
assumptions. }
\label{fig:intro-upward-descent-informal}
\end{figure}
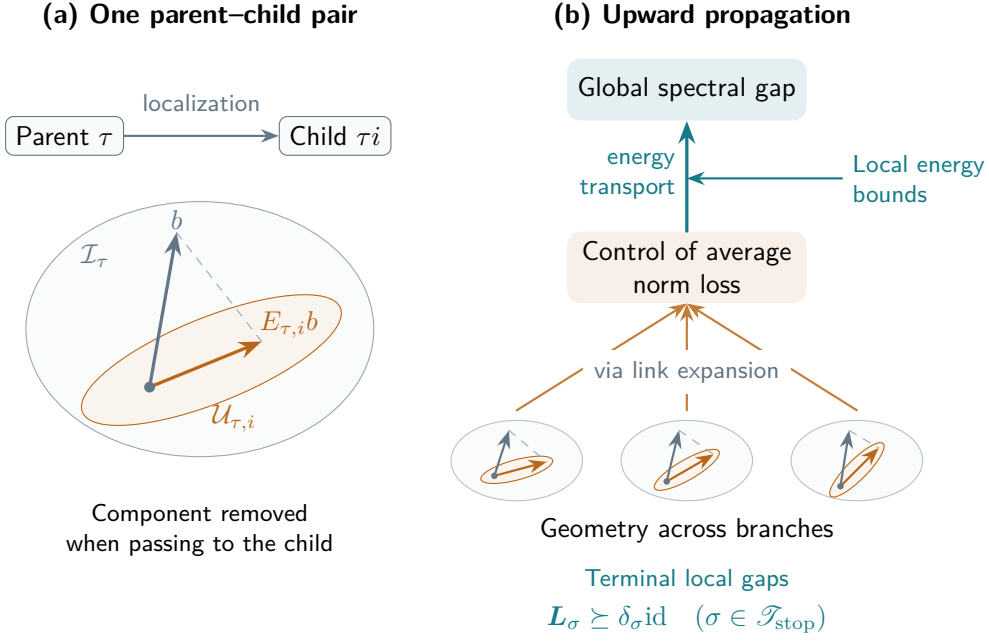

At each stopping node $\sigma$, assume a local energy bound
$\Eff_\sigma(b)\geq\delta_\sigma\|b\|_\sigma^2$ for all
$b\in\II_\sigma$, with $\delta_\sigma>0$; these are the terminal
energy bounds.
In the theorem below, $\overline{\theta}_\tau$ bounds the fraction of
squared norm lost on average at $\tau$. Starting from overlap bounds
at the deepest level, it is obtained recursively from the corresponding
bounds one level below and the link-walk gap at $\tau$; see
\cref{thm:weighted-link-gap} for the precise recursion.

\begin{informaltheorem}[Descent through subspace geometry, informal]
Under the assumptions of \cref{thm:weighted-link-gap}, overlap bounds
at the deepest level and classical link-walk gaps at intermediate levels
recursively determine norm-retention factors $1-\overline{\theta}_\tau$.
If these factors and the terminal energy bounds are positive, then,
on a nonzero root space,
\[
\operatorname{gap}(\LOp_\varnothing)
\geq\min_{p\in\mathcal P_{\mathrm{stop}}}
\left\{
\underbrace{\delta_{t(p)}}_{\substack{\text{terminal}\\\text{energy bound}}}
\qquad
\prod_{\substack{\tau\text{ on }p,\\\tau\text{ internal}}}
\quad
\underbrace{(1-\overline{\theta}_\tau)}_{\substack{\text{norm-retention}\\\text{factor}}}
\right\}.
\]
\end{informaltheorem}

\paragraph{Classical realizations and exact correspondences.}
\Cref{sec:classical-comparison} develops two connections.
First, specializing link localization to scalar-valued functions
identifies effective energies with variances and the gluing defect
with the variance of conditional means. The ROL decomposition thus
becomes the law of total variance.
The link-walk spectrum controls this defect, giving a variance-retention
estimate whose iteration recovers the classical down-up gap bound of
Alev and Lau~\cite[Theorem~1.5]{alev2020improved}
(\cref{prop:classical-variance-rol}, \eqref{eq:classical-down-up-gap}).
In this realization, link expansion controls variance retention directly,
without a separate geometric analysis of subspace overlap.

Second, spaces of functions with zero conditional mean provide a concrete
example of the subspace hierarchy. The associated energy at the root
is that of single-site heat-bath dynamics.
For full-support binary laws, the increment Gram matrix and conditional
influence matrix encode the same correlation data
(\cref{prop:classical-fibers,prop:classical-si-gram}).
Conditional influence therefore describes the overlap of the removed
subspaces. This identifies spectral independence as a geometric input
to subspace descent and recovers the heat-bath gap product of
Anari, Liu, and Oveis Gharan~\cite[Theorem~1.3]{anari2020spectral}
(\cref{cor:classical-si-gap}).

\subsection{Outlook}\label{sec:intro-outlook}
The framework separates the construction of local energy problems from
the geometric estimates that turn their assembly into spectral control.
This suggests several directions for further work.

\paragraph{Quantitative inputs and scaling.}
\label{par:outlook-inputs}
For families of models, the main task is to estimate the terminal
energy lower bounds, the removed-space overlaps, and the internal link-walk gaps
with controlled size dependence. In the complete uniform case,
\cref{cor:weighted-complete-gap} reads
\[
\operatorname{gap}(\LOp_\varnothing)
\geq\frac{\delta_n b_n}{n},\qquad
b_n:=1-(n-1)\epsilon_n>0,
\]
where $\delta_n$ is a common terminal energy lower bound and $\epsilon_n$
controls the deepest overlaps. For the rescaled operator
$n\LOp_\varnothing$, a constant, inverse-polynomial,
inverse-quasipolynomial, or inverse-exponential lower bound on
$\delta_n b_n$ gives the corresponding gap guarantee.
Finding useful estimates for these
inputs is a concrete way to turn the descent theorem into applications.
Returning to the combinatorial Laplacians underlying Garland's method,
a natural direction is to apply this framework to their spectral gaps
and the gap promises arising in computational topology.
The quantum complexity of clique homology~\cite{crichigno2024clique}
motivates studying spectral-gap promises in this setting, as considered
in weighted and unweighted gapped
variants~\cite{king2026gapped,hayakawa2026unweighted}.

\paragraph{Quantum Gibbs sampling.}
\label{par:outlook-gibbs}
Quantum Gibbs sampling is a natural setting in which to explore this
program: reversible sampler gaps govern relaxation to equilibrium,
while noncommuting observables do not admit a direct description by
simultaneous classical configurations. A candidate hierarchy uses local
observable algebras $\mathcal A_\tau=\mathcal B(\mathfrak h_\tau)
\otimes\mathbf1_{\tau^c}$ and their orthogonal complements
$\II_\tau=\mathcal A_\tau^\perp$ in one fixed KMS inner product
determined by a faithful Gibbs state $\rho$, where $\mathfrak h_\tau$
is the physical Hilbert space of the sites in $\tau$.
Here $\mathcal A_\tau$ consists of observables acting only on the sites
in $\tau$, while $\II_\tau$ contains the residual components orthogonal
to all such local observables in this metric.
The root is the space of mean-zero observables; enlarging the site set
removes more observable components and shrinks the residual space,
as in \cref{sec:toolkit-observables}.
Obtaining sampler-gap bounds requires verifying the subspace conditions,
bounding overlaps, and comparing the resulting energy with the sampler's
Dirichlet form.
Model-specific constructions and size-dependent or size-independent estimates are left for
future work.

\paragraph{Beyond the present descent mechanism.}
General ROL allows localization maps beyond nested orthogonal
projections. Extending quantitative descent to these settings would
broaden the range of admissible energy hierarchies.
Another direction is to use the gluing defect quantitatively:
\cref{sec:gluing-response-cost} gives criteria for controlling
directions missed by local energies, whereas the main descent uses
only its nonnegativity. 

\subsection{Organization}\label{sec:intro-organization}
\Cref{sec:preliminaries} fixes the Hilbert-space conventions and reviews
Gram operators, minimizing responses, and simplicial localization.
\Cref{sec:rol} introduces effective energies and develops recursive
decompositions and spectral comparisons.
\Cref{sec:rol-toolkit} illustrates ROL through link functions and
observable residuals.
\Cref{sec:weighted-descent} then establishes
quantitative descent on weighted link hierarchies and its complete
uniform specialization.
\Cref{sec:classical-comparison} relates these structures to classical
spin systems, HDX methods, and spectral independence.
Finally, \cref{sec:gluing-response-cost} develops an optional refinement
that uses the gluing defect to control directions missed by local energies.

\section{Preliminaries}\label{sec:preliminaries}
\ROLstatementboxesfalse
We fix the metric and spectral-gap conventions, collect standard formulas
for minimizing responses, then introduce simplicial
complexes and links. The concrete facet spaces, metrics, and restriction
maps are collected with the link example in \cref{sec:link-response-example}.

\subsection{Hilbert spaces, quadratic forms, and spectral gaps}
\label{sec:forms-metrics-gaps}
All Hilbert spaces in this paper are finite-dimensional over $\C$.
Each space
carries a specified positive definite inner product, and every adjoint
refers to the inner products on its domain and codomain. In coordinates,
a positive definite Hermitian matrix $W$ gives
\begin{equation}\label{eq:hilbert-metric}
\langle f,g\rangle_W=f^*Wg,\qquad
\|f\|_W^2=\langle f,f\rangle_W.
\end{equation}
We suppress the metric subscript when the underlying space is clear.

A Hermitian quadratic form $\Energy$ on a Hilbert space $\HH$ has a unique
self-adjoint representing operator $\AOp$ such that
\begin{equation}\label{eq:energy-representation}
\Energy(f)=\langle f,\AOp f\rangle\qquad(f\in\HH).
\end{equation}
The representing operator depends on the chosen inner product.
Energy forms are positive semidefinite, so $\AOp\succeq0$.
The abstract framework does not assume Markov or Dirichlet-form structure.

\begin{definition}[Spectral gap above the kernel]\label{def:spectral-gap}
Let $\AOp\succeq0$ be self-adjoint on $\HH$, and set
$\mathcal N=\ker\AOp$. If $\mathcal N^\perp\ne\{0\}$, define
\begin{equation}\label{eq:spectral-gap}
\operatorname{gap}(\AOp)
=\min_{0\ne f\in\mathcal N^\perp}
\frac{\langle f,\AOp f\rangle}{\|f\|^2}.
\end{equation}
\end{definition}
For $\gamma\geq0$, the bound $\operatorname{gap}(\AOp)\geq\gamma$
is equivalent to
\begin{equation}\label{eq:gap-coercivity}
\Energy(f)\geq\gamma\,\operatorname{dist}(f,\mathcal N)^2
\qquad\text{for every }f\in\HH,
\end{equation}
where distance and orthogonality use the specified inner product.
Indeed, the orthogonal decomposition $f=f_0+f_\perp$, with
$f_0\in\mathcal N$ and $f_\perp\in\mathcal N^\perp$, gives
$\Energy(f)=\Energy(f_\perp)$ and
$\operatorname{dist}(f,\mathcal N)=\|f_\perp\|$.
If $\mathcal N^\perp=\{0\}$, the inequality
\eqref{eq:gap-coercivity} is vacuous, and we do not assign a smallest
positive eigenvalue.

\paragraph{Notation.}
Bold symbols typically denote operators representing quadratic energies,
such as $\AOp$ and the effective operator $\LOp$, or associated Gram
operators such as $\bm C$. Maps between spaces, such as readouts
$\Gamma$, localization maps $R$, and response maps $h$, generally use
ordinary symbols. 

\subsection{Gram operators of subspaces}\label{sec:gram-operators}
For vectors $u_1,\ldots,u_m$, the Gram matrix records the inner products
$G_{ij}=\langle u_i,u_j\rangle$. The following construction replaces
individual vectors by subspaces.

\begin{definition}[Subspace Gram operator]\label{def:subspace-gram-operator}
Let $\mathcal U_1,\ldots,\mathcal U_m$ be subspaces of a
finite-dimensional Hilbert space $\HH$, and let
$J_i:\mathcal U_i\hookrightarrow\HH$ be their isometric inclusions.
On the orthogonal direct sum, define the summation map and Gram operator by
\begin{equation}\label{eq:gram-definition}
\mathsf S:\bigoplus_i\mathcal U_i\to\HH,\qquad
\mathsf S((u_i)_i)=\sum_iJ_i u_i,\qquad \bm C:=\mathsf S^*\mathsf S.
\end{equation}
\end{definition}
The direct sum keeps separate copies of the subspaces; their images
in $\HH$ need not be orthogonal. The blocks of $\bm C$ are
$C_{ij}=J_i^*J_j$. Thus its diagonal blocks are identities and its
off-diagonal blocks measure the overlap between different subspaces.

\begin{lemma}[Gram norm and projection sum]\label{lem:gram-norm}
For the data of \cref{def:subspace-gram-operator}, let
$E_i=J_iJ_i^*$ be the orthogonal projection onto $\mathcal U_i$.
Then
\begin{equation}\label{eq:gram-norm}
\bm C\succeq0,\qquad
\|\bm C\|=\|\mathsf S\|^2=\left\|\sum_iE_i\right\|,
\qquad
\left\|\sum_iJ_i u_i\right\|^2
\leq\|\bm C\|\sum_i\|u_i\|^2.
\end{equation}
\end{lemma}
\begin{proof}
First, $\bm C=\mathsf S^*\mathsf S$ is self-adjoint, and for every
$u=(u_i)_i\in\bigoplus_i\mathcal U_i$,
\[
\langle u,\bm C u\rangle=\langle \mathsf S u,\mathsf S u\rangle=\|\mathsf S u\|^2\geq0.
\]
Thus $\bm C\succeq0$.
Next, the adjoint of the summation map is
$\mathsf S^*x=(J_i^*x)_i$. Hence, for $x\in\HH$,
\[
\mathsf S\mathsf S^*x=\sum_iJ_iJ_i^*x=\sum_iE_i x.
\]
Using the Hilbert-space operator-norm identity
$\|\mathsf S^*\mathsf S\|=\|\mathsf S\|^2=\|\mathsf S\mathsf S^*\|$, we obtain
$\|\bm C\|=\|\mathsf S\|^2=\|\sum_iE_i\|$.
Finally, the orthogonal direct-sum norm satisfies
$\|u\|^2=\sum_i\|u_i\|^2$, so
\[
\left\|\sum_iJ_i u_i\right\|^2
=\|\mathsf S u\|^2\leq\|\mathsf S\|^2\|u\|^2
=\|\bm C\|\sum_i\|u_i\|^2.
\]
\end{proof}

\subsection{Minimizing responses and energy decomposition}
\label{sec:minimizing-responses}
We begin with the Schur-complement calculation: minimize a
block quadratic energy over one group of variables while retaining
the other. Besides the minimum energy, the calculation identifies
the minimizing response and the cost of departing from it.
These are the data used later to describe the gluing defect.

\paragraph{Schur complement and minimizing response.}
Let $\II$ and $\YY$ be finite-dimensional Hilbert spaces with specified
inner products. Write $b\in\II$ for the retained variables and
$y\in\YY$ for the internal variables. Consider a positive semidefinite
block operator
\[
\begin{pmatrix}\QOp&-B^*\\-B&\TOp\end{pmatrix}
\quad\text{on }\II\oplus\YY,
\]
with energy
\begin{equation}\label{eq:response-block-form}
\Energy(b,y)=\langle b,\QOp b\rangle_{\II}
-2\operatorname{Re}\langle y,Bb\rangle_{\YY}
+\langle y,\TOp y\rangle_{\YY}.
\end{equation}
Here $B:\II\to\YY$, all adjoints use the specified metrics, and
the minus sign fixes the convention for $B$.
The block $\QOp$ gives the energy at $y=0$; allowing $y$ to vary can
lower that energy.

\begin{lemma}[Response and energy decomposition]
\label{prop:minimizing-response}
Suppose $\TOp\succ0$ on $\YY$. Define the response map $h$ and
the Schur complement $\LOp$ by
\begin{equation}\label{eq:minimizing-response-operators}
h=\TOp^{-1}B,\qquad
\LOp=\QOp-B^*\TOp^{-1}B.
\end{equation}
Then, for every $(b,y)\in\II\oplus\YY$,
\begin{equation}\label{eq:response-energy-decomposition}
\Energy(b,y)
=\underbrace{\langle b,\LOp b\rangle_{\II}}_{\text{minimum energy at fixed }b}
+\underbrace{\langle y-hb,\TOp(y-hb)\rangle_{\YY}}_{\text{cost away from the minimizer}}.
\end{equation}
In particular,
$\min_y\Energy(b,y)=\langle b,\LOp b\rangle_{\II}$,
$\LOp\succeq0$, and the unique minimizing internal vector is
$y_{\min}=hb$.
\end{lemma}
\begin{proof}
Expand the second term in \eqref{eq:response-energy-decomposition}
and use $\TOp h=B$ and $h^*\TOp h=B^*\TOp^{-1}B$.
The resulting expression is \eqref{eq:response-block-form}.
Positive definiteness of $\TOp$ makes the second term vanish exactly
when $y=hb$. Minimizing gives the stated effective energy, which is
nonnegative because the full block operator is positive semidefinite.
\end{proof}

In this setting, the retained data are $b$, and eliminating the internal
variable $y$ gives the effective energy
\[
\Eff(b):=\min_{y\in\YY}\Energy(b,y)
=\langle b,\LOp b\rangle_{\II}.
\]
Thus $\LOp$ is the effective operator on the retained space.
If $\II\ne\{0\}$, minimizing also over unit retained data gives
\[
\min_{\|b\|_{\II}=1}\Eff(b)=\lambda_{\min}(\LOp).
\]
The triple $(\LOp,h,\TOp)$ records the minimum energy, the internal
response, and the cost of deviating from that response. Keeping only
$\LOp$ may lose information needed when internal variables of different
problems must later agree. In the link example of
\cref{sec:link-response-example}, the minimizing response is a weighted
mean and the effective energy is the variance.

\paragraph{Coordinates and metrics for a general readout.}
The same block calculation applies when retained variables are
specified by a surjective map rather than a given direct-sum splitting.
Let $\AOp\succeq0$ act on $\HH$, let
$\Gamma:\HH\to\II$ be a surjective linear map, and set
$\Energy(f)=\langle f,\AOp f\rangle$.
We first choose one representative of each retained value.
A linear right inverse $S:\II\to\HH$, satisfying $\Gamma S=\id_{\II}$,
makes this choice: $Sb$ has readout $b$.
All other vectors with the same readout differ from $Sb$ by a vector
in $\YY:=\ker\Gamma$. Thus
\[
\Gamma f=b
\quad\Longleftrightarrow\quad
f=Sb+y\quad\text{for a unique }y\in\YY.
\]
The variable $y$ describes precisely the freedom to change the full
vector without changing the retained data.

These coordinates turn the constrained problem into the block problem
above. Give $\YY$ the inner product inherited from $\HH$ and
$\II\oplus\YY$ its product inner product. The coordinate map
\begin{equation}\label{eq:response-coordinates}
J:\II\oplus\YY\to\HH,\qquad J(b,y)=Sb+y,
\end{equation}
is a linear isomorphism with $\Gamma J(b,y)=b$.
Writing the energy in these coordinates gives
\[
J^*\AOp J=\begin{pmatrix}\QOp&-B^*\\-B&\TOp\end{pmatrix},
\]
so $\Energy(J(b,y))$ has the form \eqref{eq:response-block-form}.
Minimizing at fixed readout is therefore exactly minimizing over $y$
while keeping $b$ fixed. Here $\QOp$ is the retained-coordinate
block, not the full operator $\AOp$ or, in general, the effective
operator $\LOp$.

For fixed retained data $b$, every $f\in\HH$ satisfying $\Gamma f=b$
has the form $f=Sb+y$, where $y\in\ker\Gamma$.
When $\TOp\succ0$, \cref{prop:minimizing-response} shows that the
energy is uniquely minimized at $y=hb$. Thus the unique
energy-minimizing vector with readout $b$ and its energy are
\[
f_{\min}(b)=Sb+hb\in\HH,\qquad
\Energy(f_{\min}(b))=\langle b,\LOp b\rangle_{\II}.
\]
Here $hb$ is the correction to the chosen representative $Sb$.
Changing $S$ changes these coordinates
and the response map, but not the constrained minimum energy
defined in \cref{def:retained-energy}.

Finally, this coordinate change preserves the energy description,
but need not preserve the norm. The two relevant norms are
\[
\|f\|_{\HH}=\|Sb+y\|_{\HH},
\qquad
\|(b,y)\|_{\II\oplus\YY}
=(\|b\|_{\II}^2+\|y\|_{\YY}^2)^{1/2},
\]
and they need not agree. Even choosing
$S\II=(\ker\Gamma)^\perp$ does not make $S$ isometric unless the
specified metric on $\II$ agrees with the one induced by $S$.
A full-space gap estimate must account for this distinction.

\paragraph{Singular costs and shorting.}
The assumption $\TOp\succ0$ says that the energy is positive on every
nonzero vector in $\ker\Gamma$, excluding zero-energy internal
directions. It is stronger than what is needed to define the effective
energy: singular costs are also covered by constrained minimization
(\cref{sec:constrained-minimization,lem:effective-energy-basic-facts}),
although the inverse formula above does not apply.
For an orthogonal splitting with inherited metrics, this variational
construction is the classical shorted operator, as detailed next.

\begin{remark}[Compression versus elimination]
\label{rem:compression-effective}
Let $\Pi$ be an orthogonal projection on $\HH$, with
$\II=\ran\Pi$ carrying the inherited metric and $\Gamma=\Pi$.
Compression fixes the complementary variable to zero, whereas
elimination minimizes over it. Thus
\begin{equation}\label{eq:compression-effective}
\QOp=\left.\Pi\AOp\Pi\right|_{\II},\qquad
\LOp=\left.\operatorname{Short}_{\Pi}(\AOp)\right|_{\II},
\qquad 0\preceq\LOp\preceq\QOp,
\end{equation}
where the classical shorted operator has quadratic form
$\min_{y\in\II^\perp}\Energy(b+y)$ on $\II$
\cite{anderson1975shorted,friedrich2018generalized}.
The companion paper~\cite{hayakawa2026double} uses an interface
compression estimate together with control of the loss due to coupling.
\end{remark}

\subsection{Simplicial complexes and links}
\label{sec:simplicial-links}
Simplicial links provide a geometric example of recursive localization.
We recall their definitions before specifying spaces and operators.

\begin{definition}[Faces, facets, and links]\label{def:simplicial-links}
A finite simplicial complex $X$ is a nonempty family of subsets of a
finite vertex set, closed under taking subsets. Its elements are faces,
including the empty face. A face $\tau$ has dimension $|\tau|-1$;
in particular, $\dim\varnothing=-1$. A facet is a face maximal under
inclusion. The complex is pure if all its facets have the same dimension.
For $\tau\in X$, its link is the simplicial complex
\begin{equation}\label{eq:simplicial-link}
\lk_X(\tau)=
\{\eta\in X:\eta\cap\tau=\varnothing,\ \eta\cup\tau\in X\}.
\end{equation}
\end{definition}
We do not assume purity unless stated explicitly. Thus facets may have
different dimensions, although pure complexes are included.
The link retains the faces that can extend $\tau$, with the selected
vertices removed. It is not merely the complex obtained by deleting
those vertices. In particular, $\lk_X(\varnothing)=X$.

We use $\tau,\tau'$ for faces at the same dimension and
$\tau\subseteq\nu\subseteq\xi$ for nested faces.
In the facet-function construction, $\sigma$ denotes a facet of the original complex.
For a vertex $i$ of $\lk_X(\tau)$, write
$\tau i=\tau\cup\{i\}$. The passage $\tau\to\tau i$ is a one-step
face extension, with $|\tau i|=|\tau|+1$.

Links compose by adjoining the selected faces:
\begin{equation}\label{eq:simplicial-link-composition}
\lk_{\lk_X(\tau)}(\eta)=\lk_X(\tau\cup\eta),
\qquad \eta\in\lk_X(\tau).
\end{equation}
Indeed, membership on either side means being disjoint from
$\tau\cup\eta$ and forming a face of $X$ upon adjoining $\tau\cup\eta$.
In particular, taking the link at $i$ within $\lk_X(\tau)$ gives
$\lk_X(\tau i)$. The indexing face grows while its link is further
localized.

\paragraph{Faces and paths.}
The depth of a face $\tau$ from $\varnothing$ is $|\tau|$, not its
dimension. A path is an ordered sequence of legal vertex additions.
Different paths, such as $(i,j)$ and $(j,i)$, can reach the same face
$\{i,j\}$, so the face hierarchy need not be a tree.
Equation~\eqref{eq:simplicial-link-composition} identifies their final
complexes, but does not specify weights for paths or maps between
associated Hilbert spaces.
\Cref{fig:rol-link-paths} illustrates two orders of taking links.
\begin{figure}[htbp]
\centering
\resizebox{.9\linewidth}{!}{%
\begin{tikzpicture}[x=1cm,y=1cm,
  linkvertex/.style={circle,fill=linkgray,inner sep=1.8pt},
  linkedge/.style={draw=linkgray,line width=.9pt},
  tealroute/.style={-{Stealth[length=2mm]},draw=linkteal,line width=1.1pt},
  orangeroute/.style={-{Stealth[length=2mm]},draw=defectorange,line width=1.1pt},
  sharedroute/.style={-{Stealth[length=2mm]},draw=linkgray,line width=1pt}
]
\coordinate (ti) at (-.9,-.6);
\coordinate (tj) at (.8,-.6);
\coordinate (tk) at (.4,.8);
\coordinate (tl) at (.12,-.08);
\fill[linkgray!9] (ti)--(tj)--(tk)--cycle;
\fill[linkteal!12] (ti)--(tl)--(tk)--cycle;
\fill[defectorange!10] (ti)--(tj)--(tl)--cycle;
\draw[linkedge] (ti)--(tj)--(tk)--cycle;
\draw[linkedge] (ti)--(tl)--(tk);
\draw[linkedge] (tl)--(tj);
\node[linkvertex,label=below left:$i$] at (ti) {};
\node[linkvertex,label=below right:$j$] at (tj) {};
\node[linkvertex,label=above:$k$] at (tk) {};
\node[linkvertex,label=right:$\ell$] at (tl) {};
\node[font=\small] at (0,1.6) {$X=2^{\{i,j,k,\ell\}}$};
\node[font=\scriptsize] at (0,-1.3) {Solid tetrahedron};

\fill[linkteal!10] (3.3,1.25)--(4.7,1.25)--(4,2.25)--cycle;
\draw[linkedge] (3.3,1.25)--(4.7,1.25)--(4,2.25)--cycle;
\node[linkvertex,label=below:$j$] at (3.3,1.25) {};
\node[linkvertex,label=below:$k$] at (4.7,1.25) {};
\node[linkvertex,label=above:$\ell$] at (4,2.25) {};
\node[font=\small] at (4,3.05) {$\lk_X(\{i\})$};

\fill[defectorange!10] (3.3,-2.2)--(4.7,-2.2)--(4,-1.2)--cycle;
\draw[linkedge] (3.3,-2.2)--(4.7,-2.2)--(4,-1.2)--cycle;
\node[linkvertex,label=below:$i$] at (3.3,-2.2) {};
\node[linkvertex,label=below:$k$] at (4.7,-2.2) {};
\node[linkvertex,label=above:$\ell$] at (4,-1.2) {};
\node[font=\small] at (4,-2.95) {$\lk_X(\{j\})$};

\draw[tealroute] (1.2,.4)--node[above,sloped,font=\footnotesize\sffamily] {link at $i$} (2.85,1.5);
\draw[orangeroute] (1.2,-.4)--node[below,sloped,font=\footnotesize\sffamily] {link at $j$} (2.85,-1.65);
\draw[tealroute] (5.2,1.6)--node[above,sloped,font=\footnotesize\sffamily] {link at $j$} (7,.35);
\draw[orangeroute] (5.2,-1.65)--node[below,sloped,font=\footnotesize\sffamily] {link at $i$} (7,-.35);

\draw[linkedge] (7.25,0)--(8.55,0);
\node[linkvertex,label=below:$k$] at (7.25,0) {};
\node[linkvertex,label=below:$\ell$] at (8.55,0) {};
\node[font=\small] at (7.9,.85) {$\lk_X(\{i,j\})$};
\draw[sharedroute] (8.95,0)--node[above,font=\footnotesize\sffamily] {link at $k$} (10.8,0);
\node[linkvertex,label=below:$\ell$] at (11.2,0) {};
\node[font=\small] at (11.2,.85) {$\lk_X(\{i,j,k\})$};
\end{tikzpicture}%
}
\caption{Ordered paths in a simplicial-link hierarchy.
For the full tetrahedron $X$ (including its interior), taking the link
at $i$ and then at $j$, or at $j$ and then at $i$, gives the same
full edge on $\{k,\ell\}$. Taking its link at $k$ leaves the vertex
$\ell$. Each triangle is filled. The two colored routes select
the orders $(i,j,k)$ and $(j,i,k)$. This diagram describes link
operations, not elimination of operator variables.}
\label{fig:rol-link-paths}
\end{figure}
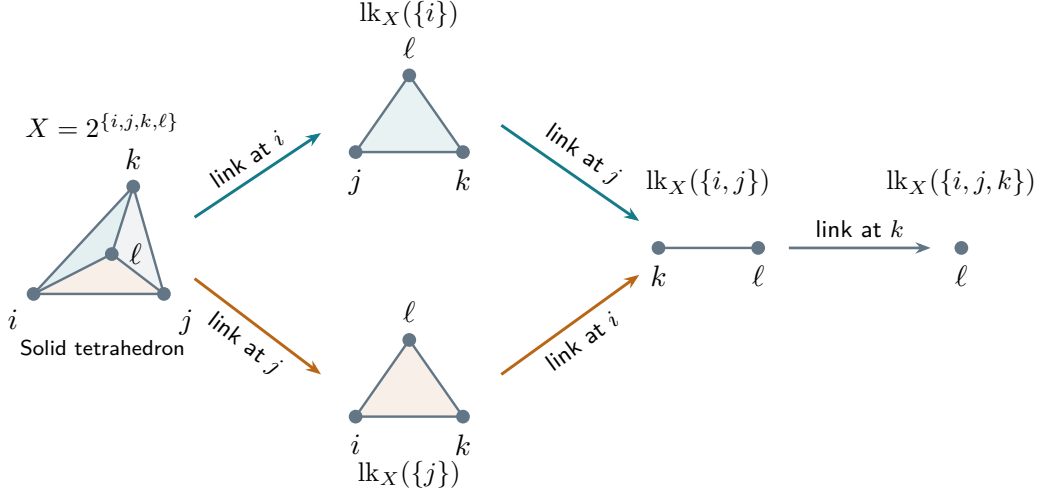

\subsection{Cochains and localization on links}
\label{sec:cochain-localization}
We recall the cochain convention used for comparison with classical
localization; see \cite[Sections~2--3]{oppenheim2017local}.
Write $\Sigma_k(X)$ for ordered $(k+1)$-tuples spanning faces.
A cochain $\phi\in C^k(X;\C)$ is an alternating function on
$\Sigma_k(X)$: permuting the vertices multiplies its value by the
sign of the permutation. We use the augmented convention
$C^{-1}(X;\C)=\C$, corresponding to the empty face.
Positive face weights $m$ give the inner product
\[
\langle\phi,\psi\rangle_m
=\frac1{(k+1)!}\sum_{\eta\in\Sigma_k(X)}
m(\eta)\overline{\phi(\eta)}\psi(\eta),
\]
where $m$ is independent of ordering. Thus each unoriented face is
counted once. No balance condition on these weights is imposed here.

\begin{definition}[Cochain localization]\label{def:cochain-localization}
Fix an ordered face $\tau=(v_0,\ldots,v_j)$ and $k\ge j$.
Concatenation defines
\[
\operatorname{loc}_\tau:C^k(X;\C)
\longrightarrow C^{k-j-1}(\lk_X(\tau);\C),
\qquad \phi_\tau(\eta)=\phi(\tau\eta).
\]
For the empty face, localization is the identity.
\end{definition}
This is the localization of \cite[Definition~3.3]{oppenheim2017local},
with complex coefficients. In contrast, degree-preserving restriction
uses $\phi^\tau(\eta)=\phi(\eta)$ for $\eta\in\Sigma_k(\lk_X(\tau))$;
see \cite[Definition~3.12]{oppenheim2017local}.
For a vertex $v$, localization of a $1$-cochain reads an edge value
$\phi(v,a)$ as a value at the link vertex $a$; restriction of a
$0$-cochain reads the original vertex value there.

The facet-component maps in \cref{sec:toolkit-links} select and
relabel coordinates. In the pure scalar case, with consistent
orientation signs, they correspond to cochain localization rather
than degree-preserving restriction.

\ROLstatementboxestrue
\section{Recursive Operator Localization}\label{sec:rol}
This section develops the structural identities of ROL. We first define
the effective energy of a single problem and organize such problems into
a localization hierarchy (\cref{sec:rol-global-effective,sec:rol-hierarchy}).
Compatible localization then decomposes the effective energy into local
effective energies and a nonnegative gluing defect, first in one step
and then through a finite hierarchy
(\cref{sec:two-level-localization,sec:finite-depth-localization}).
We next explain how these identities yield spectral-gap comparisons
(\cref{sec:rol-spectral-comparison}).

\subsection{Retained variables and effective energy}
\label{sec:rol-global-effective}
\label{sec:constrained-minimization}
We first work with a single energy problem, before introducing a
localization hierarchy. A positive semidefinite operator $\AOp$ on
a full Hilbert space $\HH$ defines the energy
$\Energy(f)=\langle f,\AOp f\rangle$. To obtain a problem on fewer
retained variables, we specify what information to keep and record
the least energy needed to realize it.

\begin{definition}[Retained variables and effective energy]
\label{def:retained-energy}
Let $\HH$ and $\II$ be finite-dimensional Hilbert spaces with specified
inner products, and let $\Energy(f)=\langle f,\AOp f\rangle$ with
$\AOp\succeq0$ on $\HH$. A surjective linear map
\begin{equation}\label{eq:retained-map}
\Gamma:\HH\to\II,\qquad \ket{f}\mapsto\ket{b}=\Gamma\ket{f},
\end{equation}
specifies the retained variables. We call $\Gamma$ the \emph{readout operator}.
For $b\in\II$, define the effective energy by
\begin{equation}\label{eq:retained-effective-energy}
\Eff(b):=\min_{f:\,\Gamma f=b}\Energy(f).
\end{equation}
Its representing operator in the specified inner product on $\II$ is
the effective operator $\LOp$, so that
\begin{equation}\label{eq:retained-effective-operator}
\Eff(b)=\langle b,\LOp b\rangle_{\II},\qquad \LOp\succeq0.
\end{equation}
\end{definition}

Surjectivity ensures that each constraint set $\Gamma^{-1}(b)$ is nonempty.
The map $\Gamma$ only reads the retained variables; it does not perform
the minimization. The retained space $\II$ need not be a subspace of
$\HH$, a tensor factor, or a physical boundary.

The minimum is attained and defines a positive semidefinite quadratic
form, even when the minimizing vector is not unique; see
\cref{lem:effective-energy-basic-facts} below.
In particular,
\begin{equation}\label{eq:energy-effective-lower-bound}
\Energy(f)\geq\Eff(\Gamma f).
\end{equation}
The effective energy records the exact minimum of the original energy for each fixed readout $b$; it does not, by itself, provide a spectral approximation to $\AOp$. A gap bound for $\LOp$ need not
by itself be a gap bound for $\AOp$ in the original metric;
the required comparison is discussed in
\cref{sec:rol-spectral-comparison}.

The following lemma verifies that the effective problem is well defined.

\begingroup
\ROLstatementboxesfalse
\begin{lemma}[Existence and structure of effective energies]
\label{lem:effective-energy-basic-facts}
For the data of \cref{def:retained-energy}, the minimum defining
$\Eff(b)$ is attained for every $b\in\II$, and $\Eff$ is a
positive semidefinite quadratic form with a unique representing
operator $\LOp$ in the specified metric.
\end{lemma}
\endgroup
\begin{proof}
Choose a linear right inverse
$S:\II\to\HH$ of $\Gamma$. Every feasible vector has the form $Sb+z$
with $z\in\ker\Gamma$. Let $\operatorname{Proj}_U$ be the orthogonal projection onto
$U=\AOp^{1/2}(\ker\Gamma)\subseteq\HH$. This subspace is closed in
finite dimensions, and hence
\begin{equation}\label{eq:effective-projection-formula}
\Eff(b)=\|(\id-\operatorname{Proj}_U)\AOp^{1/2}Sb\|_{\HH}^2.
\end{equation}
Indeed, $-\operatorname{Proj}_U\AOp^{1/2}Sb$ has a preimage in $\ker\Gamma$, which gives
a minimizing vector. The displayed expression is quadratic in $b$ and
therefore defines the unique operator in
\eqref{eq:retained-effective-operator}. Although $S$ was used to express
the minimum, its value is independent of that choice.

\end{proof}

We next place effective problems at the nodes of a hierarchy.
The minimization just defined takes place within each local problem;
maps between different problems and conditions relating their energies
are additional data, not consequences of minimization.

\subsection{The localization hierarchy}
\label{sec:rol-hierarchy}
\begin{conceptbox}
\begin{definition}[Effective problem at an index]\label{def:rol-effective-problem}
An effective problem at $\tau$ is a pair $(\II_\tau,\LOp_\tau)$:
a finite-dimensional retained Hilbert space with a specified inner
product, and a positive semidefinite self-adjoint operator
\[
\LOp_\tau:\II_\tau\longrightarrow\II_\tau.
\]
Its energy is
\[
\Eff_\tau(b)=\langle b,\LOp_\tau b\rangle_\tau,
\qquad b\in\II_\tau.
\]
\end{definition}
\end{conceptbox}
\medskip
The minimization in \cref{sec:rol-global-effective} constructs such
a pair. We now consider a family of these problems indexed by $\tau$.

\medskip
\begin{conceptbox}
\begin{definition}[Nodes and levels]\label{def:rol-nodes-levels}
A finite graded index set is a disjoint union
\[
\mathscr T=\bigsqcup_{k=0}^{h}\mathscr T_k,
\qquad h\in\mathbb Z_{\geq0},
\qquad \mathscr T_0=\{\varnothing\}.
\]
An element $\tau\in\mathscr T$ is a \emph{node}, indexing an
effective problem $(\II_\tau,\LOp_\tau)$.
Its depth is specified by
\[
k(\tau)=k\quad\Longleftrightarrow\quad\tau\in\mathscr T_k.
\]
The root $\varnothing$ is the unique node at depth zero.
\end{definition}
\end{conceptbox}
\medskip

\paragraph{Effective problems at a level.}
The problems at depth $k$ form the indexed family
\begin{equation}\label{eq:rol-level-family}
\bigl\{(\II_\tau,\LOp_\tau)\bigr\}_{\tau\in\mathscr T_k},
\qquad k=0,\ldots,h.
\end{equation}
Each member has energy
$\Eff_\tau(b_\tau)=\langle b_\tau,\LOp_\tau b_\tau\rangle_\tau$
for $b_\tau\in\II_\tau$.
Different nodes remain distinct indices even when their assigned
spaces and operators coincide.

\begin{definition}[Retained localization hierarchy]\label{def:rol-hierarchy}
Use the indexed families in \eqref{eq:rol-level-family}.
Equip this set with a rooted directed acyclic multigraph structure,
whose root is $\varnothing$. A specified step
$a:\tau\to\nu$ advances one level and carries a positive weight $\omega_a$
and a linear map $R_a:\II_\tau\to\II_\nu$. Parallel steps are
distinguished. The indices label local problems, not physical vertices.
\end{definition}

\paragraph{A path through the levels.}
A root path follows a sequence of nodes of increasing depth:
\begin{equation}\label{eq:rol-node-path}
p:\quad \tau_0=\varnothing
\xrightarrow{a_1}\tau_1\xrightarrow{a_2}\cdots
\xrightarrow{a_k}\tau_k,\qquad \tau_r\in\mathscr T_r.
\end{equation}
Each arrow is a specified localization step.
If parallel steps connect the same nodes, their choices are part of
the path data. The case $k=0$ is the empty path at the root.

Fix such a path. It selects one effective problem at each visited level:
\[
(\II_{\tau_0},\LOp_{\tau_0}),\
(\II_{\tau_1},\LOp_{\tau_1}),\ \ldots,\
(\II_{\tau_k},\LOp_{\tau_k}).
\]
\Cref{tab:rol-hierarchy} displays this sequence for $k\geq2$,
with $\tau_k$ chosen as its stopping point.

\begin{table}[htbp]
\centering
\small
\setlength{\tabcolsep}{4pt}
\renewcommand{\arraystretch}{1.18}
\begin{tabular}{@{}c c c p{.38\linewidth}@{}}
\toprule
Depth & Node & Effective problem & Step to the next problem\\
\midrule
$0$ & $\tau_0=\varnothing$ & $(\II_\varnothing,\LOp_\varnothing)$
& $R_{a_1}:\II_{\tau_0}\to\II_{\tau_1}$\\
$1$ & $\tau_1$ & $(\II_{\tau_1},\LOp_{\tau_1})$
& $R_{a_2}:\II_{\tau_1}\to\II_{\tau_2}$\\
$\vdots$ & $\vdots$ & $\vdots$ & $\vdots$\\
$k-1$ & $\tau_{k-1}$ & $(\II_{\tau_{k-1}},\LOp_{\tau_{k-1}})$
& $R_{a_k}:\II_{\tau_{k-1}}\to\II_{\tau_k}$\\
$k$ & $\tau_k$ & $(\II_{\tau_k},\LOp_{\tau_k})$ & Stop on this path\\
\bottomrule
\end{tabular}
\caption{Effective problems along one fixed root path $p$.
Each row selects one member of the family at that depth, not the entire
level. The maps act on vectors in the retained spaces.
Other paths may stop at different depths; see
\hyperref[par:rol-stopping]{below} for the stopping rule.}
\label{tab:rol-hierarchy}
\end{table}

The input to the root effective operator is a single vector
$b\in\II_\varnothing$; successive localization maps produce the vectors
along a path.

\begin{definition}[Path data]\label{def:rol-path-data}
For a root path $p=(a_1,\ldots,a_k)$ in \cref{def:rol-hierarchy}, let
$t(p)$ denote its endpoint and set
\begin{equation}\label{eq:rol-path-data}
\omega_p=\prod_{j=1}^k \omega_{a_j},\qquad
R_p=R_{a_k}\cdots R_{a_1}:\II_\varnothing\to\II_{t(p)}.
\end{equation}
The empty path has weight $1$, map $\id_{\II_\varnothing}$, and
endpoint $\varnothing$.
Write $|p|$ for the number of arrows in $p$ and
$\mathcal P_r:=\{p:\text{$p$ is a root path with $|p|=r$}\}$.
In particular, $\mathcal P_0$ consists of the empty path.
\end{definition}

\paragraph{Paths stopping at different depths.}
\phantomsection\label{par:rol-stopping}
Starting at the root, choose at each visited path either to stop or to
continue along every specified outgoing step. Stop whenever no further
step is possible. Thus every route from the root to a node with no
outgoing steps encounters exactly one stopping point, possibly the root
itself. The decision may depend on the path taken, not just its endpoint,
and different paths may stop at different depths.

Let $\mathcal P_{\mathrm{stop}}$ be the paths from the root to these
stopping points, and let $\mathcal P_{\mathrm{int}}$ be the distinct
paths visited before stopping. Equivalently, the latter are the
\emph{proper prefixes} of the stopped paths: initial segments obtained
by removing one or more arrows from the end, including the empty path.
Each such path is counted once, even if several stopped paths extend it.
Distinct arrow sequences remain distinct paths even when their endpoints
agree. If we stop at the root, $\mathcal P_{\mathrm{stop}}$ contains only
the empty path and $\mathcal P_{\mathrm{int}}$ is empty.

\paragraph{Nodes selected by the stopping rule.}
The corresponding sets of stopping nodes and internal nodes are
\begin{equation}\label{eq:rol-stopping-node-sets}
\mathscr T_{\mathrm{stop}}
:=\{t(p):p\in\mathcal P_{\mathrm{stop}}\},
\qquad
\mathscr T_{\mathrm{int}}
:=\{t(p):p\in\mathcal P_{\mathrm{int}}\}.
\end{equation}
These sets record the nodes where paths stop and where localization
continues. They need not be disjoint: one path may stop at a node that
another path passes through. Sums in the localization identities
remain indexed by paths, so distinct contributions with the same
endpoint are not lost.

\Cref{fig:rol-root-paths} illustrates this notation at depth three.
Only the displayed steps belong to the chosen hierarchy; it is not
the full face diagram of a simplex.

\begin{figure}[htbp]
\centering
\resizebox{0.9\linewidth}{!}{%
\begin{tikzpicture}[
  x=1cm,y=1cm,
  problem/.style={circle,draw=linkgray,fill=white,
    minimum size=7mm,inner sep=2pt,font=\small},
  stop/.style={rectangle,rounded corners=2pt,draw=linkgray,
    fill=linkgray!8,minimum height=7mm,minimum width=9mm,
    inner sep=3pt,font=\small},
  rolpath/.style={-{Stealth[length=2mm]},draw=linkgray,line width=.8pt},
  prolpath/.style={rolpath,draw=linkteal,line width=1.2pt},
  qrolpath/.style={rolpath,draw=defectorange,line width=1.2pt}
]
\node[rolpanel] at (3.4,.9) {(a) Localization indices};
\node[rolpanel] at (10.5,.9) {(b) Paths unfolded into a tree};
\node[font=\scriptsize,anchor=east] at (.4,0) {Depth 0};
\node[font=\scriptsize,anchor=east] at (.4,-1.25) {1};
\node[font=\scriptsize,anchor=east] at (.4,-2.5) {2};
\node[font=\scriptsize,anchor=east] at (.4,-3.75) {3};

\node[problem] (root) at (3.4,0) {$\varnothing$};
\node[problem] (i) at (2.25,-1.25) {$i$};
\node[problem] (j) at (4.55,-1.25) {$j$};
\node[stop] (ik) at (1.2,-2.5) {$ik$};
\node[problem] (ij) at (3.4,-2.5) {$ij$};
\node[stop] (jk) at (5.6,-2.5) {$jk$};
\node[stop] (ijk) at (3.4,-3.75) {$ijk$};
\draw[prolpath] (root)--(i);
\draw[qrolpath] (root)--(j);
\draw[rolpath] (i)--(ik);
\draw[prolpath] (i)--(ij);
\draw[qrolpath] (j)--(ij);
\draw[rolpath] (j)--(jk);
\draw[rolpath] (ij)--node[right,font=\scriptsize] {shared step} (ijk);

\node[problem] (troot) at (10.5,0) {$\varnothing$};
\node[problem] (ti) at (9.1,-1.25) {$i$};
\node[problem] (tj) at (11.9,-1.25) {$j$};
\node[stop] (tik) at (8.1,-2.5) {$ik$};
\node[problem,draw=linkteal] (tij1) at (9.8,-2.5) {$ij$};
\node[problem,draw=defectorange] (tij2) at (11.2,-2.5) {$ij$};
\node[stop] (tjk) at (12.9,-2.5) {$jk$};
\node[stop,draw=linkteal,fill=linkteal!7] (tijk1) at (9.8,-3.75) {$ijk$};
\node[stop,draw=defectorange,fill=defectorange!7] (tijk2) at (11.2,-3.75) {$ijk$};
\draw[prolpath] (troot)--(ti);
\draw[qrolpath] (troot)--(tj);
\draw[rolpath] (ti)--(tik);
\draw[prolpath] (ti)--(tij1);
\draw[qrolpath] (tj)--(tij2);
\draw[rolpath] (tj)--(tjk);
\draw[prolpath] (tij1)--(tijk1);
\draw[qrolpath] (tij2)--(tijk2);
\node[below=2pt,font=\small,text=linkteal] at (tijk1.south) {$p$};
\node[below=2pt,font=\small,text=defectorange] at (tijk2.south) {$q$};
\node[below=2pt,font=\small] at (tik.south) {$s$};
\node[below=2pt,font=\small] at (tjk.south) {$t$};
\node[font=\scriptsize,align=center] at (3.4,-4.8)
  {Circles: continue localization.\quad Boxes: stop.};
\node[font=\scriptsize,align=center] at (10.5,-4.8)
  {Separate copies of $ij$ and $ijk$ retain the path history.};
\end{tikzpicture}%
}
\[
\begin{aligned}
p=(\varnothing\to i\to ij\to ijk),\qquad
q=(\varnothing\to j\to ij\to ijk),\\
s=(\varnothing\to i\to ik),\qquad
t=(\varnothing\to j\to jk),\qquad
\mathcal P_{\mathrm{stop}}=\{p,q,s,t\}.
\end{aligned}
\]
\caption{Root paths and their prefixes in a chosen localization hierarchy.
On the left, paths $p$ and $q$ merge at the same index $ij$ and have
the same endpoint $t(p)=t(q)=ijk$. On the right, each tree vertex
represents an entire prefix, labelled only by its endpoint for brevity.
The internal prefixes are the empty path,
$(\varnothing\to i)$, $(\varnothing\to j)$,
$(\varnothing\to i\to ij)$, and $(\varnothing\to j\to ij)$.
A prefix shared by several stopped paths is counted once, but different
prefixes with the same endpoint remain distinct. Paths $s,t$ stop at
depth two, whereas $p,q$ stop at depth three.}
\label{fig:rol-root-paths}
\end{figure}

In particular, the common endpoint of $p$ and $q$ does not identify
their composite maps:
\[
R_p=R_{ijk\leftarrow ij}R_{ij\leftarrow i}R_{i\leftarrow\varnothing},
\qquad
R_q=R_{ijk\leftarrow ij}R_{ij\leftarrow j}R_{j\leftarrow\varnothing}.
\]
These maps need not agree. Their weights are likewise products along
their respective paths, so the two terminal contributions remain
separate in the decomposition of \cref{prop:rol-finite-depth}.
For the contrasting link construction, \Cref{fig:rol-link-paths}
shows two orders of taking links with the same endpoint and the same
facet-restriction map.

Concrete realizations of these nodes, spaces, and maps are given
in \cref{sec:rol-toolkit}.

\subsection{One-step localization and depth-two decomposition}
\label{sec:two-level-localization}
We first prove a one-step identity at an arbitrary starting node,
then decompose the root problem through depths one and two.
\Cref{fig:rol-one-two-depths} contrasts one-step assembly at an arbitrary
node with the depth-two decomposition of the root operator.

\begin{figure}[htbp]
\centering
\begin{tikzpicture}[
  x=1cm,y=1cm,
  rolvertex/.style={circle,draw=linkgray,fill=white,
    minimum size=7mm,inner sep=1pt,font=\small},
  roleaf/.style={rolvertex,draw=linkteal,fill=linkteal!7},
  roledge/.style={-{Stealth[length=1.7mm]},draw=linkgray,line width=.7pt},
  rolheading/.style={rolpanel,align=center},
  rollevel/.style={font=\scriptsize,text=linkgray,anchor=east}
]
\begin{scope}
\node[rolheading] at (0,0.85) {(a) One step};
\node[rollevel] at (-1.65,0) {$k$};
\node[rollevel] at (-1.65,-1.2) {$k+1$};
\node[rolvertex] (a0) at (0,0) {$\tau$};
\node[roleaf] (a1) at (-1,-1.2) {$\nu_1$};
\node[roleaf] (a2) at (1,-1.2) {$\nu_2$};
\draw[roledge] (a0)--(a1);
\draw[roledge] (a0)--(a2);
\node[font=\scriptsize,align=center,anchor=north] at (0,-2.85)
  {All steps from any starting node\\Depth $k\to k+1$};
\end{scope}
\begin{scope}[xshift=6.5cm]
\node[rolheading] at (0,0.85) {(b) Depth two};
\node[rollevel] at (-1.9,0) {$0$};
\node[rollevel] at (-1.9,-1.2) {$1$};
\node[rollevel] at (-1.9,-2.25) {$2$};
\node[rolvertex] (c0) at (0,0) {$\varnothing$};
\node[rolvertex] (c1) at (-1,-1.2) {$i_1$};
\node[rolvertex] (c2) at (1,-1.2) {$i_2$};
\node[roleaf] (c11) at (-1.5,-2.25) {};
\node[roleaf] (c12) at (-.5,-2.25) {};
\node[roleaf] (c21) at (.5,-2.25) {};
\node[roleaf] (c22) at (1.5,-2.25) {};
\draw[roledge] (c0)--(c1);
\draw[roledge] (c0)--(c2);
\draw[roledge] (c1)--(c11);
\draw[roledge] (c1)--(c12);
\draw[roledge] (c2)--(c21);
\draw[roledge] (c2)--(c22);
\node[font=\scriptsize,align=center,anchor=north] at (0,-2.85)
  {All length-two root paths\\Stop at depth $2$};
\end{scope}
\end{tikzpicture}
\caption{One-step localization and depth-two decomposition
(\cref{lem:rol-one-step,cor:rol-depth-two}).
Each node carries an effective problem. Teal nodes are the endpoints
where the displayed decomposition stops.
Panel (a) assembles all outgoing steps at any depth.
Panel (b) applies this operation at the root and at every depth-one node.
Two outgoing steps are shown only for illustration.
The tree layout separates paths for readability; their endpoints
need not be distinct nodes of the underlying hierarchy.}
\label{fig:rol-one-two-depths}
\end{figure}
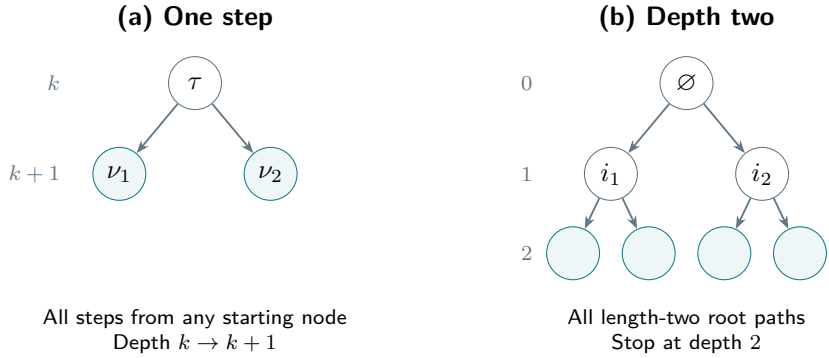

\paragraph{Full problems for one step.}
Fix a starting node $\tau$. At $\tau$ and each subsequent node $\nu$,
supply a finite-dimensional full Hilbert space, a nonnegative quadratic
energy, and a surjective linear readout to the retained space.
Writing $\eta$ for any of these indices, these data are
$\HH_\eta$, $\Energy_\eta$, and
$\Gamma_\eta:\HH_\eta\to\II_\eta$. Define their effective operators by
\begin{equation}\label{eq:rol-one-step-effective-data}
\langle b,\LOp_\eta b\rangle_\eta
=\Eff_\eta(b)
:=\min_{\Gamma_\eta f=b}\Energy_\eta(f).
\end{equation}
For each specified step $a:\tau\to\nu$, supply a linear map
$r_a:\HH_\tau\to\HH_\nu$ in addition to the retained map
$R_a:\II_\tau\to\II_\nu$ and positive weight $\omega_a$.
All adjoints below use the specified retained-space inner products.
Compatibility is expressed by the commutative square in
\eqref{eq:intro-compatible-square}: either route from $\HH_\tau$ to
$\II_\nu$ gives the same vector, as required in
\eqref{eq:rol-one-step-compatibility}.
The maps $r_a,R_a$ need not be surjective or isometric, and the weights
need not sum to one. 
The maps and energy assembly
are supplied as data, not inferred from the effective operators.

\medskip
\begin{lemma}[One-step localization]\label{lem:rol-one-step}
Consider the effective problems
\[
(\II_\tau,\LOp_\tau),\qquad
\bigl\{(\II_\nu,\LOp_\nu)\bigr\}_{a:\tau\to\nu},
\]
obtained from the full energies and readouts specified above.
Assume compatible readouts and exact energy assembly:
\begin{align}
\Gamma_\nu r_a&=R_a\Gamma_\tau
\qquad(a:\tau\to\nu),
\label{eq:rol-one-step-compatibility}\\
\Energy_\tau(f)
&=\sum_{a:\tau\to\nu}\omega_a\Energy_\nu(r_a f)
\qquad(f\in\HH_\tau).
\label{eq:rol-one-step-full-assembly}
\end{align}
Then there is a positive semidefinite operator
$\DOp_\tau$ on $\II_\tau$ such that, for every $b\in\II_\tau$,
the effective energy decomposes as
\begin{equation}\label{eq:rol-one-step-energy}
\Eff_\tau(b)=\langle b,\LOp_\tau b\rangle_\tau
=\sum_{a:\tau\to\nu}\omega_a\Eff_\nu(R_a b)
+\langle b,\DOp_\tau b\rangle_\tau.
\end{equation}
Consequently, the effective operator has the form
\begin{equation}\label{eq:rol-one-step-lemma}
\LOp_\tau
=\sum_{a:\tau\to\nu}\omega_aR_a^*\LOp_\nu R_a
+\DOp_\tau.
\end{equation}
The sum includes every specified step out of $\tau$.
The remainder $\DOp_\tau$ is the \emph{gluing defect}.
\end{lemma}
\medskip
\begin{proof}
Fix $b\in\II_\tau$. If $\Gamma_\tau f=b$, then
\eqref{eq:rol-one-step-compatibility} gives
$\Gamma_\nu r_a f=R_a b$. Applying
\eqref{eq:energy-effective-lower-bound} to each localized vector $r_a f$
and using exact energy assembly, we obtain
\[
\Energy_\tau(f)
=\sum_{a:\tau\to\nu}\omega_a\Energy_\nu(r_a f)
\geq\sum_{a:\tau\to\nu}\omega_a\Eff_\nu(R_a b).
\]
The right-hand side depends only on the fixed retained data $b$,
not on the choice of $f$ satisfying $\Gamma_\tau f=b$.
It therefore remains a lower bound when we minimize over all such $f$:
\[
\Eff_\tau(b)
=\min_{\Gamma_\tau f=b}\Energy_\tau(f)
\geq\sum_{a:\tau\to\nu}\omega_a\Eff_\nu(R_a b).
\]
Writing each effective energy in terms of its representing operator gives
\[
\langle b,\LOp_\tau b\rangle_\tau
\geq\sum_{a:\tau\to\nu}
\omega_a\langle R_a b,\LOp_\nu R_a b\rangle_\nu.
\]
Thus the self-adjoint operator
$\DOp_\tau:=\LOp_\tau-
\sum_{a:\tau\to\nu}\omega_aR_a^*\LOp_\nu R_a$
is positive semidefinite, proving \eqref{eq:rol-one-step-lemma}.
\end{proof}
We define the \emph{energy form of the gluing defect} by
\begin{equation}\label{eq:one-step-effective-assembly}
\Delta_\tau(b):=\langle b,\DOp_\tau b\rangle_\tau
=\Eff_\tau(b)-\sum_{a:\tau\to\nu}\omega_a\Eff_\nu(R_a b)
\geq0.
\end{equation}
It measures the excess of the minimum energy attainable from a common
full vector over the sum of independently minimized local energies.

\paragraph{Depth-one decomposition.}
Using the path data of \cref{def:rol-path-data}, collect the
depth-one contributions and the root defect:
\begin{equation}\label{eq:rol-depth-one-local-part}
\LOp_\varnothing^{[1]}
:=\sum_{p\in\mathcal P_1}\omega_pR_p^*\LOp_{t(p)}R_p,
\qquad \DOp_\varnothing^{[1]}:=\DOp_\varnothing.
\end{equation}
Both operators act on $\II_\varnothing$.
Bracketed superscripts record localization depth, not powers.

Under its hypotheses at the root, \cref{lem:rol-one-step} gives
\begin{equation}\label{eq:rol-root-assembly}
\LOp_\varnothing
=\LOp_\varnothing^{[1]}+\DOp_\varnothing^{[1]},
\qquad \DOp_\varnothing^{[1]}\succeq0.
\end{equation}
\paragraph{Depth-two decomposition.}
We now apply the one-step lemma twice, following all branches through
depth two. Each path visits three nodes, at depths $0,1,2$.
Collect the depth-two contributions and accumulated defect separately:
\begin{equation}\label{eq:rol-depth-two-local-part}
\LOp_\varnothing^{[2]}
:=\sum_{p\in\mathcal P_2}\omega_pR_p^*\LOp_{t(p)}R_p.
\end{equation}
\begin{equation}\label{eq:rol-depth-two-defect}
\DOp_\varnothing^{[2]}
:=\DOp_\varnothing^{[1]}
+\sum_{p\in\mathcal P_1}\omega_pR_p^*\DOp_{t(p)}R_p.
\end{equation}

\medskip
\begin{corollary}[Localization to depth two]\label{cor:rol-depth-two}
Assume the hypotheses of \cref{lem:rol-one-step} at the root and
at every depth-one problem. Then the effective operator satisfies
\begin{equation}\label{eq:rol-two-level-expansion}
\LOp_\varnothing=\LOp_\varnothing^{[2]}+\DOp_\varnothing^{[2]},
\qquad \DOp_\varnothing^{[2]}\succeq0.
\end{equation}
\end{corollary}
\medskip
\begin{proof}
Apply \cref{lem:rol-one-step} at each depth-one problem to obtain
\begin{equation}\label{eq:rol-next-assembly}
\LOp_\tau=\sum_{a:\tau\to\nu}\omega_aR_a^*\LOp_\nu R_a
+\DOp_\tau,\qquad \tau\in\mathscr T_1.
\end{equation}
Substitute these identities into \eqref{eq:rol-depth-one-local-part}
and then \eqref{eq:rol-root-assembly}. Use
\eqref{eq:rol-path-data} to combine the maps along each path.
The remaining terms give \eqref{eq:rol-depth-two-defect}.
Each is positive semidefinite, so their sum is positive semidefinite.
\end{proof}
Writing $b_p:=R_pb$, the energy identity displays the contributions
at all three levels:
\begin{equation}\label{eq:rol-two-level-energy}
\Eff_\varnothing(b)
={}\sum_{p\in\mathcal P_2}\omega_p
\langle b_p,\LOp_{t(p)}b_p\rangle_{t(p)}
+\sum_{p\in\mathcal P_1}\omega_p
\langle b_p,\DOp_{t(p)}b_p\rangle_{t(p)}
+\langle b,\DOp_\varnothing^{[1]}b\rangle_\varnothing.
\end{equation}
In the second sum, $\DOp_{t(p)}$ is the one-step defect at the
depth-one node $t(p)$; the final term is the defect at the root.
Together, these two contributions equal
$\langle b,\DOp_\varnothing^{[2]}b\rangle_\varnothing$, with
$\DOp_\varnothing^{[2]}$ defined in \eqref{eq:rol-depth-two-defect}.

The identities preserve the operator contributions without replacing
them by scalar gap bounds; the quantitative comparison with the root
norm and kernel is developed in \cref{sec:rol-spectral-comparison}.

\subsection{Finite-depth localization}
\label{sec:finite-depth-localization}
We first follow all root paths to a common depth $k$, using the same
notation as at depths one and two. We then allow different paths to
stop at different depths. Throughout, paths remain distinct even when
their endpoint nodes agree.

\paragraph{Decomposition to a common depth.}
For an integer $0\leq k\leq h$, collect the depth-$k$ contributions
and the defects accumulated before reaching that depth:
\begin{equation}\label{eq:rol-common-depth-local-part}
\LOp_\varnothing^{[k]}
:=\sum_{p\in\mathcal P_k}\omega_pR_p^*\LOp_{t(p)}R_p,
\end{equation}
\begin{equation}\label{eq:rol-common-depth-defect}
\DOp_\varnothing^{[k]}
:=\sum_{r=0}^{k-1}\sum_{p\in\mathcal P_r}
\omega_pR_p^*\DOp_{t(p)}R_p.
\end{equation}
Both act on $\II_\varnothing$. At depth zero,
$\LOp_\varnothing^{[0]}=\LOp_\varnothing$ and
$\DOp_\varnothing^{[0]}=0$.
For $k=1,2$, these definitions recover the preceding formulas.

\medskip
\begin{proposition}[Localization to a common depth]
\label{prop:rol-common-depth}
Fix $0\leq k\leq h$. Suppose every root path of length less than $k$
can be extended by at least one arrow. At each of its endpoint nodes
$\tau$, assume
\[
\LOp_\tau=\sum_{a:\tau\to\nu}\omega_aR_a^*\LOp_\nu R_a+\DOp_\tau,
\qquad \DOp_\tau\succeq0.
\]
Then, for every $b\in\II_\varnothing$, the effective energy is
\begin{equation}\label{eq:rol-common-depth-energy}
\Eff_\varnothing(b)=\langle b,\LOp_\varnothing b\rangle_\varnothing
=\sum_{p\in\mathcal P_k}\omega_p\Eff_{t(p)}(R_p b)
+\langle b,\DOp_\varnothing^{[k]}b\rangle_\varnothing.
\end{equation}
Thus the effective operator satisfies
\begin{equation}\label{eq:rol-common-depth-identity}
\LOp_\varnothing
=\LOp_\varnothing^{[k]}+\DOp_\varnothing^{[k]},
\qquad \DOp_\varnothing^{[k]}\succeq0.
\end{equation}
The two terms are defined in
\eqref{eq:rol-common-depth-local-part}--\eqref{eq:rol-common-depth-defect}.
\end{proposition}
\medskip
\begin{proof}
The identity holds at $k=0$. To pass from depth $r$ to depth $r+1$,
substitute the one-step identity at $t(p)$ into each term of
$\LOp_\varnothing^{[r]}$. Appending an arrow $a$ to $p$ gives a path
$pa$ with
\[
\omega_{pa}=\omega_p\omega_a,\qquad R_{pa}=R_aR_p.
\]
Every path in $\mathcal P_{r+1}$ has exactly one such decomposition,
so
\[
\LOp_\varnothing^{[r]}
=\LOp_\varnothing^{[r+1]}
+\sum_{p\in\mathcal P_r}\omega_pR_p^*\DOp_{t(p)}R_p.
\]
Adding this last sum to $\DOp_\varnothing^{[r]}$ gives
$\DOp_\varnothing^{[r+1]}$ by \eqref{eq:rol-common-depth-defect}.
Induction proves \eqref{eq:rol-common-depth-identity}.
Every accumulated term is positive semidefinite because $\omega_p>0$.
\end{proof}

\paragraph{Localization with a stopping rule.}
Use the \hyperref[par:rol-stopping]{stopping paths and internal prefixes}
introduced above, illustrated in \Cref{fig:rol-root-paths}.
The following identity allows paths to stop at different depths.

\begin{proposition}[Localization with a stopping rule]
\label{prop:rol-finite-depth}
Use the hierarchy and path data of
\cref{def:rol-hierarchy,def:rol-path-data}, together with the
\hyperref[par:rol-stopping]{stopping rule above}.
Suppose that every index $\tau$ reached before stopping, at which
localization continues, satisfies
\begin{equation}\label{eq:rol-node-recursion}
\LOp_\tau=\sum_{a:\tau\to\nu}\omega_aR_a^*\LOp_{\nu}R_a+\DOp_\tau,
\qquad \DOp_\tau\succeq0.
\end{equation}
Then, for every $b\in\II_\varnothing$, the effective energy is
\begin{equation}\label{eq:rol-finite-depth-energy}
\begin{aligned}
\Eff_\varnothing(b)
&=\langle b,\LOp_\varnothing b\rangle_\varnothing\\
&=\sum_{p\in\mathcal P_{\mathrm{stop}}}\omega_p\Eff_{t(p)}(R_p b)
+\sum_{p\in\mathcal P_{\mathrm{int}}}\omega_p
\langle R_p b,\DOp_{t(p)}R_p b\rangle_{t(p)}.
\end{aligned}
\end{equation}
Consequently, the effective operator has the form
\begin{equation}\label{eq:rol-finite-depth-identity}
\LOp_\varnothing=
\sum_{p\in\mathcal P_{\mathrm{stop}}}
\omega_pR_p^*\LOp_{t(p)}R_p
+\sum_{p\in\mathcal P_{\mathrm{int}}}
\omega_pR_p^*\DOp_{t(p)}R_p.
\end{equation}
Every term acts on the root retained space $\II_\varnothing$.
No path independence of the composite maps is required.
\end{proposition}
In \eqref{eq:rol-finite-depth-identity}, the first sum collects terminal
energies and the second the defects at proper prefixes.
Both sums retain distinct paths even when their endpoints agree.
\begin{proof}
Unfold the stopped paths into the finite tree whose nodes are their
prefixes. Start with the empty path as the only term awaiting substitution,
representing $\LOp_\varnothing$. For a nonstopped prefix $p$ ending
at $\tau$, substitute \eqref{eq:rol-node-recursion} into
$\omega_pR_p^*\LOp_{\tau}R_p$. Since $\omega_{pa}=\omega_p\omega_a$ and $R_{pa}=R_aR_p$,
this replaces that term by the terms for its one-edge extensions $pa$
and the additional term $\omega_pR_p^*\DOp_{\tau}R_p$.
All outgoing extensions occur in the prefix tree: otherwise a maximal
path through a missing extension would violate the stopping rule.
Each substitution preserves the identity. After finitely many such
steps, the remaining local-operator terms are exactly those indexed by
$\mathcal P_{\mathrm{stop}}$, and each proper prefix has contributed
its defect once. This gives \eqref{eq:rol-finite-depth-identity}.
Adjoints compose in reverse order in the specified metrics.
\end{proof}

In ROL, \cref{lem:rol-one-step}, applied wherever localization continues, supplies
\eqref{eq:rol-node-recursion} from compatible local problems.
The decomposition itself uses only that operator identity.
Stopping at the root gives a tautology. If every path reaches depth
two and stops there, the formula is \eqref{eq:rol-two-level-expansion}.
The terminal operators record what remains to be estimated directly,
while the second sum retains all intermediate compatibility costs.
Neither positivity of a terminal gap nor control of all root directions
follows solely from continuing the decomposition. Quantitative spectral
estimates require additional input beyond this identity.

\subsection{From effective energies to spectral gaps}
\label{sec:rol-spectral-comparison}
The finite-depth identity determines the root effective energy from
terminal energies and accumulated defects. To turn this identity into
a gap estimate, we must identify its zero-energy space and control the
norm away from that space. We record these two consequences here;
the additional subspace geometry in \cref{sec:weighted-descent} will
supply a concrete mechanism for the required quantitative control.

\subsubsection{Kernel identification}
\label{sec:rol-kernels}
The next identity relates zero-energy retained values to full-space
zero-energy vectors.
\begin{lemma}[Kernel under elimination]
\label{prop:effective-kernel}
For the data of \cref{def:retained-energy},
\begin{equation}\label{eq:effective-kernel}
\ker\LOp=\Gamma(\ker\AOp).
\end{equation}
\end{lemma}
\begin{proof}
If $b=\Gamma f$ with $f\in\ker\AOp$, then $\Eff(b)=0$.
Conversely, if $\LOp b=0$, choose an attained minimizing vector $f$
with $\Gamma f=b$. Its energy is zero, so $\AOp\succeq0$ implies
$f\in\ker\AOp$.
\end{proof}

Positivity also identifies the root kernel from the terms in the
decomposition, without providing a quantitative bound away from it.
\begin{lemma}[Recursive kernel identification]
\label{prop:rol-recursive-kernel}
Under the assumptions of \cref{prop:rol-finite-depth},
\begin{equation}\label{eq:rol-recursive-kernel}
\ker\LOp_\varnothing
=\bigcap_{p\in\mathcal P_{\mathrm{stop}}}
R_p^{-1}(\ker\LOp_{t(p)})
\ \cap\!\bigcap_{p\in\mathcal P_{\mathrm{int}}}
R_p^{-1}(\ker\DOp_{t(p)}).
\end{equation}
Here $R_p^{-1}$ denotes a subspace preimage, not an inverse map.
An intersection over no indices is the entire root retained space.
\end{lemma}
\begin{proof}
Every term of \eqref{eq:rol-finite-depth-identity} is PSD with a
positive path weight. Its quadratic form vanishes at $b$ exactly
when the corresponding $R_pb$ is in the kernel of the local operator.
The sum vanishes exactly when all these conditions hold.
\end{proof}

Together with \cref{prop:effective-kernel}, this reduces root kernel
identification to local zero-energy conditions when those are known.
Neither identity supplies a lower bound on positive eigenvalues,
and neither kernel formula is asserted for indefinite shifted forms.

\subsubsection{Spectral comparison from local lower bounds}
To use the decomposition quantitatively, we substitute available local
lower bounds for the exact effective operators. For example, a local
gap estimate provides the lower operator
\[
\LowerLOp_\tau
=\delta_\tau\operatorname{Proj}_{(\ker\LOp_\tau)^\perp}
\preceq\LOp_\tau.
\]
The underline denotes a supplied lower bound, not a new effective
problem. We likewise write $\LowerDOp_\tau$ for a lower bound on the
defect. The next lemma collects these bounds at the root and
states the norm control needed for a gap.

Use the hierarchy and stopping paths of \cref{prop:rol-finite-depth}.
Using the node sets in \eqref{eq:rol-stopping-node-sets}, choose PSD
operators on the corresponding retained spaces:
\[
\LowerLOp_\tau:\II_\tau\to\II_\tau
\quad(\tau\in\mathscr T_{\mathrm{stop}}),\qquad
\LowerDOp_\tau:\II_\tau\to\II_\tau
\quad(\tau\in\mathscr T_{\mathrm{int}}).
\]
They serve as lower bounds for terminal energies and one-step
defects, respectively; no recursion among these comparison operators
is required. Collect their weighted pullbacks on the root
space by defining
\begin{equation}\label{eq:rol-root-comparison}
\LowerLOp_\varnothing:=
\sum_{p\in\mathcal P_{\mathrm{stop}}}\omega_pR_p^*
\LowerLOp_{t(p)}R_p+
\sum_{p\in\mathcal P_{\mathrm{int}}}\omega_pR_p^*
\LowerDOp_{t(p)}R_p.
\end{equation}
Thus $\LowerLOp_\varnothing$ collects the proposed local lower bounds in
the root metric.

Put $\mathcal N_\varnothing:=\ker\LOp_\varnothing$ and let
$\operatorname{Proj}_{\mathcal N_\varnothing^\perp}$ denote the orthogonal projection
onto its orthogonal complement in $\II_\varnothing$.
To obtain a root effective gap, the assembled comparison must control
every direction in this complement.

\begin{lemma}[Comparison from localized lower bounds]
\label{prop:rol-spectral-comparison}
Assume the hypotheses of \cref{prop:rol-finite-depth} and the local bounds
\begin{equation}\label{eq:rol-comparison-inputs}
\LOp_{t(p)}\succeq\LowerLOp_{t(p)}
\quad(p\in\mathcal P_{\mathrm{stop}}),\qquad
\DOp_{t(p)}\succeq\LowerDOp_{t(p)}
\quad(p\in\mathcal P_{\mathrm{int}}).
\end{equation}
Then $\LOp_\varnothing\succeq\LowerLOp_\varnothing\succeq0$.
If, in addition, for some $\gamma>0$,
\begin{equation}\label{eq:rol-root-coverage}
\LowerLOp_\varnothing\succeq\gamma\operatorname{Proj}_{\mathcal N_\varnothing^\perp},
\end{equation}
then
\begin{equation}\label{eq:rol-effective-gap}
\langle b,\LOp_\varnothing b\rangle_{\II_\varnothing}
\geq\gamma\operatorname{dist}_{\II_\varnothing}
(b,\mathcal N_\varnothing)^2\qquad(b\in\II_\varnothing).
\end{equation}
If $\mathcal N_\varnothing^\perp\ne\{0\}$, this gives
$\operatorname{gap}(\LOp_\varnothing)\geq\gamma$.
\end{lemma}
\begin{proof}
Subtract \eqref{eq:rol-root-comparison} from
\eqref{eq:rol-finite-depth-identity}. Each difference is a positive
weighted pullback of a PSD operator by \eqref{eq:rol-comparison-inputs}.
Thus $\LOp_\varnothing\succeq\LowerLOp_\varnothing\succeq0$.
Combining this with \eqref{eq:rol-root-coverage} and evaluating at $b$
proves \eqref{eq:rol-effective-gap}.
\end{proof}

The choice $\LowerDOp_\tau=0$ is allowed, but defects can also
supply directions not controlled by the terminal terms.
Positive terminal coefficients alone do not establish
\eqref{eq:rol-root-coverage}. This is the substantive quantitative input:
it must follow from the local data and an independently identified
kernel, not from assuming the desired root gap.

The local lower bounds may be assembled as operators before being
replaced by scalar gap estimates, preserving their directional information.

\subsubsection{From the retained gap to the full-space gap}
\label{par:rol-root-full-gap}
The preceding comparison bounds the gap of the root retained operator
$\LOp_\varnothing$. When does this also bound the gap of the original
operator $\AOp_\varnothing$ on $\HH_\varnothing$, as sought in
\eqref{eq:global-target}?
Elimination at the root is optional, and two choices of retained space
and readout make the two gaps coincide:
\begin{itemize}
\item Take $\II_\varnothing=\HH_\varnothing$ with the same metric
and $\Gamma_\varnothing=\id$. Then
$\LOp_\varnothing=\AOp_\varnothing$.
\item If $\mathcal N=\ker\AOp_\varnothing$ is known, take
$\II_\varnothing=\mathcal N^\perp$ with its inherited metric and
$\Gamma_\varnothing=\operatorname{Proj}_{\mathcal N^\perp}$ as a map onto this space.
This removes only zero-energy directions. Every lift of $b$ is
$b+z$ with $z\in\mathcal N$, and its energy is
$\langle b,\AOp_\varnothing b\rangle$. Hence
$\LOp_\varnothing=\AOp_\varnothing|_{\mathcal N^\perp}$.
\end{itemize}
In either case, \eqref{eq:rol-effective-gap} therefore yields
\[
\operatorname{gap}(\AOp_\varnothing)
=\operatorname{gap}(\LOp_\varnothing)\geq\gamma,
\]
provided the positive spectrum is nonempty. Once the localization and
spectral-comparison hypotheses are verified, no additional full-space
comparison is needed.

For a general readout, directions carrying nonzero energy may also be
eliminated. A retained gap alone then need not control the full-space
gap: one must also estimate the energy cost of the eliminated directions
and compare the resulting norms with the original metric.
The response decomposition in \cref{prop:minimizing-response} provides
a starting point for this additional comparison.

The subspace descent below takes the defect lower bounds to be zero
and uses only the nonnegativity of the gluing defect.
An optional refinement exploiting the defect itself is developed in
\cref{sec:gluing-response-cost}; it is not needed for the subspace
descent.

\section{Concrete realizations of ROL}
\label{sec:rol-toolkit}
Two constructions connect the abstract framework to the later results.
Functions on links make effective energies and gluing defects explicit,
leading to the classical variance decomposition and down-up gap bounds
(\cref{sec:classical-dictionaries}).
Observable residuals provide the common-space geometry used in subspace
descent (\cref{sec:weighted-descent}) and connect it to classical spin
systems and spectral independence (\cref{sec:classical-si}).

\subsection{Functions on links: restriction and effective variance}
\label{sec:toolkit-links}
\label{sec:link-response-example}
Here a node $\tau$ is a face of a finite simplicial complex $X$.
It selects the facets containing that face, rather than necessarily
specifying values of physical variables. A retained vector assigns a
coefficient to each selected facet. Adding a vertex to the face keeps
only the coefficients of facets containing the enlarged face.

\begin{figure}[htbp]
\centering
\begin{tikzpicture}[scale=.95,
  edge/.style={draw=linkgray,line width=.7pt},
  selected/.style={draw=defectorange,line width=1.5pt},
  every node/.style={font=\small}]
\foreach \s/\heading in {0/{\varnothing},4.7/{\{v\}},9.4/{\{v,a\}}}{
  \begin{scope}[xshift=\s cm]
    \coordinate (v) at (0,0);
    \coordinate (a) at (-.15,1.15);
    \coordinate (b) at (1.25,.05);
    \coordinate (c) at (.1,-1.05);
    \coordinate (d) at (-1.2,-.12);
    \fill[linkteal!12] (v)--(a)--(b)--cycle;
    \fill[linkteal!12] (v)--(d)--(a)--cycle;
    \ifdim\s pt<9pt
      \fill[linkteal!12] (v)--(b)--(c)--cycle;
      \fill[linkteal!12] (v)--(c)--(d)--cycle;
    \else
      \fill[linkgray!6] (v)--(b)--(c)--cycle;
      \fill[linkgray!6] (v)--(c)--(d)--cycle;
    \fi
    \draw[edge] (a)--(b)--(c)--(d)--cycle;
    \foreach \q in {a,b,c,d}{\draw[edge] (v)--(\q);}
    \ifdim\s pt>9pt \draw[selected] (v)--(a);\fi
    \ifdim\s pt>0pt \fill[defectorange] (v) circle (2pt);\fi
    \node[above] at (a) {$a$};
    \node[right] at (b) {$b$};
    \node[below] at (c) {$c$};
    \node[left] at (d) {$d$};
    \node[below right] at (v) {$v$};
    \node at (.42,.43) {$1$};
    \node at (.48,-.4) {$2$};
    \node at (-.4,-.43) {$3$};
    \node at (-.48,.4) {$4$};
    \node at (0,1.95) {$\tau=\heading$};
  \end{scope}
}
\draw[-{Stealth},edge] (1.65,0)--(3.05,0)
  node[midway,above] {add $v$};
\draw[-{Stealth},edge] (6.35,0)--(7.75,0)
  node[midway,above] {add $a$};
\node at (0,-1.95) {$(f_1,f_2,f_3,f_4)$};
\node at (4.7,-1.95) {$(f_1,f_2,f_3,f_4)$};
\node at (9.4,-1.95) {$(f_1,f_4)$};
\end{tikzpicture}
\caption{Facet selection along one path. The facets are
$\sigma_1=vab$, $\sigma_2=vbc$, $\sigma_3=vcd$, and $\sigma_4=vda$;
shading highlights those containing the selected face.
All facets contain $v$, so the first restriction keeps every component.
The second keeps only facets $1$ and $4$. The pictures show selected
facets in the original complex, not drawings of the links themselves.
Each $f_j$ may be a vector, and its original facet label is preserved.}
\label{fig:toolkit-link-restriction}
\end{figure}
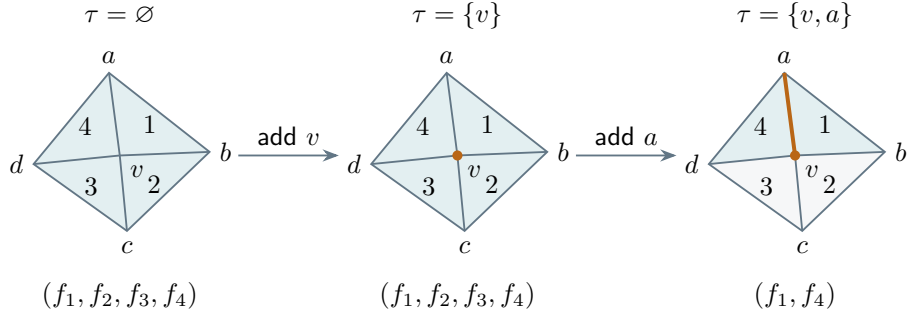

\paragraph{Spaces, metrics, and maps.}
\label{sec:link-spaces}
No purity assumption for the underlying simplicial complex is needed for the following construction.
Write $X_\tau=\lk_X(\tau)$ and set
\begin{equation}\label{eq:facet-index}
\Omega_\tau=\{\sigma:\sigma\text{ is a facet of }X,\ \tau\subseteq\sigma\}.
\end{equation}
The correspondence $\sigma\mapsto\sigma\setminus\tau$ identifies $\Omega_\tau$
with the facets of $\lk_X(\tau)$. We keep the original label $\sigma$ to
track the same components and weights across localization levels.

Fix a coefficient space $\C^q$, with $q\geq1$. Each facet carries a vector
$f_\sigma\in\C^q$; $q$ is the coefficient dimension, not the simplex
dimension. A variable in $\II_\tau$ is the collection
$f=(f_\sigma)_{\sigma\in\Omega_\tau}$. Thus $\II_\varnothing$ has one vector
for every facet, $\II_i$ retains facets containing $i$, and $\II_{ij}$ retains
facets containing both $i$ and $j$. This gives a concrete realization of
the retained-space hierarchy illustrated in \cref{tab:rol-hierarchy}.

Assign a positive definite Hermitian matrix $W_\sigma$ to each facet.
The matrices need not commute. Define
\begin{equation}\label{eq:link-space}
\II_\tau=\bigoplus_{\sigma\in\Omega_\tau}\C^q,
\qquad
\langle f,g\rangle_\tau=\sum_{\sigma\in\Omega_\tau}f_\sigma^*W_\sigma g_\sigma.
\end{equation}
Concretely, $\II_\tau=\operatorname{span}\{|\sigma,i\rangle:
\sigma\in\Omega_\tau,\ 1\leq i\leq q\}$, where $|\sigma,i\rangle$
represents the $i$-th coordinate at facet $\sigma$, with inner products
determined by the weights $W_\sigma$.
We use $\II_\tau$ because the facet values will be the variables retained
after an auxiliary mean is eliminated.

For $\tau\subseteq\nu$, the restriction
$R_{\nu\leftarrow\tau}:\II_\tau\to\II_{\nu}$ discards the components
whose facets do not contain $\nu$. Explicitly,
\begin{equation}\label{eq:link-restriction}
(R_{\nu\leftarrow\tau}f)_\sigma=f_\sigma\qquad(\sigma\in\Omega_{\nu}).
\end{equation}
Equivalently, the retained components define a function on the facets of
$X_\nu$, with value $f_{\nu\cup\eta}$ at $\eta$.
This is a concrete choice of retained-variable localization map of the
type introduced in \cref{lem:rol-one-step}.
The remaining weights are unchanged, so the adjoint is extension by zero.
Consequently,
\begin{equation}\label{eq:path-consistency}
R_{\xi\leftarrow\nu}R_{\nu\leftarrow\tau}
=R_{\xi\leftarrow\tau}
\qquad(\tau\subseteq\nu\subseteq\xi).
\end{equation}
Both the metric and the restriction at the final face are independent of the
path used to reach it.

\begin{remark}[Scalar cochain interpretation]
For $q=1$ on a pure complex with $d$ vertices per facet, a vertex
ordering identifies facet functions with $C^{d-1}(X;\C)$.
Relabeling a retained facet by $\sigma\setminus\tau$ and accounting
for orientation signs identifies the root-to-$\tau$ component map
with cochain localization to degree $d-1-|\tau|$
(\cref{def:cochain-localization}), not degree-preserving restriction.
This identifies the maps, not the variance energy with a cochain Laplacian.
\end{remark}

\paragraph{Weighted variance as an effective energy.}
Effective energy minimizes over the variables not retained by the
readout. Here we retain the facet values $f$ and minimize over a common
representative $\alpha$. The resulting energy measures the discrepancy
that remains after choosing the best representative.
For scalar values and positive scalar weights $w_\sigma$, this is
precisely the familiar least-squares characterization of variance:
\[
\min_{\alpha\in\C}\sum_{\sigma\in\Omega_\tau}
w_\sigma|f_\sigma-\alpha|^2
=W_\tau\operatorname{Var}_{\mu_\tau}(f),
\qquad
W_\tau=\sum_{\sigma\in\Omega_\tau}w_\sigma,\quad
\mu_\tau(\sigma)=\frac{w_\sigma}{W_\tau}.
\]
The best representative is $\mathbb E_{\mu_\tau}f$: geometrically,
the mean is the orthogonal projection onto constant functions, and
variance is the squared norm of the residual. This is the constant
case of the standard least-squares interpretation of conditional
expectation~\cite[Theorem~4.1.15]{durrett2019probability}.

The same minimization principle applies to vector-valued facet
functions and positive definite matrix weights. With a common
representative $\alpha\in\C^q$, define
\begin{equation}\label{eq:link-full-problem}
\HH_\tau=\II_\tau\oplus\C^q,\qquad \Gamma_\tau(f,\alpha)=f,\qquad
\Energy_\tau(f,\alpha)=\sum_{\sigma\in\Omega_\tau}
(f_\sigma-\alpha)^*W_\sigma(f_\sigma-\alpha).
\end{equation}
Use the facet metric on $\II_\tau$ and the Euclidean metric on $\C^q$.
The retained variable $b$ of the general construction is here the entire
function $f$, not the coefficient $\alpha$.

No purity assumption is needed for this minimization.
Put $W_\tau=\sum_{\sigma\in\Omega_\tau}W_\sigma$ and define the
weighted mean and variance by
\begin{align}
\Av_\tau(f)&=W_\tau^{-1}\sum_{\sigma\in\Omega_\tau}W_\sigma f_\sigma,
\label{eq:mean}\\
\Var_\tau(f)&=\min_{\alpha\in\C^q}
\sum_{\sigma\in\Omega_\tau}(f_\sigma-\alpha)^*W_\sigma(f_\sigma-\alpha).
\label{eq:variance}
\end{align}
The variance here is unnormalized, as in the scalar formula above.
The following lemma identifies the minimizing representative and
the effective energy, and records the extra cost of any other choice.

\begin{lemma}[Weighted variance as an effective energy]
\label{lem:link-effective-variance}
For the full problem \eqref{eq:link-full-problem} and each
$f\in\II_\tau$, the unique minimizing representative and the effective
energy are
\[
\alpha_{\min}=\Av_\tau(f),\qquad
\Eff_\tau(f)=\Var_\tau(f).
\]
More precisely, for every $\alpha\in\C^q$,
\begin{equation}\label{eq:mean-response}
\Energy_\tau(f,\alpha)
=\Var_\tau(f)+(\alpha-\Av_\tau(f))^*W_\tau(\alpha-\Av_\tau(f)).
\end{equation}
\end{lemma}
\begin{proof}
We minimize $\Energy_\tau(f,\alpha)$ over $\alpha$ with $f$ fixed.
For any direction $h\in\C^q$ and real parameter $t$,
\[
\left.\frac{d}{dt}\Energy_\tau(f,\alpha+th)\right|_{t=0}
=2\operatorname{Re}\left[
h^*\left(W_\tau\alpha-\sum_{\sigma\in\Omega_\tau}W_\sigma f_\sigma\right)
\right],
\]
where $W_\tau=\sum_{\sigma\in\Omega_\tau}W_\sigma$.
Requiring this derivative to vanish for every $h$ gives the
stationary equation
\[
W_\tau\alpha=\sum_{\sigma\in\Omega_\tau}W_\sigma f_\sigma.
\]
Since $W_\tau\succ0$, the energy is strictly convex in $\alpha$, so
its unique minimizer is $\alpha=\Av_\tau(f)$.
The readout fixes $f$ and leaves precisely $\alpha$ free; hence its
minimum is both $\Eff_\tau(f)$ and $\Var_\tau(f)$.
For the final identity, write
\[
f_\sigma-\alpha
=(f_\sigma-\Av_\tau(f))-(\alpha-\Av_\tau(f)).
\]
Expanding the quadratic form and summing cancels the cross terms because
$\sum_\sigma W_\sigma(f_\sigma-\Av_\tau(f))=0$, which gives
\eqref{eq:mean-response}.
\end{proof}

The two terms in \eqref{eq:mean-response} separate the variation that
remains even at the optimal representative from the extra cost of choosing
a representative away from the mean. This extra-cost formula will also
measure the gluing defect below, when the child links have different
optimal means but must share one representative.

To identify the remaining ROL data, use the lift $f\mapsto(f,0)$;
the response decomposition then reads
\begin{equation}\label{eq:link-response-dictionary}
\Eff_\tau(f)=\Var_\tau(f),\qquad
h_\tau f=\Av_\tau(f),\qquad \TOp_\tau=W_\tau.
\end{equation}
The last equality uses the Euclidean internal metric.
Thus the mean is the minimizing response; $\Gamma_\tau$ only reads
the retained facet values.

The effective operator has a simple projection form. Define
\[
(\Pi_\tau^{\mathrm{const}}f)_\sigma:=\Av_\tau(f)
\qquad(\sigma\in\Omega_\tau).
\]
This is the orthogonal projection onto constant functions in the
weighted metric \eqref{eq:link-space}. Hence
\begin{equation}\label{eq:variance-operator}
\LOp_\tau=I-\Pi_\tau^{\mathrm{const}}.
\end{equation}
No commutation of the matrices $W_\sigma$ is used here.

\paragraph{Compatible localization and the gluing defect.}
We next connect these effective energies across one localization step.
Each child can minimize over its own representative, whereas localization
of a parent variable passes the same representative to every child.
The gluing defect records the cost of this shared choice.
For the counting and variance identities below, assume that $X$ is pure
with $d$ vertices in each facet, so its dimension is $d-1$.
Write $\tau i=\tau\cup\{i\}$ for a legal one-vertex extension and let
$m_\tau=d-|\tau|>0$ be the remaining rank. Each facet containing $\tau$ occurs in exactly
$m_\tau$ such extensions, hence
\begin{equation}\label{eq:incidence}
\frac1{m_\tau}\sum_i R_i^*R_i=I,
\qquad R_i=R_{\tau i\leftarrow\tau}.
\end{equation}
The full-variable localization keeps the internal coefficient:
\begin{equation}\label{eq:link-compatible-assembly}
r_i(f,\alpha)=(R_i f,\alpha),\qquad \Gamma_{\tau i}r_i=R_i\Gamma_\tau,\qquad
\Energy_\tau(f,\alpha)=\frac1{m_\tau}\sum_i\Energy_{\tau i}(R_i f,\alpha).
\end{equation}
These data satisfy \cref{lem:rol-one-step}.
The weights $1/m_\tau$ count facet incidences; they need not sum to
one over all legal vertices.
Set $\alpha=\Av_\tau(f)$ and $\alpha_i=\Av_{\tau i}(R_i f)$. Applying
\eqref{eq:mean-response} on each one-step link and counting facets yields
\begin{align}
\Var_\tau(f)&=\frac1{m_\tau}\sum_i\Var_{\tau i}(R_i f)
+\GG_\tau(f),\label{eq:link-splitting}\\
\GG_\tau(f)&=\frac1{m_\tau}\sum_i
(\alpha_i-\alpha)^*W_{\tau i}(\alpha_i-\alpha)\geq0.\label{eq:link-defect}
\end{align}
Equivalently, the response-cost formula of
\cref{prop:gluing-response-cost} minimizes the cost of fitting a common
representative to the child means:
\begin{equation}\label{eq:link-gluing-minimum}
\Delta_\tau(f)=\GG_\tau(f)
=\min_{c\in\C^q}\frac1{m_\tau}\sum_i(c-\alpha_i)^*W_{\tau i}(c-\alpha_i).
\end{equation}
Facet counting gives $\sum_iW_{\tau i}=m_\tau W_\tau$ and
$\sum_iW_{\tau i}\alpha_i=m_\tau W_\tau\Av_\tau(f)$, so the minimizer is
$c=\Av_\tau(f)$. The defect vanishes exactly when all local means
agree. It is a positive energy contribution, not necessarily an error
to be made small.

This example realizes the ROL decomposition, but does not by itself
provide the common-metric subspace hierarchy required for geometric
descent. Its scalar specialization yields the law of total variance;
\cref{sec:classical-dictionaries} develops this connection and the
resulting classical gap bounds.

\subsection{Observable spaces: passing a residual to the next node}
\label{sec:toolkit-observables}
We next construct a common-space localization hierarchy in which
enlarging the observable space removes newly observable components
and passes the remaining residual to the next node.
Here we choose to label each node by a set of sites, without specifying
their values. In one fixed Hilbert space $\mathcal K$, choose nested
observable spaces $\mathcal A_\tau$, with
$\mathcal A_\tau\subseteq\mathcal A_\nu$ whenever $\tau\subseteq\nu$,
and let $\Pi_\tau^{\mathrm{obs}}$ project onto
$\mathcal A_\tau$. Set $\II_\tau=\mathcal A_\tau^\perp$ and
$\Pi_\tau=I-\Pi_\tau^{\mathrm{obs}}$. For $\tau\subseteq\nu$, the retained map is
$R_{\nu\leftarrow\tau}=\Pi_\nu|_{\II_\tau}$.

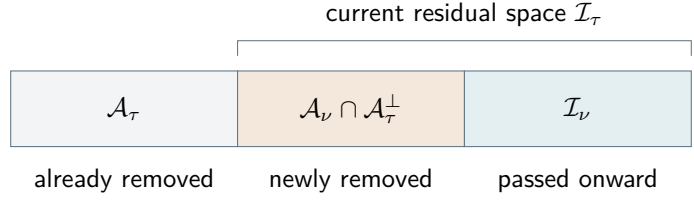
\begin{figure}[htbp]
\centering
\begin{tikzpicture}[every node/.style={font=\small}]
\draw[draw=linkgray,fill=linkgray!8] (0,0) rectangle (3,1);
\draw[draw=linkgray,fill=defectorange!15] (3,0) rectangle (6,1);
\draw[draw=linkgray,fill=linkteal!12] (6,0) rectangle (9,1);
\node at (1.5,.5) {$\mathcal A_\tau$};
\node at (4.5,.5) {$\mathcal A_\nu\cap\mathcal A_\tau^\perp$};
\node at (7.5,.5) {$\II_\nu$};
\draw[linkgray] (3,1.2)--(3,1.4)--(9,1.4)--(9,1.2);
\node at (6,1.75) {current residual space $\II_\tau$};
\node at (1.5,-.45) {already removed};
\node at (4.5,-.45) {newly removed};
\node at (7.5,-.45) {passed onward};
\end{tikzpicture}
\caption{One localization in a common metric. The three blocks are
orthogonal subspaces of $\mathcal K$, not tensor factors or dimensions
drawn to scale. Enlarging the observable space shrinks the residual.}
\label{fig:toolkit-residual}
\end{figure}

\paragraph{The meaning of one step.}
For $b\in\II_\tau$, the orthogonal decomposition is
\[
b=(\Pi_\nu^{\mathrm{obs}}-\Pi_\tau^{\mathrm{obs}})b+\Pi_\nu b.
\]
The first component is newly representable on the enlarged site set;
the second is the residual passed to the child. Adding a site can
remove correlations with previously selected sites, not only an
observable of the new site alone. Nested projections compose along
paths, although projections at incomparable nodes need not commute.

Classically, take $\mathcal K=L^2(\mu)$ and let $\mathcal A_\tau$
be the functions of $X_\tau$. Then $\Pi_\tau^{\mathrm{obs}}$ is conditional expectation.
Here $f-\mathbb E[f\mid X_\tau]$ retains residuals across all boundary
values, instead of fixing a single assignment $X_\tau=s$.
The classical construction is detailed in \cref{sec:classical-configurations}.
For qubits or qudits, one can similarly take local observable algebras
in one fixed KMS Hilbert space, as discussed in
\cref{par:outlook-gibbs}.
This specifies the spaces and localization maps; compatible energies
and the additional assumptions needed for descent must be supplied
separately.

\section{Descent of Spectral Gaps from Subspace Geometry}
\label{sec:weighted-descent}
To propagate local energy bounds into a global spectral estimate,
we control how much norm is retained under localization. We study this
question for nested subspaces indexed by a weighted simplicial complex.
The complex need not be complete, and different outgoing steps may
have different weights. The result combines three quantitative inputs:
overlap bounds near the stopping level, spectral gaps of the link
walks, and terminal energy lower bounds.

We first specify the hierarchy and its geometric and energy inputs,
then state and prove the gap theorem. The complete uniform specialization
gives an explicit coefficient in \cref{cor:weighted-complete-gap}.

\subsection{Links, spaces, and localization maps}
\label{sec:weighted-hierarchy}
Let $X$ be a finite pure simplicial complex whose facets have
cardinality $d\geq2$. Give each facet $F$ a positive weight $w_F$.\footnote{These facet weights are distinct from the ROL assembly weights $\omega_a$; their relation is discussed in \cref{sec:descent-rol-energy}.}
For a face $\tau$, define its total extension weight and remaining
rank by
\begin{equation}\label{eq:weighted-face-mass}
W_\tau:=\sum_{F\supseteq\tau}w_F,
\qquad m_\tau:=d-|\tau|.
\end{equation}
Write $N_\tau=\{i:\tau i\in X,\ i\notin\tau\}$ for the vertices
of $\lk_X(\tau)$. The number $m_\tau$ counts the vertices still
needed to complete a facet. It need not equal $|N_\tau|$.

The ROL label graph has faces as nodes and legal inclusions as arrows:
\[
\mathscr T_k=\{\tau\in X:|\tau|=k\},
\qquad a:\tau\longrightarrow\tau i\quad(i\in N_\tau).
\]
We stop one level before the facets, so the visited nodes are
\begin{equation}\label{eq:weighted-stopping}
\mathscr T_{\mathrm{stop}}=\mathscr T_{d-1},\qquad
\mathscr T_{\mathrm{int}}=\bigsqcup_{k=0}^{d-2}\mathscr T_k,
\qquad
\mathscr T=\mathscr T_{\mathrm{int}}\sqcup\mathscr T_{\mathrm{stop}}.
\end{equation}
Thus $m_\tau=1$ at a stop. A root-to-stop path
$p=(\tau_0=\varnothing,\tau_1,\ldots,\tau_{d-1})$ records an
ordered sequence of legal additions, with endpoint $t(p)=\tau_{d-1}$.
Different paths can reach the same node. No space or energy at a
facet is needed for the argument.

The deepest internal level consists of the parents of stopping nodes:
\begin{equation}\label{eq:weighted-bottom-level}
\mathscr T_{\mathrm{bot}}:=\mathscr T_{d-2}
=\{\tau\in\mathscr T_{\mathrm{int}}:m_\tau=2\}.
\end{equation}
We write $\tau_\star$ for a node in this level when emphasizing the
bottom geometric input; its children $\tau_\star i$ are stopping nodes.
The undecorated symbol $\tau$ continues to denote a general node.

\subsubsection{Weights and walks on links}
At each internal node $\tau$, a classical weighted random walk on
the vertices of the link $\lk_X(\tau)$ measures the connectivity
among the outgoing localization labels. Its spectral gap enters
the one-step Gram estimate used in the descent.
The walk has edge weights $W_{\tau ij}$ on the $1$-skeleton of the
link. It can be viewed as a non-lazy
up--down walk: from a vertex, choose an incident edge proportionally
to its weight, then move to its other endpoint
(\Cref{fig:link-up-down-walk}). The outgoing probability for the ROL
hierarchy and the transition matrix of this link walk are
\begin{equation}\label{eq:weighted-link-walk}
\pi_\tau(i):=\frac{W_{\tau i}}{m_\tau W_\tau},\qquad
P_\tau(i,j):=\frac{W_{\tau ij}}{(m_\tau-1)W_{\tau i}}
\quad(i\ne j),\qquad P_\tau(i,i):=0.
\end{equation}
Here $W_{\tau ij}=0$ when $\tau ij\notin X$.
For the complete uniform hierarchy, explicit formulas for $\pi_\tau$,
$P_\tau$, and the link-walk gap are given in
\eqref{eq:weighted-complete-parameters} below.

\begin{figure}[htbp]
\centering
\resizebox{.98\linewidth}{!}{%
\begin{tikzpicture}[x=1cm,y=1cm,font=\small,
  >=Stealth,line cap=round,line join=round]
\node[rolpanel] at (3.6,3.5) {(a) In the link};
\node[rolpanel] at (12,3.5) {(b) In the original complex};
\fill[linkteal] (0,0) circle (2.5pt);
\node[below=5pt] at (0,0) {$i$};
\draw[thick,linkgray] (3,1.5) -- (4.2,1.5);
\fill[linkteal] (3,1.5) circle (2.5pt);
\fill[defectorange] (4.2,1.5) circle (2.5pt);
\node[above=6pt] at (3.6,1.5) {$\{i,j\}$};
\fill[defectorange] (7.2,0) circle (2.5pt);
\node[below=5pt] at (7.2,0) {$j$};
\draw[->,thick,linkteal] (.3,.15) -- (2.65,1.32)
  node[midway,above,sloped,font=\small\sffamily] {up};
\draw[->,thick,defectorange] (4.55,1.32) -- (6.9,.15)
  node[midway,above,sloped,font=\small\sffamily] {down};
\node[font=\scriptsize\sffamily,align=center] at (1.7,-.85)
  {Choose an incident edge\\with probability proportional to $W_{\tau ij}$};
\node[font=\scriptsize\sffamily,align=center] at (5.65,-.85)
  {Move to the\\other endpoint};
\foreach \x/\y/\rx/\ry in {9/0/.72/.62,12/1.5/.84/.72,15/0/.72/.62}{
  \begin{scope}[shift={(\x,\y)}]
    \draw[linkgray!55,fill=linkgray!3] (0,.325) ellipse (\rx cm and \ry cm);
    \foreach \s in {-.4,0,.4}
      \fill[linkgray!65] (\s,0) circle (2pt);
  \end{scope}
}
\fill[linkteal] (9,.65) circle (2.5pt);
\fill[linkteal] (11.6,2.15) circle (2.5pt);
\fill[defectorange] (12.4,2.15) circle (2.5pt);
\fill[defectorange] (15,.65) circle (2.5pt);
\node[below] at (9,-.9) {$\tau\cup\{i\}$};
\node[above] at (12,2.55) {$\tau\cup\{i,j\}$};
\node[below] at (15,-.9) {$\tau\cup\{j\}$};
\draw[->,thick,linkteal] (9.8,.5)--(11,1.45)
  node[midway,above,sloped,font=\small\sffamily] {add $j$};
\draw[->,thick,defectorange] (13,1.45)--(14.2,.5)
  node[midway,above,sloped,font=\small\sffamily] {remove $i$};
\node[font=\scriptsize\sffamily] at (12,-.85)
  {The vertices of $\tau$ stay fixed};
\end{tikzpicture}}
\caption{One step of the link walk $P_\tau$: vertex to edge to the
other vertex (left), or equivalently
$\tau\cup\{i\}\to\tau\cup\{i,j\}\to\tau\cup\{j\}$ in the original
complex (right). The gray groups schematically represent the vertices
of the same fixed face $\tau$; this is not a
path descending the ROL hierarchy. Choosing either endpoint uniformly
in the down step would instead give the lazy walk
$(\id+P_\tau)/2$.}
\label{fig:link-up-down-walk}
\end{figure}
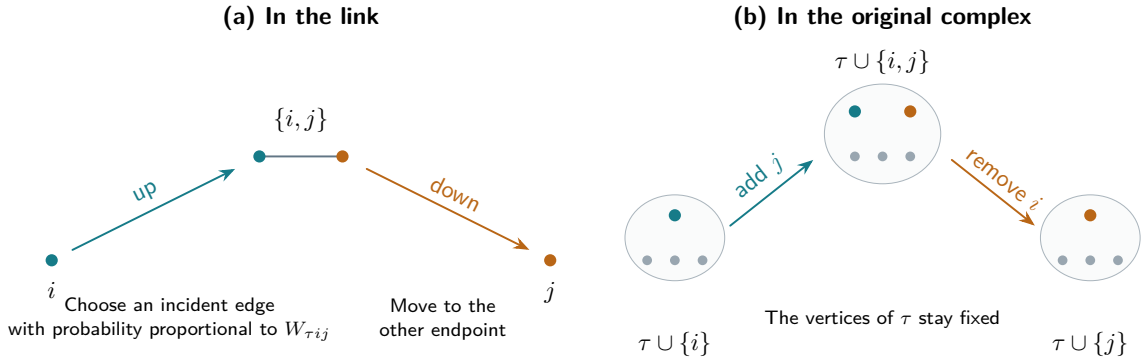

Equivalently, $\pi_\tau$ is sampled by choosing a facet $F$ containing
$\tau$ proportionally to its weight, then a vertex of $F\setminus\tau$
uniformly. To sample $P_\tau(i,\cdot)$, choose a facet $F$ containing
$\tau i$ proportionally to its weight, then choose $j$ uniformly from
$F\setminus(\tau\cup\{i\})$. The facet is sampled afresh at each step.
Counting the remaining vertices in each facet gives
\[
\sum_iW_{\tau i}=m_\tau W_\tau,\qquad
\sum_{j\ne i}W_{\tau ij}=(m_\tau-1)W_{\tau i}.
\]
Consequently $\pi_\tau$ is a probability distribution and $P_\tau$
is a Markov matrix. It is reversible, since
\begin{equation}\label{eq:weighted-link-detailed-balance}
\pi_\tau(i)P_\tau(i,j)
=\frac{W_{\tau ij}}{m_\tau(m_\tau-1)W_\tau}
=\pi_\tau(j)P_\tau(j,i).
\end{equation}
Moreover, on each legal edge, the transition distribution from $i$ in
the parent's link agrees with the outgoing distribution at the child
$\tau i$:
\begin{equation}\label{eq:weighted-link-consistency}
P_\tau(i,j)=\pi_{\tau i}(j)\qquad(m_\tau\geq3).
\end{equation}
This consistency is why one common set of facet
weights is used, rather than unrelated choices at different nodes.
Even equal facet weights need not give uniform outgoing probabilities:
the numbers of containing facets may differ between extensions.

Assume that each link graph with $m_\tau\geq2$ is connected.
For $m_\tau\geq3$, let $\lambda_\tau\in(0,2]$ be a specified
lower bound on the spectral gap of the classical link walk $P_\tau$,
measured by its normalized Laplacian $\id-P_\tau$:
\begin{equation}\label{eq:weighted-link-gap}
\operatorname{gap}(\id-P_\tau)
=1-\beta_2(P_\tau)\geq\lambda_\tau>0
\quad\text{in }L^2(\pi_\tau).
\end{equation}
Here $\beta_2(P_\tau)$ denotes the second-largest eigenvalue. This is not the absolute
spectral gap, and bipartite link graphs are allowed.

\subsubsection{The common-metric subspace hierarchy}
Fix a finite-dimensional Hilbert space $\mathcal K$. At every visited
face choose a linear subspace $\II_\tau\subseteq\mathcal K$, with
its inherited metric. The structural requirements are
\begin{equation}\label{eq:weighted-space-conditions}
\II_{\tau i}\subseteq\II_\tau,
\qquad \sum_{i\in N_\tau}\II_{\tau i}=\II_\tau
\qquad(\tau\in\mathscr T_{\mathrm{int}}).
\end{equation}
The sum denotes linear span, not a direct sum: the child spaces may
overlap and need not be mutually orthogonal. Thus every vector in
$\II_\tau$ can be written as $\sum_{i\in N_\tau}u_i$ with
$u_i\in\II_{\tau i}$, but this decomposition need not be unique.
Let $\Pi_\tau$ be the orthogonal projection in $\mathcal K$ onto
$\II_\tau$. The localization and removed-space projection are
\begin{equation}\label{eq:weighted-localization}
R_a=\Pi_{\tau i}|_{\II_\tau}:\II_\tau\to\II_{\tau i},
\qquad E_{\tau,i}=\Pi_\tau-\Pi_{\tau i}
\quad(a:\tau\to\tau i).
\end{equation}
In particular,
\[
\mathcal U_{\tau,i}:=\II_\tau\cap\II_{\tau i}^\perp,
\qquad \II_\tau=\II_{\tau i}\oplus\mathcal U_{\tau,i},
\qquad \bigcap_{i\in N_\tau}\mathcal U_{\tau,i}=\{0\}.
\]
Thus $E_{\tau,i}$ is the orthogonal projection onto the removed
subspace $\mathcal U_{\tau,i}$.
The last identity follows from the span condition. Zero-dimensional
spaces and zero increments are allowed.

Along a fixed path, these splittings form an orthogonal filtration:
\[
\II_\varnothing=\II_{t(p)}\oplus
\bigoplus_{k=0}^{d-2}
\bigl(\II_{\tau_k}\cap\II_{\tau_{k+1}}^\perp\bigr).
\]
In contrast, the removed spaces at different outgoing edges of one
node can overlap. The next quantity measures their weighted overlap.
This nested, common-metric geometry is an additional assumption beyond
general ROL; specifying the label complex alone does not supply it.

\subsection{Geometric inputs and Gram recursion}
\label{sec:weighted-gap-inputs}
\subsubsection{Gram data}
For each internal node $\tau$, define its \emph{Gram norm} by
\begin{equation}\label{eq:weighted-gram}
\theta_\tau:=\left\|\sum_{i\in N_\tau}\pi_\tau(i)E_{\tau,i}|_{\II_\tau}\right\|.
\end{equation}
Thus $\theta_\tau$ measures the worst-case average loss of squared
norm under localization:
\begin{equation}\label{eq:weighted-loss}
\sum_{i\in N_\tau}\pi_\tau(i)\|E_{\tau,i}b\|^2
\leq\theta_\tau\|b\|^2\qquad(b\in\II_\tau).
\end{equation}

\begin{remark}[Gram-operator interpretation]
Let $J_{\tau,i}:\mathcal U_{\tau,i}\hookrightarrow\II_\tau$ be the
isometric inclusions. On the external orthogonal direct sum, define
\begin{equation}\label{eq:weighted-synthesis}
\mathsf S_\tau^\pi:\bigoplus_{i\in N_\tau}\mathcal U_{\tau,i}\to\II_\tau,
\qquad \mathsf S_\tau^\pi((u_i)_i)
=\sum_i\sqrt{\pi_\tau(i)}\,J_{\tau,i}u_i.
\end{equation}
The weighted Gram operator is $\bm C_\tau^\pi=(\mathsf S_\tau^\pi)^*\mathsf S_\tau^\pi$,
with blocks
\[
(\bm C_\tau^\pi)_{ij}
=\sqrt{\pi_\tau(i)\pi_\tau(j)}J_{\tau,i}^*J_{\tau,j}.
\]
By \cref{lem:gram-norm}, $\|\bm C_\tau^\pi\|=\theta_\tau$;
this explains the term Gram norm.
The external coordinates are orthogonal, but their images in
$\II_\tau$ need not be.
\end{remark}

\begin{figure}[tbp]
\centering
\begin{tikzpicture}[x=1cm,y=1cm,font=\small,
  >=Stealth,line cap=round,line join=round]
\node[rolpanel] at (2.3,3.5) {(a) The last branching};
\node[rolpanel] at (8.6,3.5) {(b) Removed directions};
\node[anchor=east,text=linkgray,font=\footnotesize] at (-.15,2.5)
  {$\mathscr T_{\mathrm{bot}}$};
\node[anchor=east,text=linkgray,font=\footnotesize] at (-.15,.3)
  {$\mathscr T_{\mathrm{stop}}$};
\node[circle,draw=linkgray,minimum size=9mm] (gramparent) at (2.3,2.5) {$\tau_\star$};
\node[circle,draw=linkteal,fill=linkteal!7,minimum size=9mm] (grami) at (.65,.3) {$\tau_\star i$};
\node[circle,draw=defectorange,fill=defectorange!7,minimum size=9mm] (gramj) at (3.3,.3) {$\tau_\star j$};
\node at (4.35,.3) {$\cdots$};
\draw[->,thick,linkteal] (gramparent) -- node[left,font=\small] {$\mathcal U_{\tau_\star,i}$} (grami);
\draw[->,thick,defectorange] (gramparent) -- node[right,font=\small] {$\mathcal U_{\tau_\star,j}$} (gramj);
\draw[->,linkgray,dashed] (gramparent) -- (4.3,.8);
\node[font=\small\sffamily] at (2.3,-.65) {Stopping nodes};
\draw[draw=linkgray,fill=linkgray!3] (8.6,1.35) ellipse (2.65 and 1.65);
\node[text=linkgray] at (7.3,2.15) {$\II_{\tau_\star}$};
\begin{scope}[shift={(8.4,.9)}]
\draw[rotate=15,draw=linkteal,fill=linkteal!9] (0,0) ellipse (1.8 and .42);
\draw[rotate=65,draw=defectorange,fill=defectorange!9,fill opacity=.7] (.5,0) ellipse (1.25 and .42);
\draw[->,thick,linkgray] (0,0) -- (.25,1.65) node[above] {$b$};
\draw[->,thick,linkteal] (0,0) -- (.646,.173);
\draw[->,thick,defectorange] (0,0) -- (.677,1.451);
\draw[densely dotted,linkteal] (.25,1.65) -- (.646,.173);
\draw[densely dotted,defectorange] (.25,1.65) -- (.677,1.451);
\node[text=linkteal,font=\scriptsize] at (1.05,.25) {$E_{\tau_\star,i}b$};
\node[text=defectorange,font=\scriptsize,right] at (.78,1.18) {$E_{\tau_\star,j}b$};
\node[text=linkteal] at (1.9,.75) {$\mathcal U_{\tau_\star,i}$};
\node[text=defectorange] at (-.45,-.87) {$\mathcal U_{\tau_\star,j}$};
\fill[linkgray] (0,0) circle (1.3pt);
\end{scope}
\node at (8.6,-.9) {$\displaystyle\sum_{i\in N_{\tau_\star}}\pi_{\tau_\star}(i)\|E_{\tau_\star,i}b\|^2
  \leq\theta_{\tau_\star}\|b\|^2$};
\end{tikzpicture}
\caption{The deepest Gram input is measured at a parent $\tau_\star\in\mathscr T_{\mathrm{bot}}$ of stopping
nodes, across its outgoing localizations. Each $\mathcal U_{\tau_\star,i}$
contains the directions removed along one branch. The same parent
vector can have substantial projections onto several such spaces;
$\theta_{\tau_\star}$ bounds the weighted fraction of squared norm removed.
Only two branches are illustrated in detail; a parent at this last
internal level may have more children. The rounded regions are schematic depictions of subspaces:
their overlapping areas do not represent subspace intersections.
Nonzero Gram overlap can occur even when the common intersection is
$\{0\}$.}
\label{fig:deepest-gram-geometry}
\end{figure}
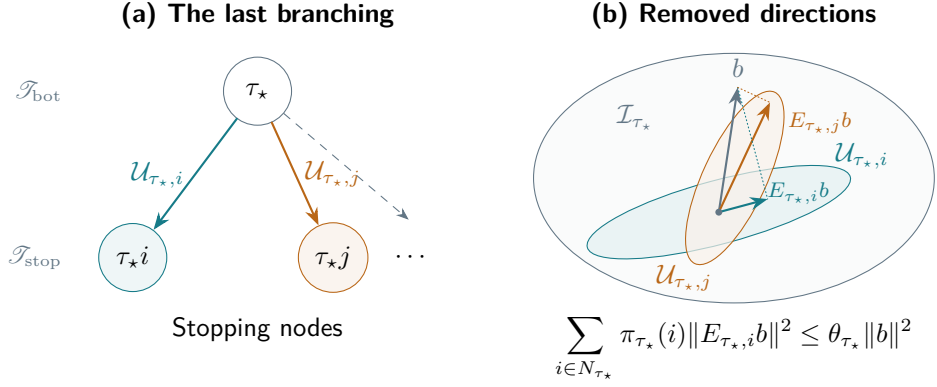

The deepest geometric input is imposed one step before the stopping
nodes, as illustrated in \Cref{fig:deepest-gram-geometry}.
These parents form $\mathscr T_{\mathrm{bot}}$: a node $\tau_\star$
has two vertices still needed to complete a facet, while its children
have only one and are stopping nodes. At each node $\tau_\star$ in
this level, supply a possibly node-dependent \emph{Gram bound}
$\overline{\theta}_{\tau_\star}$, meaning an upper bound on its Gram norm:
\begin{equation}\label{eq:weighted-bottom-input}
\theta_{\tau_\star}\leq \overline{\theta}_{\tau_\star}<1,\qquad \overline{\theta}_{\tau_\star}\geq0
\qquad(\tau_\star\in\mathscr T_{\mathrm{bot}}).
\end{equation}
The remaining rank does not count the available children: a rank-two
link may have more than two vertices. Thus this is generally a multi-subspace
Gram bound, not just an angle between two spaces.

The following update comes from the one-step Gram estimate proved in
\cref{lem:weighted-one-step-gram}. Write
\[
F_\lambda(c):=\frac{(2-\lambda)c}{\lambda(1-c)}
\qquad(0<\lambda\leq2,\ 0\leq c<1).
\]
Starting with these Gram bounds on $\mathscr T_{\mathrm{bot}}$, recursively
set, for $\tau\in\mathscr T_{\mathrm{int}}\setminus\mathscr T_{\mathrm{bot}}$,
\begin{equation}\label{eq:weighted-gram-recursion}
c_\tau^{\downarrow}:=\max_{i\in N_\tau}\overline{\theta}_{\tau i},\qquad
\overline{\theta}_\tau:=F_{\lambda_\tau}(c_\tau^{\downarrow}).
\end{equation}
The downward arrow indicates that the child Gram bound
$c_\tau^{\downarrow}$ collects the largest Gram bound from the
children of $\tau$; $\overline{\theta}_\tau$ is the propagated
Gram bound at the parent.
To keep these bounds below one, require at each such node
\begin{equation}\label{eq:weighted-retention-condition}
c_\tau^{\downarrow}<\lambda_\tau/2.
\end{equation}
Bottom Gram bounds and internal
link-walk gap bounds are separate inputs: neither is inferred merely from
the existence of a weighted complex. 

\subsection{From ROL decomposition to energy transport}
\label{sec:descent-rol-energy}
The geometric recursion controls how much norm is retained. To turn
this control into a spectral-gap bound, we need compatible energy
bounds. The connection with ROL is its one-step decomposition:
under the compatibility and full-energy assembly assumptions of
\cref{lem:rol-one-step}, \eqref{eq:rol-one-step-energy} gives
\[
\Eff_\tau(b)
=\sum_{a:\tau\to\tau i}\omega_a\Eff_{\tau i}(R_ab)
+\langle b,\DOp_\tau b\rangle,
\qquad \DOp_\tau\succeq0.
\]
Here $\omega_a$ is an energy assembly coefficient, whereas the facet
weights $w_F$ determine the child probabilities $\pi_\tau(i)$ and
the link walks. For the descent below, the assembly coefficients
must match these child probabilities, possibly after normalizing
the energies.

More explicitly, dividing the energy at each node by a positive
number $s_\tau$ changes its assembly coefficients to
\[
\widehat{\Eff}_\tau:=\frac{\Eff_\tau}{s_\tau},
\qquad
\widehat \omega_a:=\omega_a\frac{s_{\tau i}}{s_\tau}
\quad(a:\tau\to\tau i).
\]
The defect is divided by $s_\tau$ and remains nonnegative.
The required matching is $\widehat \omega_a=\pi_\tau(i)$; it is a
condition on the energy construction, not an automatic consequence
of the facet weights. For example, the scalar variance construction
in \cref{prop:classical-variance-rol} has $\omega_a=1/m_\tau$ before
normalization. Taking $s_\tau=W_\tau$ gives
\[
\widehat \omega_a=\frac{W_{\tau i}}{m_\tau W_\tau}=\pi_\tau(i).
\]
In what follows, the energies and their terminal bounds are understood
in the chosen normalization, and we omit the hats.

\paragraph{Energy assembly and terminal lower bounds.}
Let $\LOp_\tau\succeq0$ act on $\II_\tau$, and write
$\Eff_\tau(b)=\langle b,\LOp_\tau b\rangle$.
For a compatible ROL construction with matching coefficients, the
passage from the exact decomposition to a lower bound is
\[
\begin{aligned}
\Eff_\tau(b)
&=\sum_{a:\tau\to\tau i}\pi_\tau(i)\Eff_{\tau i}(R_ab)
  +\underbrace{\langle b,\DOp_\tau b\rangle}_{\geq0}\\
&\geq\sum_{a:\tau\to\tau i}\pi_\tau(i)\Eff_{\tau i}(R_ab).
\end{aligned}
\]
Thus the lower bound retains only the energy inherited from the
children; any additional energy supplied by the gluing defect is left
unused here. \Cref{sec:compatible-assembly} develops spectral bounds
that use this additional contribution.
The descent argument requires only the final inequality, which we
take as our assumption whether obtained from ROL or verified directly:
\begin{equation}\label{eq:weighted-energy-assembly}
\Eff_\tau(b)\geq\sum_{a:\tau\to\tau i}
\pi_\tau(i)\Eff_{\tau i}(R_ab)
\qquad(b\in\II_\tau,\ \tau\in\mathscr T_{\mathrm{int}}).
\end{equation}
At the stopping nodes, assume positive lower bounds, which may depend
on the node:
\begin{equation}\label{eq:weighted-terminal-input}
\Eff_\sigma(b)\geq\delta_\sigma\|b\|^2,
\qquad\delta_\sigma>0
\qquad(\sigma\in\mathscr T_{\mathrm{stop}},\ b\in\II_\sigma).
\end{equation}
The geometric inputs control loss of norm during localization;
the terminal constants supply the energy scale.

\subsection{The main gap theorem}
The following statement first propagates the geometric input and then
transports terminal energy bounds along every root-to-stop path.
\begin{theorem}[Spectral descent along the hierarchy]
\label{thm:weighted-link-gap}
In the hierarchy of \cref{sec:weighted-hierarchy}, take the following
numerical inputs:
\begin{itemize}
\item deepest Gram bounds $\overline{\theta}_{\tau_\star}$ for
$\tau_\star\in\mathscr T_{\mathrm{bot}}$ satisfying
\eqref{eq:weighted-bottom-input};
\item classical link-walk gap lower bounds
$\lambda_\tau>0$ for
$\tau\in\mathscr T_{\mathrm{int}}\setminus\mathscr T_{\mathrm{bot}}$
satisfying \eqref{eq:weighted-link-gap};
\item terminal energy lower bounds $\delta_\sigma>0$ for
$\sigma\in\mathscr T_{\mathrm{stop}}$
satisfying \eqref{eq:weighted-terminal-input}.
\end{itemize}
Define the remaining $\overline{\theta}_\tau$ by \eqref{eq:weighted-gram-recursion},
and assume the retention condition \eqref{eq:weighted-retention-condition}
and the energy assembly inequality \eqref{eq:weighted-energy-assembly}.
Then
$\theta_\tau\leq \overline{\theta}_\tau<1$ at every internal node.
Extend the given terminal energy lower bounds $\delta_\sigma$ to all internal nodes
recursively by
\begin{equation}\label{eq:weighted-gap-recursion}
\delta_\tau:=(1-\overline{\theta}_\tau)\min_{i\in N_\tau}\delta_{\tau i}
\quad(\tau\in\mathscr T_{\mathrm{int}}).
\end{equation}
If $\II_\varnothing\ne\{0\}$, then $\LOp_\varnothing$ is positive
definite and its spectral gap satisfies
\begin{equation}\label{eq:weighted-root-path-bound}
\operatorname{gap}(\LOp_\varnothing)\geq \delta_\varnothing
=\min_{p\in\mathcal P_{\mathrm{stop}}}
\left\{\delta_{t(p)}\prod_{k=0}^{d-2}(1-\overline{\theta}_{\tau_k})\right\},
\end{equation}
\end{theorem}

Under the root identifications of \cref{par:rol-root-full-gap}, this
also bounds the gap of the original operator; a general readout
requires an additional full-space comparison.

\subsection{Proof of the main gap theorem}
\label{sec:weighted-gap-proof}
The proof separates geometry from energy. We first express norm
retention in terms of Gram overlap and state the geometric bound
needed throughout the hierarchy. We then use that bound to transport
terminal energy lower bounds to the root and obtain the global gap.
Finally, we prove the geometric bound using the link-walk gaps.

Starting from the bottom Gram bounds $\overline{\theta}_{\tau_\star}$ and terminal
energy bounds $\delta_\sigma$, the updates propagate these bounds upward
from children to parents:
\[
\begin{aligned}
\text{Geometry:}\quad
\underbrace{\overline{\theta}_\tau}_{\text{parent bound}}
&=F_{\lambda_\tau}\!\left(
\underbrace{\max_i\overline{\theta}_{\tau i}}_{\text{child bounds}}
\right),\\[6pt]
\text{Energy:}\quad
\underbrace{\delta_\tau}_{\text{parent bound}}
&=\underbrace{(1-\overline{\theta}_\tau)}_{\text{retention factor}}
\underbrace{\min_i\delta_{\tau i}}_{\text{child bounds}}.
\end{aligned}
\]
The geometric update applies above the bottom level; the energy
update applies at every internal node.

\subsubsection{Norm retention and the Gram-descent bound}
For $b\in\II_\tau$ and $a:\tau\to\tau i$, localization retains
$R_ab$ and removes $E_{\tau,i}b$. These components are orthogonal, so
\[
\|R_ab\|^2=\|b\|^2-\|E_{\tau,i}b\|^2.
\]
Averaging over the outgoing branches and using
$\sum_i\pi_\tau(i)=1$ gives
\[
\begin{aligned}
\sum_{a:\tau\to\tau i}\pi_\tau(i)\|R_ab\|^2
&=\|b\|^2-\sum_i\pi_\tau(i)\|E_{\tau,i}b\|^2\\
&\geq(1-\theta_\tau)\|b\|^2.
\end{aligned}
\]
The last step is the average-loss bound \eqref{eq:weighted-loss},
which follows from the operator-norm identity \eqref{eq:weighted-gram}.
Thus $\theta_\tau$ bounds the fraction of squared norm lost on average,
while $1-\theta_\tau$ is the fraction guaranteed to remain.
Equivalently, in operator form,
\begin{equation}\label{eq:weighted-norm-retention}
\sum_{a:\tau\to\tau i}\pi_\tau(i)R_a^*R_a
=\id_{\II_\tau}-\sum_i\pi_\tau(i)E_{\tau,i}|_{\II_\tau}
\succeq(1-\theta_\tau)\id_{\II_\tau}.
\end{equation}
The span condition and finite dimensionality imply $\theta_\tau<1$:
equality would require a unit vector in every removed space, whose
intersection is zero. If $\II_\tau=\{0\}$, its norm is zero instead.
This strict inequality is qualitative. The role of Gram descent is
to obtain the explicit Gram bounds $\overline{\theta}_\tau$ needed in
\eqref{eq:weighted-norm-retention}.

\begin{proposition}[Gram descent through the hierarchy]
\label{prop:weighted-gram-descent}
Consider the subspace hierarchy of \cref{sec:weighted-hierarchy},
with classical link-walk gap bounds \eqref{eq:weighted-link-gap}
and bottom Gram bounds \eqref{eq:weighted-bottom-input}.
Define the remaining Gram bounds by
\eqref{eq:weighted-gram-recursion} and assume
\eqref{eq:weighted-retention-condition}. Then
\[
\theta_\tau\leq\overline{\theta}_\tau<1
\qquad(\tau\in\mathscr T_{\mathrm{int}}).
\]
\end{proposition}
We defer its proof to \cref{sec:gram-descent-proof} and first show
how these bounds yield the energy conclusion of the main theorem.

\subsubsection{Energy transport and the global gap}
\label{sec:gap-transport-proof}
\begin{proof}[Proof of \cref{thm:weighted-link-gap} assuming
\cref{prop:weighted-gram-descent}]
The proposition gives the geometric conclusion. We now prove the
stronger nodewise energy bound
\begin{equation}\label{eq:weighted-gap-conclusion}
\Eff_\tau(b)\geq \delta_\tau\|b\|^2\quad(b\in\II_\tau),
\qquad \LOp_\tau\succeq \delta_\tau\id_{\II_\tau},
\end{equation}
at every visited node. At a stopping node the
energy bound is \eqref{eq:weighted-terminal-input}. Suppose it holds
at every child of $\tau$. For $b\in\II_\tau$, assembly and norm
retention give
\begin{align*}
\Eff_\tau(b)
&\geq\sum_{a:\tau\to\tau i}\pi_\tau(i)
\delta_{\tau i}\|R_ab\|^2\\
&\geq\left(\min_i \delta_{\tau i}\right)(1-\theta_\tau)\|b\|^2
\geq \delta_\tau\|b\|^2.
\end{align*}
The last step uses \cref{prop:weighted-gram-descent} and the
definition of $\delta_\tau$ in \eqref{eq:weighted-gap-recursion}.
This proves \eqref{eq:weighted-gap-conclusion} by induction.

To make the path formula explicit, write a root-to-stop path as
$p=(\tau_0,\ldots,\tau_{d-1})$, where
$\tau_0=\varnothing$ and $\tau_{d-1}=t(p)$.
At each internal node, \eqref{eq:weighted-gap-recursion} reads
\[
\delta_\tau
=\min_{i\in N_\tau}
\bigl\{(1-\overline{\theta}_\tau)\delta_{\tau i}\bigr\}.
\]
The retention factor is positive, so substituting the same formula
for each child lets us combine the successive minima. More generally,
after $j$ steps, for $0\leq j\leq d-1$,
\[
\delta_\varnothing
=\min_{\substack{\tau_0=\varnothing\to\tau_1\to\cdots\to\tau_j}}
\left\{\delta_{\tau_j}
\prod_{k=0}^{j-1}(1-\overline{\theta}_{\tau_k})\right\},
\]
with the empty product equal to one. At $j=d-1$, all endpoints
are stopping nodes, so $\delta_{\tau_j}$ is the prescribed terminal
energy lower bound. This expansion is over paths, and remains valid
when different paths reach the same node.

All factors are positive. Thus \eqref{eq:weighted-gap-conclusion}
at the nonzero root space makes $\LOp_\varnothing$ positive definite,
and gives
\[
\operatorname{gap}(\LOp_\varnothing)\geq\delta_\varnothing
=\min_{p\in\mathcal P_{\mathrm{stop}}}
\left\{\delta_{t(p)}
\prod_{k=0}^{d-2}(1-\overline{\theta}_{\tau_k})\right\},
\]
which is \eqref{eq:weighted-root-path-bound}.
\end{proof}

\subsubsection{Proof of the Gram-descent bound}
\label{sec:gram-descent-proof}
It remains to establish \cref{prop:weighted-gram-descent}.
The recursion \eqref{eq:weighted-gram-recursion} is obtained by
repeatedly applying a one-step estimate. As we will show below, if
all child Gram norms satisfy $\theta_{\tau i}\leq c^{\downarrow}<1$, the parent's
link-walk gap bound $\lambda_\tau$ gives
\[
\theta_\tau\leq\frac{(2-\lambda_\tau)c^{\downarrow}}{\lambda_\tau(1-c^{\downarrow})}
=F_{\lambda_\tau}(c^{\downarrow})\qquad(m_\tau\geq3).
\]
Here $\theta_\tau$ is the actual Gram norm, whereas
$\overline{\theta}_\tau$ is the Gram bound computed from the geometric inputs. Monotonicity of
$F_\lambda$ lets us replace the child norms by their bounds
$\overline{\theta}_{\tau i}$, giving exactly the stated recursion.
The condition $c_\tau^{\downarrow}<\lambda_\tau/2$ is equivalent to
$F_{\lambda_\tau}(c_\tau^{\downarrow})<1$ on its domain. It therefore leaves a
positive retained fraction $1-\overline{\theta}_\tau$ of squared norm in the weighted
average over children, by \eqref{eq:weighted-norm-retention}.

We first prove the underlying Hilbert-space estimate for a single
reversible graph, and then apply it to the hierarchy in
\eqref{eq:weighted-hierarchy-gram-step}.
For the application to our hierarchy, fix a parent $\tau$ with
$m_\tau\geq3$. The spaces below describe directions removed after
one or two localization steps. We suppress the parent label in the
lemma to keep the notation short; the correspondence is
\begin{equation}\label{eq:weighted-child-increments}
\begin{aligned}
\II&=\II_\tau,\\
\mathcal U_i&=\mathcal U_{\tau,i}
=\II_\tau\cap\II_{\tau i}^\perp,\\
\mathcal B_{ij}&=\II_\tau\cap\II_{\tau ij}^\perp,\\
\mathcal V_{ij}&=\mathcal B_{ij}\cap\mathcal U_i^\perp
=\II_{\tau i}\cap\II_{\tau ij}^\perp
=\mathcal U_{\tau i,j}.
\end{aligned}
\end{equation}
Here $\tau ij\in X$, and orthogonal complements are taken in
$\II_\tau$. Thus $\mathcal U_i$ is removed at the parent,
$\mathcal B_{ij}$ is removed over two steps, and $\mathcal V_{ij}$
is removed in the second step, from $\tau i$ to $\tau ij$.
The state space is $N=N_\tau$, the set of child labels, with
$P=P_\tau$, $\pi=\pi_\tau$, and $\lambda=\lambda_\tau$.

We now formulate the estimate for an abstract subspace configuration
over a reversible Markov chain.
Let $P$ be a connected reversible Markov matrix on a finite state space
$N$ with at least two states, with $P(i,i)=0$ as for the link walks defined above.
Write $\pi_i>0$ for its stationary probabilities.
Let $\mathcal U_i$, $i\in N$, be linear subspaces of a finite-dimensional Hilbert
space $\II$, with orthogonal projections $E_i$.
For each edge $P_{ij}>0$, let a linear subspace
$\mathcal B_{ij}=\mathcal B_{ji}$ contain
$\mathcal U_i+\mathcal U_j$.
Write $\operatorname{Proj}_{\mathcal M}$ for the orthogonal projection onto
a subspace $\mathcal M$ in the relevant ambient Hilbert space. Put
\begin{equation}\label{eq:weighted-one-step-data}
\mathcal V_{ij}:=\mathcal B_{ij}\cap\mathcal U_i^\perp,
\qquad c^{\downarrow}:=\max_i\left\|\sum_jP_{ij}\operatorname{Proj}_{\mathcal V_{ij}}\right\|,
\qquad\theta:=\left\|\sum_i\pi_iE_i\right\|,
\end{equation}
where each sum over $j$ runs over states $j\in N$ with $P_{ij}>0$.
In the hierarchy configuration \eqref{eq:weighted-child-increments},
$\theta$ is the parent Gram norm $\theta_\tau$, and $c^{\downarrow}$ is the
largest child Gram norm, bounded by $c_\tau^{\downarrow}$.

\begin{lemma}[One-step Gram estimate]
\label{lem:weighted-one-step-gram}
In the setting above, suppose that
\[
\bigcap_{i\in N}\mathcal U_i=\{0\},\qquad
\operatorname{gap}(\id-P)\geq\lambda>0
\quad\text{in }L^2(\pi),\qquad c^{\downarrow}<1.
\]
Then
\begin{equation}\label{eq:weighted-one-step-bound}
\theta\leq F_\lambda(c^{\downarrow})=\frac{(2-\lambda)c^{\downarrow}}{\lambda(1-c^{\downarrow})}.
\end{equation}
\end{lemma}

\begin{remark}[Relation to the spectral criterion]
\label{rem:weighted-spectral-criterion}
The proof of \cref{lem:weighted-one-step-gram} uses the same Hilbert-space argument as
the spectral criterion of Ershov and
Jaikin-Zapirain~\cite[Theorem~5.1 and Section~5.3]{ershov2010property};
Section~5.5 of that work also treats weighted generalizations.
The subspaces $\mathcal U_i$ and $\mathcal B_{ij}$ play the roles of
the vertex and edge fixed spaces, respectively; the Laplacian is
normalized as $\id-P$ on $L^2(\pi)$.
For a $k$-regular graph with unnormalized Laplacian gap $\Lambda$,
$\lambda=\Lambda/k$, so \eqref{eq:weighted-one-step-bound} recovers
the coefficient $[c^{\downarrow}/(1-c^{\downarrow})](2k/\Lambda-1)$.
Here the argument is formulated for a fixed subspace configuration,
rather than as a group-theoretic statement over all representations.
We give a self-contained proof to make the connection to weighted
child Gram operators explicit.
\end{remark}

\begin{proof}[Proof of \cref{lem:weighted-one-step-gram}]
If every $\mathcal U_i$ is zero, then $\theta=0$ and the conclusion
holds. Otherwise $0<\theta<1$, using the zero intersection and finite
dimensionality.
The key is to encode subspace overlap in a mean-zero
vector field, where the Markov-chain spectral gap becomes available. Local
geometric bounds then connect this overlap to graph energy.

\paragraph{Vector fields and the Markov-chain gap.}
We work with vector fields assigning a vector $f_i\in\II$
to each state $i\in N$ of the Markov chain $P$, in the Hilbert space
\[
\mathcal H_\pi=L^2(\pi;\II),
\qquad \|f\|_\pi^2=\sum_{i\in N}\pi_i\|f_i\|^2.
\]
The constant fields assign the same vector $u\in\II$ to every state:
\[
\mathcal H_{\mathrm{const}}
=\{f\in\mathcal H_\pi:\text{there exists }u\in\II
\text{ such that }f_i=u\text{ for all }i\in N\}.
\]
The second subspace allows the vector to vary with the state, but requires
each value $f_i$ to lie in the corresponding subspace $\mathcal U_i$:
\[
\mathcal U_\pi
=\{f\in\mathcal H_\pi:f_i\in\mathcal U_i\text{ for all }i\in N\}.
\]
Their intersection is zero because $\bigcap_i\mathcal U_i=\{0\}$.
We use their sum and its mean-zero part,
\[
\mathcal Z=\mathcal H_{\mathrm{const}}+\mathcal U_\pi,
\qquad
\mathcal Z_0=\mathcal Z\cap\mathcal H_{\mathrm{const}}^\perp.
\]
Thus a field in $\mathcal Z$ is a common vector plus components in
the local removed spaces; membership in $\mathcal Z_0$ additionally
means $\sum_i\pi_i f_i=0$.

The Markov-chain Laplacian acts on these fields by
\[
(\Delta f)_i:=f_i-\sum_jP_{ij}f_j.
\]
In particular, it annihilates every constant field, since
$\sum_jP_{ij}=1$.
The space $\mathcal Z$ need not be invariant under $\Delta$, so we
apply $\Delta$ and then project back to $\mathcal Z$:
\[
T=\operatorname{Proj}_{\mathcal Z}\Delta|_{\mathcal Z}.
\]

\medskip
\noindent\textbf{Claim.}
The operator $T$ is self-adjoint and positive semidefinite, with
\[
\ker T=\mathcal H_{\mathrm{const}},
\qquad
\langle f,Tf\rangle_\pi\geq\lambda\|f\|_\pi^2
\quad(f\in\mathcal Z_0).
\]
In particular, $T$ preserves $\mathcal Z_0$ and is invertible on it.

\begin{proof}[Proof of the claim]
This compression preserves the energy. Indeed, for $f\in\mathcal Z$
the projection can be removed from the inner product, and reversibility gives
\[
\langle f,Tf\rangle_\pi
=\langle f,\Delta f\rangle_\pi
 =\frac12\sum_{i,j}\pi_iP_{ij}\|f_i-f_j\|^2.
\]
Thus $T$ is self-adjoint and positive semidefinite. The energy vanishes
exactly when the field has the same value at the two endpoints of
every edge. Since the chain is connected, such a field is constant
on $N$. All constant fields belong to $\mathcal Z$, so
\[
\ker T=\mathcal H_{\mathrm{const}}.
\]

Removing these zero-energy fields gives the orthogonal decomposition
\[
\mathcal Z=\mathcal H_{\mathrm{const}}\oplus\mathcal Z_0,
\qquad
\mathcal Z_0=\left\{f\in\mathcal Z:\sum_i\pi_i f_i=0\right\}.
\]
On this mean-zero part, the scalar Markov-chain Poincar\'e inequality,
applied in each orthonormal coordinate of $\II$ and summed, yields
\[
\langle f,Tf\rangle_\pi
\geq\lambda\|f\|_\pi^2
\qquad(f\in\mathcal Z_0).
\]
Since $T$ is self-adjoint and annihilates $\mathcal H_{\mathrm{const}}$,
it preserves $\mathcal Z_0$. Its restriction to this space is therefore
invertible, which will allow us to solve $T\psi=x$ for a mean-zero field $x$.
\end{proof}

\paragraph{Encoding Gram overlap in a mean-zero field.}
We now construct a field to which the preceding gap estimate applies,
while keeping its two orthogonal components explicitly tied to $\theta$.
The operator $\sum_i\pi_iE_i$ is positive semidefinite and self-adjoint,
so its norm $\theta$ is its largest eigenvalue. Choose a unit
eigenvector $v$ realizing this value:
\[
\sum_i\pi_iE_iv=\theta v.
\]
Thus the field of local projections $(E_iv)_i$ has average $\theta v$.
Since each $E_i$ is an orthogonal projection,
\[
\sum_i\pi_i\|E_iv\|^2
=\left\langle v,\sum_i\pi_iE_iv\right\rangle=\theta,
\qquad
\sum_i\pi_i\|(\id-E_i)v\|^2=1-\theta.
\]
Subtracting the average therefore gives
\[
\sum_i\pi_i\|E_iv-\theta v\|^2
=\theta-2\theta^2+\theta^2=\theta(1-\theta).
\]
Because $0<\theta<1$, we can normalize this centered field by setting
\[
x_i:=\frac{E_iv-\theta v}{\sqrt{\theta(1-\theta)}}.
\]
By construction,
\[
\sum_i\pi_i x_i=0,\qquad
\|x\|_\pi^2
=\frac{\theta-\theta^2}{\theta(1-\theta)}=1.
\]
The field is a sum of a constant field and a field in
$\mathcal U_\pi$, hence $x\in\mathcal Z_0$.
Its two orthogonal components at state $i\in N$ are
\[
E_i x_i=\frac{(1-\theta)E_iv}{\sqrt{\theta(1-\theta)}},
\qquad
(\id-E_i)x_i=\frac{-\theta(\id-E_i)v}{\sqrt{\theta(1-\theta)}}.
\]
Using the two averaged squared norms of $v$ above, we obtain
\begin{equation}\label{eq:weighted-complementary-angle}
\begin{aligned}
\sum_i\pi_i\|E_i x_i\|^2
&=\frac{(1-\theta)^2}{\theta(1-\theta)}\,\theta=1-\theta,\\
\sum_i\pi_i\|(\id-E_i)x_i\|^2
&=\frac{\theta^2}{\theta(1-\theta)}\,(1-\theta)=\theta.
\end{aligned}
\end{equation}

Thus $x$ is a unit mean-zero field whose average squared component
outside the local spaces $\mathcal U_i$ is exactly the Gram norm
$\theta$. It records the subspace geometry while lying in the
mean-zero sector controlled by the Markov-chain gap.

\paragraph{Combining child geometry with the Markov-chain gap.}
Recall that $c^{\downarrow}$ is the largest child Gram norm in the hierarchy,
as defined in \eqref{eq:weighted-one-step-data}. We now combine this
child-level geometric control with the Markov-chain gap $\lambda$
to bound the parent Gram norm $\theta$ encoded in $x$.

\medskip
\noindent\textbf{Claim.}
The Gram norm $\theta$ satisfies
\[
c^{\downarrow}(1-\theta)+\theta\leq\frac{2c^{\downarrow}}{\lambda}.
\]

\begin{proof}[Proof of the claim]
By \eqref{eq:weighted-complementary-angle}, the quantity to be bounded is
\[
c^{\downarrow}(1-\theta)+\theta
=c^{\downarrow}\sum_i\pi_i\|E_i x_i\|^2
+\sum_i\pi_i\|(\id-E_i)x_i\|^2.
\]
We estimate the two terms on the right in order: the first by Jensen's
inequality, and the second by the child Gram bound. Adding the estimates
\eqref{eq:weighted-local-component-bound} and
\eqref{eq:weighted-complement-component-bound}, and then applying the
Markov-chain gap, will prove the claim.

The Markov-chain gap makes $T$ invertible on $\mathcal Z_0$, so we can choose
$\psi\in\mathcal Z_0$ with $T\psi=x$.
Both estimates use the relation
$x=T\psi=\operatorname{Proj}_{\mathcal Z}\Delta\psi$
and the expression of $\Delta\psi$ in terms of edge differences.
Since $\sum_jP_{ij}=1$,
\[
(\Delta\psi)_i
=\psi_i-\sum_jP_{ij}\psi_j
=\sum_jP_{ij}(\psi_i-\psi_j).
\]
Thus the relevant quantities are the differences between the values of
$\psi$ at adjacent states. For each edge $P_{ij}>0$, write
\[
d_{ij}:=\psi_i-\psi_j,
\qquad (\Delta\psi)_i=\sum_jP_{ij}d_{ij}.
\]

\medskip
\noindent\textit{The first term: Jensen's inequality.}
Because $\mathcal U_\pi\subseteq\mathcal Z$, projecting onto
$\mathcal Z$ preserves the $\mathcal U_\pi$ component. Consequently,
\[
E_ix_i=E_i(\Delta\psi)_i=\sum_jP_{ij}E_id_{ij}.
\]
Jensen's inequality, followed by multiplication by $c^{\downarrow}\geq0$, gives
\begin{equation}\label{eq:weighted-local-component-bound}
c^{\downarrow}\sum_i\pi_i\|E_i x_i\|^2
\leq c^{\downarrow}\sum_{i,j}\pi_iP_{ij}\|E_id_{ij}\|^2.
\end{equation}

\medskip
\noindent\textit{The second term: the child Gram bound.}
Since $\psi\in\mathcal H_{\mathrm{const}}+\mathcal U_\pi$, we can
write $\psi_i=u+u_i$, where $u\in\II$ is common to all states and
$u_i\in\mathcal U_i$. The common component cancels, giving
\[
d_{ij}=u_i-u_j\in\mathcal U_i+\mathcal U_j
\subseteq\mathcal B_{ij}.
\]
Because $\mathcal U_i\subseteq\mathcal B_{ij}$, removing the
$\mathcal U_i$ component keeps the difference in $\mathcal B_{ij}$
and makes it orthogonal to $\mathcal U_i$. Hence
\[
(\id-E_i)d_{ij}\in
\mathcal B_{ij}\cap\mathcal U_i^\perp=\mathcal V_{ij}.
\]

For a fixed state $i\in N$, consider the map that adds vectors from
these spaces with weights $\sqrt{P_{ij}}$:
\[
\bigoplus_{j:P_{ij}>0}\mathcal V_{ij}\longrightarrow\II,
\qquad (z_j)_j\longmapsto\sum_j\sqrt{P_{ij}}\,z_j.
\]
This is the weighted synthesis map, with the orthogonal direct-sum
metric on its domain. By \cref{lem:gram-norm}, its squared operator
norm is $\|\sum_jP_{ij}\operatorname{Proj}_{\mathcal V_{ij}}\|$, which is at most
$c^{\downarrow}$ by \eqref{eq:weighted-one-step-data}. Hence
\[
\left\|\sum_j\sqrt{P_{ij}}\,z_j\right\|^2
\leq c^{\downarrow}\sum_j\|z_j\|^2
\qquad(z_j\in\mathcal V_{ij}).
\]
We apply this bound to the edge differences by taking
\[
z_j=\sqrt{P_{ij}}(\id-E_i)d_{ij}.
\]
The weighted sum on the left is then
$(\id-E_i)(\Delta\psi)_i$, so
\[
\|(\id-E_i)(\Delta\psi)_i\|^2
\leq c^{\downarrow}\sum_jP_{ij}\|(\id-E_i)d_{ij}\|^2.
\]
Because $\mathcal U_\pi\subseteq\mathcal Z$, the projection
$\operatorname{Proj}_{\mathcal Z}$ preserves the $\mathcal U_\pi$ component and
contracts its orthogonal complement. Using
$x=\operatorname{Proj}_{\mathcal Z}\Delta\psi$, we consequently have
\begin{equation}\label{eq:weighted-complement-component-bound}
\sum_i\pi_i\|(\id-E_i)x_i\|^2
\leq\sum_i\pi_i\|(\id-E_i)(\Delta\psi)_i\|^2
\leq c^{\downarrow}\sum_{i,j}\pi_iP_{ij}\|(\id-E_i)d_{ij}\|^2.
\end{equation}

\medskip
\noindent\textit{Combining the estimates.}
Add \eqref{eq:weighted-local-component-bound} and
\eqref{eq:weighted-complement-component-bound}. Both orthogonal
parts of each edge difference have the same coefficient $c^{\downarrow}$.
The left-hand side is the quantity identified at the start of the proof;
on the right, the two orthogonal parts of each $d_{ij}$ recombine.
Thus orthogonality and reversibility yield
\begin{align*}
c^{\downarrow}(1-\theta)+\theta
&\leq c^{\downarrow}\sum_{i,j}\pi_iP_{ij}\|d_{ij}\|^2
=2c^{\downarrow}\langle\psi,\Delta\psi\rangle\\
&=2c^{\downarrow}\langle\psi,T\psi\rangle
\leq\frac{2c^{\downarrow}}{\lambda}\|T\psi\|^2
=\frac{2c^{\downarrow}}{\lambda}.
\end{align*}
The penultimate inequality uses $T\succeq\lambda\id$ on
$\mathcal Z_0$, and the last equality uses $T\psi=x$ and
$\|x\|_\pi=1$.
\end{proof}

Rearranging the claim now gives
\[
(1-c^{\downarrow})\theta\leq\frac{2c^{\downarrow}}{\lambda}-c^{\downarrow}
=\frac{(2-\lambda)c^{\downarrow}}{\lambda}.
\]
Division by $1-c^{\downarrow}>0$ proves the result.
\end{proof}

\begin{proof}[Proof of \cref{prop:weighted-gram-descent}]
Fix $\tau$ with $m_\tau\geq3$. Apply \cref{lem:weighted-one-step-gram}
with the spaces in \eqref{eq:weighted-child-increments}.
Nesting gives $\mathcal U_i+\mathcal U_j\subseteq\mathcal B_{ij}$,
and the span condition gives $\bigcap_i\mathcal U_i=\{0\}$.
By \eqref{eq:weighted-link-consistency}, the projection sum for a
fixed $i$ in \cref{lem:weighted-one-step-gram} is exactly the
child's weighted increment sum, extended by zero from $\II_{\tau i}$
to $\II_\tau$. Zero extension does not change its norm, so the local
parameter in the lemma is
\[
c^{\downarrow}=\max_{i\in N_\tau}
\left\|\sum_j\pi_{\tau i}(j)
\operatorname{Proj}_{\mathcal U_{\tau i,j}}|_{\II_{\tau i}}\right\|
=\max_{i\in N_\tau}\theta_{\tau i}.
\]
Therefore
\begin{equation}\label{eq:weighted-hierarchy-gram-step}
\theta_\tau\leq
F_{\lambda_\tau}\left(\max_{i\in N_\tau}\theta_{\tau i}\right).
\end{equation}
Starting at \eqref{eq:weighted-bottom-input}, use monotonicity of
$F_\lambda$ on $[0,1)$ and
\eqref{eq:weighted-retention-condition} to conclude inductively that
$\theta_\tau\leq \overline{\theta}_\tau<1$ at every internal node. In particular,
the hypothesis $c^{\downarrow}<1$ of the one-step lemma holds at each application.
\end{proof}

\subsection{The complete uniform hierarchy}
\label{sec:weighted-complete-case}
We now specialize the general descent theorem to the complete uniform
hierarchy, where the recursion yields an explicit spectral-gap lower bound.
For $X=2^V$, with $|V|=n$, there is one facet, so $d=n$.
Its induced weights are uniform at every node.
Writing $m:=m_\tau=n-|\tau|$, for $m\geq2$ we have
\begin{equation}\label{eq:weighted-complete-parameters}
\pi_\tau(i)=\frac1m,\qquad
P_\tau(i,j)=\frac1{m-1}\ (i\ne j),\qquad
\operatorname{gap}(\id-P_\tau)=\frac m{m-1}.
\end{equation}
Thus the link-walk gap input is explicit and only the bottom geometric
bound remains to be supplied. For comparison with unweighted
increment Gram operators, define the unweighted Gram operator
$\bm C_\tau$ and record its norm:
\begin{equation}\label{eq:weighted-uniform-gram-notation}
\bm C_\tau:=m\bm C_\tau^\pi=(J_{\tau,i}^*J_{\tau,j})_{i,j},
\qquad
\|\bm C_\tau\|
=\left\|\sum_iE_{\tau,i}|_{\II_\tau}\right\|=m\theta_\tau.
\end{equation}

At $m=2$, with remaining vertices $i,j$, the unweighted Gram operator is
\[
\begin{pmatrix}
\id_{\mathcal U_{\tau,i}}&J_{\tau,i}^*J_{\tau,j}\\
J_{\tau,j}^*J_{\tau,i}&\id_{\mathcal U_{\tau,j}}
\end{pmatrix}.
\]
If at least one increment is nonzero, its norm is
$1+\|J_{\tau,i}^*J_{\tau,j}\|$: the off-diagonal self-adjoint
block has extreme eigenvalues equal to plus and minus the largest
singular value of $J_{\tau,i}^*J_{\tau,j}$, with the zero-block case
immediate. If both increments vanish, $\theta_\tau=0$.
Hence an angle bound $\|J_{\tau,i}^*J_{\tau,j}\|\leq\epsilon$
implies $\theta_\tau\leq(1+\epsilon)/2$ in either case.

\begin{corollary}[Complete uniform descent]
\label{cor:weighted-complete-gap}
Let $X=2^V$, $|V|=n\geq2$, and assume the subspace conditions
\eqref{eq:weighted-space-conditions} and the energy assembly
\eqref{eq:weighted-energy-assembly} with weights $1/(n-|\tau|)$.
Suppose every stopping energy has the common lower bound $\delta>0$,
and
\begin{equation}\label{eq:weighted-complete-input}
\theta_\tau\leq\frac{1+\epsilon}{2}\quad(|\tau|=n-2),
\qquad 0\leq\epsilon<\frac1{n-1}.
\end{equation}
Then, at each depth $k=|\tau|$ with $0\leq k\leq n-2$, the
levelwise Gram bound is
\begin{equation}\label{eq:weighted-complete-gram}
\theta_\tau\leq \overline{\theta}(k):=
\frac{1+\epsilon}{(n-k)[1-(n-k-2)\epsilon]}<1.
\end{equation}
If $\II_\varnothing\ne\{0\}$, then $\LOp_\varnothing$ is positive
definite and its spectral gap satisfies
\begin{equation}\label{eq:weighted-complete-gap}
\operatorname{gap}(\LOp_\varnothing)
\geq\frac{\delta[1-(n-1)\epsilon]}n.
\end{equation}
\end{corollary}
\begin{proof}
The base is $\overline{\theta}(n-2)=(1+\epsilon)/2$.
For a parent at depth $k\leq n-3$, write $m=n-k$ for the
number of remaining vertices. Substitution gives
\[
F_{m/(m-1)}(\overline{\theta}(k+1))
=\frac{m-2}{m}\frac{\overline{\theta}(k+1)}{1-\overline{\theta}(k+1)}
=\overline{\theta}(k).
\]
The assumption on $\epsilon$ ensures
$\overline{\theta}(k+1)<m/[2(m-1)]$, as required by
\eqref{eq:weighted-retention-condition}. At any internal depth $k$,
with $m=n-k\geq2$, the norm-retention factor is
\[
1-\overline{\theta}(k)
=\frac{(m-1)[1-(m-1)\epsilon]}{m[1-(m-2)\epsilon]}.
\]
Indexing the subsequent factors by their remaining vertex counts
$q=2,\ldots,m$, both the ratios $(q-1)/q$
and the ratios $[1-(q-1)\epsilon]/[1-(q-2)\epsilon]$ telescope.
Multiplying by the terminal energy lower bound $\delta$ gives
\begin{equation}\label{eq:weighted-complete-telescope}
\delta\prod_{q=2}^m(1-\overline{\theta}(n-q))
=\frac{\delta[1-(m-1)\epsilon]}m.
\end{equation}
The nodewise estimate \eqref{eq:weighted-gap-conclusion} in the proof of
\cref{thm:weighted-link-gap} therefore gives
\[
\Eff_\tau(b)\geq
\frac{\delta[1-(m-1)\epsilon]}m\|b\|^2
\qquad(b\in\II_\tau)
\]
at every proper face; the empty product at $m=1$ gives the terminal energy lower bound.
Taking $\tau=\varnothing$ proves the asserted positivity and root gap.
\end{proof}

\begin{remark}[The same endpoint bound from subspace angles]
\label{rem:weighted-kassabov}
The coefficient in \cref{cor:weighted-complete-gap} also follows from
Kassabov's distance estimate~\cite[Theorem~5.1]{kassabov2011subspace}.
Take $B_i=\II_{V\setminus\{i\}}^\perp$ inside $\II_\varnothing$.
Under the span condition, the bottom Gram assumption bounds the
pairwise Friedrichs-angle cosines of these spaces by $\epsilon$
(with zero increments contributing zero).
Kassabov's estimate, combined with uniform energy assembly and terminal
energy lower bounds, gives the same root bound $\delta[1-(n-1)\epsilon]/n$;
applying it within each face also gives the node bounds.
Thus the complete uniform coefficient is inherited from subspace-angle
estimates: our proof exhibits its propagation through the hierarchy
within the general formulation, rather than a new endpoint constant.
\end{remark}

The factor $1/n$ in the root bound corresponds to normalized
averaging over root labels.
The dependence of the gap lower bound on size-dependent terminal and
overlap inputs is discussed in \cref{par:outlook-inputs}.

\subsection{Descent with Gram bounds at every level}
\label{sec:levelwise-gram-descent}
The main theorem derives Gram bounds throughout the hierarchy from
bottom-level Gram bounds and intermediate link-walk gaps.
Alternatively, one may supply Gram bounds directly at every internal
node and use only the energy-transport step. The following result
formulates this second route, which will be used in the comparison
with spectral independence in \cref{sec:classical-si}.

\begin{corollary}[Descent from levelwise Gram bounds]
\label{cor:levelwise-gram-descent}
Consider the weighted face hierarchy with subspaces and localization
maps satisfying \cref{eq:weighted-space-conditions,eq:weighted-localization}.
Assume the energy assembly and positive terminal energy lower bounds in
\cref{eq:weighted-energy-assembly,eq:weighted-terminal-input}.
Suppose that Gram bounds are given at every internal node:
\[
\theta_\tau\leq\overline{\theta}_\tau<1
\qquad(\tau\in\mathscr T_{\mathrm{int}}),
\]
where the Gram norm $\theta_\tau$ is defined by \eqref{eq:weighted-gram}.
If $\II_\varnothing\ne\{0\}$, then $\LOp_\varnothing$ is positive
definite and
\[
\operatorname{gap}(\LOp_\varnothing)
\geq\min_{p\in\mathcal P_{\mathrm{stop}}}
\left\{\delta_{t(p)}
\prod_{k=0}^{d-2}(1-\overline{\theta}_{\tau_k})\right\}.
\]
Here $p=(\tau_0=\varnothing,\tau_1,\ldots,\tau_{d-1}=t(p))$
is a path to a stopping node, and each factor
$1-\overline{\theta}_{\tau_k}$ uses the given Gram bound at the
corresponding internal node.
Moreover, if common bounds are available at each level,
\[
\theta_\tau\leq\overline{\theta}(k)<1
\qquad(|\tau|=k,\quad k=0,\ldots,d-2),
\]
then
\[
\operatorname{gap}(\LOp_\varnothing)
\geq
\left(\min_{\sigma\in\mathscr T_{\mathrm{stop}}}\delta_\sigma\right)
\prod_{k=0}^{d-2}\bigl(1-\overline{\theta}(k)\bigr).
\]
No link-walk gap hypothesis or Gram-recursion condition is required.
\end{corollary}
\begin{proof}
The given Gram bounds and \eqref{eq:weighted-norm-retention} give
retention factors $1-\overline{\theta}_\tau$.
Define the propagated energy lower bounds $\delta_\tau$
by \eqref{eq:weighted-gap-recursion}
and repeat the energy induction in \cref{sec:gap-transport-proof}.
For the levelwise conclusion, take
$\overline{\theta}_\tau=\overline{\theta}(|\tau|)$.
Every path to a stopping node passes through exactly one internal node
at each level $k=0,\ldots,d-2$, so the product is independent of the
path and can be taken outside the minimum.
\end{proof}

\section{Comparison with Classical Local-to-Global Methods}
\label{sec:classical-comparison}
Two correspondences explain what the operator framework extends and
what its geometric quantities measure in classical settings:
\begin{enumerate}
\item \emph{Energy decomposition.} Total variance is a scalar
realization of ROL: effective energy is variance, and the gluing
defect is the variance of conditional means
(\cref{prop:classical-variance-rol}).
\item \emph{Subspace geometry.} For full-support binary laws,
conditional influence describes the overlap encoded by increment Gram
operators (\cref{prop:classical-si-gram}). Spectral-independence bounds on
conditional influence therefore yield Gram bounds, recovering the
classical spectral-independence gap product
(\cref{cor:classical-si-gap}).
\end{enumerate}

\subsection{Total variance as a classical realization of ROL}
\label{sec:classical-dictionaries}
We take the scalar specialization of the link construction in
\cref{sec:toolkit-links}: set $q=1$ and $W_\sigma=w_\sigma>0$.
The following setting is shared by the total-variance decomposition
and the down-up gap argument below.
Use the weighted complex of \cref{sec:weighted-hierarchy}, with
exactly $d$ vertices in each facet, where $d\ge2$, and positive
scalar weights $w_\sigma$.
Let $\Omega_\tau$ be the facets containing $\tau$, and write
\[
W_\tau=\sum_{\sigma\in\Omega_\tau}w_\sigma,\qquad
m_\tau=d-|\tau|.
\]
The facet distribution and, for $m_\tau\ge1$, the distribution on
legal next vertices $N_\tau$ are
\[
\mu_\tau(\sigma)=\frac{w_\sigma}{W_\tau},\qquad
\pi_\tau(i)=\frac{W_{\tau i}}{m_\tau W_\tau}
=\frac1{m_\tau}\Pr_{\sigma\sim\mu_\tau}
\bigl[i\in\sigma\setminus\tau\bigr].
\]
Thus sampling a facet with law $\mu_\tau$ and then choosing one of
its $m_\tau$ remaining vertices uniformly gives the vertex law
$\pi_\tau$. Conditional on choosing $i$, the facet law is
$\mu_{\tau i}$.
Throughout this subsection, $f\in L^2(\mu_\tau)$ is a scalar-valued
function on $\Omega_\tau$. At a child $\tau i$, we use its restriction
$f|_{\Omega_{\tau i}}$; expectations and variances under $\mu_{\tau i}$
are written with the same symbol $f$.

\subsubsection{Conditional means and the gluing defect}
We first identify the ROL data in this scalar setting.
The retained variable is $f$, while the extra $\C$ summand in
$\HH_\tau=\II_\tau\oplus\C$ carries the common scalar $\alpha$
over which we minimize.
For an arrow $a_i:\tau\to\tau i$, the ROL data are
\[
\begin{aligned}
\II_\tau&=\C^{\Omega_\tau},
&\HH_\tau&=\II_\tau\oplus\C,
&\Gamma_\tau(f,\alpha)&=f,\\
R_i f&=f|_{\Omega_{\tau i}},
&r_i(f,\alpha)&=(R_i f,\alpha),
&\omega_{a_i}&=\frac1{m_\tau}.
\end{aligned}
\]
The full energy is
\[
\Energy_\tau(f,\alpha)
=\sum_{\sigma\in\Omega_\tau}w_\sigma|f_\sigma-\alpha|^2.
\]
Localization restricts $f$ but preserves the same scalar $\alpha$
across all children. Independent minimization at the children can
instead choose a different scalar at each child; the gluing defect
measures the cost of reconciling those choices.

For these ROL data, \cref{lem:link-effective-variance} identifies the minimizing response
and effective energy as
\begin{equation}\label{eq:scalar-moment-dictionary}
\Av_\tau(f)=\mathbb E_{\mu_\tau}f,\qquad
\Eff_\tau(f)=\Var_\tau(f)
=W_\tau\operatorname{Var}_{\mu_\tau}(f).
\end{equation}
Here $\Var_\tau$ is the unnormalized energy, whereas
$\operatorname{Var}_{\mu_\tau}$ is probability variance.
The next proposition identifies the decomposition across children
with the classical total-variance formula.

\begin{proposition}[Total variance as a ROL decomposition]
\label{prop:classical-variance-rol}
For $m_\tau\ge1$, write $R_i f=f|_{\Omega_{\tau i}}$.
The effective-energy decomposition is
\begin{equation}\label{eq:classical-effective-decomposition}
\Eff_\tau(f)
=\frac1{m_\tau}\sum_i\Eff_{\tau i}(R_i f)
+W_\tau\operatorname{Var}_{i\sim\pi_\tau}
  \bigl(\mathbb E_{\mu_{\tau i}}f\bigr).
\end{equation}
The final term is the gluing defect for ROL step weights $1/m_\tau$.
Equivalently, dividing by $W_\tau$ gives the law of total variance:
\begin{equation}\label{eq:classical-link-variance}
\operatorname{Var}_{\mu_\tau}(f)
=\sum_i\pi_\tau(i)\operatorname{Var}_{\mu_{\tau i}}(f)
+\operatorname{Var}_{i\sim\pi_\tau}
  \bigl(\mathbb E_{\mu_{\tau i}}f\bigr).
\end{equation}
\end{proposition}
\begin{proof}
Each parent facet belongs to exactly $m_\tau$ children.
Consequently, $\sum_iW_{\tau i}=m_\tau W_\tau$, so
$\pi_\tau$ is a probability distribution. The same counting gives
exact assembly while keeping the representative $\alpha$ common to
all children:
\[
\Energy_\tau(f,\alpha)
=\frac1{m_\tau}\sum_i\Energy_{\tau i}(R_i f,\alpha).
\]

Write $m_i:=\mathbb E_{\mu_{\tau i}}f$ for the mean chosen by child $i$.
The scalar case of \eqref{eq:mean-response} gives
\[
\Energy_{\tau i}(R_i f,\alpha)
=\Eff_{\tau i}(R_i f)+W_{\tau i}|m_i-\alpha|^2.
\]
Substituting into the assembly identity and using
$W_{\tau i}/m_\tau=W_\tau\pi_\tau(i)$ yields
\[
\Energy_\tau(f,\alpha)
=\frac1{m_\tau}\sum_i\Eff_{\tau i}(R_i f)
+W_\tau\sum_i\pi_\tau(i)|m_i-\alpha|^2.
\]
The first term is independent of $\alpha$. Minimizing the second
over the common representative gives
\[
\min_{\alpha\in\C}\sum_i\pi_\tau(i)|m_i-\alpha|^2
=\operatorname{Var}_{i\sim\pi_\tau}(m_i),
\]
attained at $\alpha=\sum_i\pi_\tau(i)m_i$.
Thus minimizing the full energy proves
\eqref{eq:classical-effective-decomposition}: the gluing defect is
exactly the cost of replacing the separate child means by one
common representative.

Finally, divide by $W_\tau$ and use
\eqref{eq:scalar-moment-dictionary} at the parent and each child.
The coefficient of the child variance becomes
$W_{\tau i}/(m_\tau W_\tau)=\pi_\tau(i)$, giving
\eqref{eq:classical-link-variance}.
\end{proof}

Separate children choose separate means, while the parent uses one
common mean: their mismatch is precisely the defect.
ROL replaces this minimization over a common mean by minimization
under a general linear readout, and tracks the cost of reconciling
local minimizers in the same way.
This is the classical variance mechanism used in local-to-global
arguments such as \cite{guo2020local,chen2023optimal}.
ROL also allows matrix weights, which need not define a scalar
probability law.

\subsubsection{From link expansion to the down-up gap}
For the same facet functions and probability distributions fixed above,
the total-variance decomposition in \cref{prop:classical-variance-rol}
and local link expansion give a global down-up gap directly, without
a separate Gram-bound recursion. The following corollary recovers
the classical bound of Alev--Lau~\cite[Theorem~1.5]{alev2020improved}.

\begin{corollary}[Recovery of the down-up gap bound]
\label{cor:classical-down-up-gap}
In the scalar setting above, suppose $X$ has at least two facets and
the non-lazy link walks in \eqref{eq:weighted-link-walk} satisfy
\[
\beta_2(P_\tau)\le a_k<1
\qquad(|\tau|=k,\quad k=0,\ldots,d-2).
\]
Let $P_{\rm DU}$ be the walk on facets that deletes a uniformly chosen
vertex and then resamples a facet containing the remaining face
$\eta$ according to $\mu_\eta$. With one deletion and resampling
counted as one step,
\begin{equation}\label{eq:classical-down-up-gap}
\operatorname{gap}(I-P_{\rm DU})
\ge\frac1d\prod_{k=0}^{d-2}(1-a_k).
\end{equation}
\end{corollary}

\begin{proof}
\textit{One-step variance retention.}
For $m=m_\tau\ge2$, define the conditional-mean map
\begin{equation}\label{eq:classical-mean-map}
T_\tau:L^2(\mu_\tau)\longrightarrow L^2(\pi_\tau),
\qquad (T_\tau f)(i)=\mathbb E_{\mu_{\tau i}}f.
\end{equation}
With the probability inner products, the bridge between conditional
means and the link walk is
\begin{equation}\label{eq:classical-mean-walk}
T_\tau T_\tau^*=\frac{I+(m-1)P_\tau}{m}.
\end{equation}
Indeed, let $g\in L^2(\pi_\tau)$ be any scalar-valued function on
the link vertices $N_\tau$. The adjoint of the map in
\eqref{eq:classical-mean-map} sends $g$ to a function on facets
$\sigma\in\Omega_\tau$: sampling a facet and then a uniform remaining
vertex gives
\[
(T_\tau^*g)(\sigma)=\frac1m
  \sum_{i\in\sigma\setminus\tau}g(i).
\]
Conditioning this average on $i$ yields a term $g(i)/m$.
For $j\ne i$, its coefficient is
$W_{\tau ij}/(mW_{\tau i})=(m-1)P_\tau(i,j)/m$.
This proves \eqref{eq:classical-mean-walk}.

If $\beta_2(P_\tau)\le a_\tau<1$, the identity bounds the squared
norm of $T_\tau$ on mean-zero functions by $[1+(m-1)a_\tau]/m$.
Thus the normalized defect in \eqref{eq:classical-link-variance}
satisfies
\[
\operatorname{Var}_{i\sim\pi_\tau}
  \bigl(\mathbb E_{\mu_{\tau i}}f\bigr)
\le\frac{1+(m-1)a_\tau}{m}
       \operatorname{Var}_{\mu_\tau}(f).
\]
Subtracting this upper bound from the parent variance gives
\begin{equation}\label{eq:classical-variance-retention}
\sum_i\pi_\tau(i)\operatorname{Var}_{\mu_{\tau i}}(f)
\ge\frac{m-1}{m}(1-a_\tau)
       \operatorname{Var}_{\mu_\tau}(f).
\end{equation}

\medskip
\noindent\textit{Iteration and the down-up gap.}
Let $\mu=\mu_\varnothing$ and obtain $\eta$ by deleting a uniform
vertex from a sampled facet. Its law is
$q(\eta)=W_\eta/(dW_\varnothing)$ for $|\eta|=d-1$.
The facet down-up walk $P_{\rm DU}$ resamples a facet from
$\mu_\eta$, so
\begin{equation}\label{eq:classical-down-up-form}
\langle f,(I-P_{\rm DU})f\rangle_\mu
=\mathbb E_{\eta\sim q}\operatorname{Var}_{\mu_\eta}(f).
\end{equation}
This follows because, conditionally on $\eta$, the facets before
and after resampling are independent samples from $\mu_\eta$.

Starting at $\varnothing$, choose each next vertex with law
$\pi_\tau$; equivalently, sample a facet with law $\mu$ and reveal
its vertices in uniformly random order. The endpoint at depth
$d-1$ therefore has law $q$.
At depth $k$, \eqref{eq:classical-variance-retention} retains a
fraction at least $(d-k-1)(1-a_k)/(d-k)$ of the variance.
Iteration and \eqref{eq:classical-down-up-form} give the bound:
the rank factors telescope to $1/d$, and positivity identifies
the kernel as the constants.
\end{proof}

Our facet size $d$ corresponds to dimension $d-1$ in
\cite{alev2020improved}, with the same one-step normalization.

Here \eqref{eq:classical-mean-walk} gives variance retention directly
from link expansion, without separate Gram bounds. This realizes
the descent mechanism of \cref{thm:weighted-link-gap}, but is not a
direct application: the conditional spaces need not form its
common-metric nested hierarchy.

\subsection{Conditional correlations as subspace geometry}
\label{sec:classical-si}
The preceding scalar variance argument obtained retention directly
from conditional means and the link walk, without a separate
Gram-geometric estimate. We now make the subspace geometry explicit:
placing conditional problems in one common Hilbert space lets us
interpret the overlap quantities used in geometric descent as
classical conditional correlations.
The central correspondence is \cref{prop:classical-si-gram}, which
identifies Gram overlap with conditional influence.
\subsubsection{Spin configurations, links, and observable hierarchies}
\label{sec:classical-configurations}
We specialize the observable-space construction of
\cref{sec:toolkit-observables} to a full-support probability law
$\mu$ on a finite product $\prod_{i\in V}\mathcal X_i$.
Here $V$ is the set of sites, $\mathcal X_i$ is the finite state
space at site $i$, and $X_i(x)=x_i$ is the coordinate random variable.
For $S\subseteq V$, write $X_S=(X_i)_{i\in S}$.
In this classical setting, an observable is a function $f(x)$ of the
spin configuration. Observables supported on $S$ depend only on $X_S$
and form the subspace $\mathcal A_S$. In the $L^2(\mu)$ metric, the
best approximation to $f$ from this subspace is the conditional
expectation $\mathbb E_\mu[f\mid X_S]$. Subtracting it leaves the
residual component orthogonal to all observables supported on $S$.
The dictionary for the subspace hierarchy is
\begin{equation}\label{eq:classical-observable-dictionary}
\begin{aligned}
\text{node}&:\ S\subseteq V,\qquad \mathcal K=L^2(\mu),\\
\mathcal A_S&=\{g(X_S):g:\textstyle\prod_{i\in S}\mathcal X_i\to\C\},\\
\Pi_S^{\mathrm{obs}}f&=\mathbb E_\mu[f\mid X_S],\\
\II_S&=\mathcal A_S^\perp
 =\{f:\mathbb E_\mu[f\mid X_S]=0\},\\
R_{S\cup\{i\}\leftarrow S}f
&=f-\mathbb E_\mu[f\mid X_{S\cup\{i\}}],
\qquad f\in\II_S.
\end{aligned}
\end{equation}
Thus localization removes the component newly explained by observing
site $i$. The corresponding removed space is
\[
\mathcal U_{S,i}
=\II_S\cap\II_{S\cup\{i\}}^\perp
=\mathcal A_{S\cup\{i\}}\cap\mathcal A_S^\perp.
\]
The overlap of these spaces will encode \emph{conditional correlations}.
The conditional laws have the same facet-weight interpretation as
in the preceding example, but the nodes $S$ here record sites without
fixing their values. To relate the two descriptions, we use the
configuration complex below.

The configuration-complex construction of
Anari, Liu, and Oveis Gharan~\cite[Section~1.1]{anari2020spectral}
identifies binary Glauber dynamics with a walk on facets.
For finite spins, the vertices are site--value pairs in
\[
\mathcal V=\bigsqcup_{i\in V}(\{i\}\times\mathcal X_i).
\]
Each configuration $x$ selects one vertex at each site, giving a facet
\[
\sigma_x=\{(i,x_i):i\in V\},\qquad
w(\sigma_x)=\mu(x).
\]
The configuration complex consists of all subsets of supported facets:
\[
X^\mu=\{\tau:\tau\subseteq\sigma_x
\text{ for some }x\in\operatorname{supp}\mu\}.
\]
A face $\tau_{S,s}=\{(i,s_i):i\in S\}$ fixes a partial configuration.
Equivalently, join every pair of distinct values at the same site,
forming the disjoint union of complete graphs
$G=\bigsqcup_{i\in V}K_{\{i\}\times\mathcal X_i}$.
An independent set of $G$ chooses at most one value at each site.
Thus, under full support, $X^\mu=\operatorname{Ind}(G)$, the complex
of independent sets of $G$ (\Cref{fig:configuration-independent-sets}).

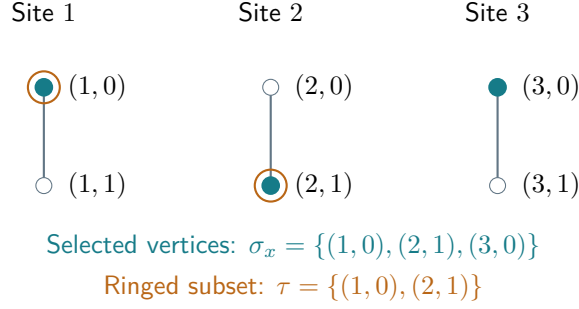
\begin{figure}[tbp]
\centering
\begin{tikzpicture}[x=1cm,y=1cm,font=\small\sffamily]
\foreach \i/\x in {1/0,2/3,3/6}{
  \node at (\x,1) {Site $\i$};
  \draw[linkgray,thick] (\x,0)--(\x,-1.3);
  \draw[linkgray,fill=white] (\x,0) circle (3pt);
  \draw[linkgray,fill=white] (\x,-1.3) circle (3pt);
  \node[anchor=west] at (\x+.18,0) {$(\i,0)$};
  \node[anchor=west] at (\x+.18,-1.3) {$(\i,1)$};
}
\foreach \x/\y in {0/0,3/-1.3,6/0}{
  \fill[linkteal] (\x,\y) circle (3.5pt);
}
\foreach \x/\y in {0/0,3/-1.3}{
  \draw[defectorange,thick] (\x,\y) circle (6pt);
}
\node[text=linkteal,align=center] at (3.3,-2.1)
  {Selected vertices: $\sigma_x=\{(1,0),(2,1),(3,0)\}$};
\node[text=defectorange,align=center] at (3.3,-2.65)
  {Ringed subset: $\tau=\{(1,0),(2,1)\}$};
\end{tikzpicture}
\caption{Configuration faces as independent sets for three binary sites.
The graph is the disjoint union of three copies of $K_2$; its edges
join mutually exclusive values at the same site, not adjacent vertices
of the configuration complex. The three filled vertices form the facet
for $x=(0,1,0)$, and the two ringed vertices form a face fixing the first
two spins. Under full support, all independent sets are faces.}
\label{fig:configuration-independent-sets}
\end{figure}

The normalized facet weights of the link of $\tau_{S,s}$ are
\[
\mu(x_{V\setminus S}\mid X_S=s)
=\frac{w(\sigma_x)}{\sum_{y:y_S=s}w(\sigma_y)}.
\]
Deleting a uniformly chosen vertex of a facet and resampling its
replacement by weight is exactly a single-site heat-bath update.
Thus the facet down-up operator is
\[
P_{\mathrm{HB}}=\frac1n\sum_{i\in V}
\mathbb E_\mu[\,\cdot\mid X_{V\setminus\{i\}}],
\qquad n=|V|.
\]
This uses one update per step, not one update per site.

Unlike the configuration complex, the complete site hierarchy records
only $S\subseteq V$, without fixing spin values. The map
$\tau_{S,s}\mapsto S$ relates the two descriptions.
For a probability law $\nu$, write
\[
L^2_0(\nu):=\{g\in L^2(\nu):\mathbb E_\nu g=0\}
\]
for its mean-zero subspace. When $\nu=\mu(\,\cdot\mid X_S=s)$,
the mean is taken under that conditional law.

\begin{lemma}[Conditional mean-zero spaces and the heat-bath operator]
\label{prop:classical-fibers}
In the observable hierarchy \eqref{eq:classical-observable-dictionary},
assume $n\ge2$ and $|\mathcal X_i|\ge2$ for every site $i$.
Writing $f_s$ for $f$ with $X_S=s$ fixed, we have the isometric
identification
\begin{equation}\label{eq:classical-fiber-decomposition}
\II_S\cong\bigoplus_s L^2_0(\mu(\,\cdot\mid X_S=s)),
\qquad
\|f\|_\mu^2=\sum_s\mu(X_S=s)\|f_s\|_{\mu(\cdot\mid s)}^2.
\end{equation}
The direct sum carries the displayed weighted metric.
These spaces satisfy the nesting and span conditions
\eqref{eq:weighted-space-conditions} on the complete site hierarchy.
Moreover,
\begin{equation}\label{eq:classical-reference-heat-bath}
I-P_{\mathrm{HB}}=\frac1n\sum_i(I-\Pi_{V\setminus\{i\}}^{\mathrm{obs}}).
\end{equation}
Its restriction to $\II_\varnothing=L^2_0(\mu)$ has
the same positive gap as $I-P_{\mathrm{HB}}$.
\end{lemma}
\begin{proof}
\textit{Conditional mean-zero spaces.}
Since $\Pi_S^{\mathrm{obs}}$ is the orthogonal projection onto
$\mathcal A_S$, we have
\[
f\in\II_S=\mathcal A_S^\perp
\quad\Longleftrightarrow\quad
\mathbb E_\mu[f\mid X_S=s]=0\quad\text{for every }s.
\]
Thus each restriction $f_s$ belongs to
$L^2_0(\mu(\,\cdot\mid X_S=s))$. Conversely, any collection of such
restrictions defines a function in $\II_S$. Splitting its squared norm
according to the value of $X_S$ gives
\[
\begin{aligned}
\|f\|_\mu^2
&=\sum_x\mu(x)|f(x)|^2\\
&=\sum_s\mu(X_S=s)
  \sum_{x:\,x_S=s}\mu(x\mid X_S=s)|f(x)|^2\\
&=\sum_s\mu(X_S=s)\,
  \mathbb E_\mu[|f|^2\mid X_S=s],
\end{aligned}
\]
which proves the weighted isometry in
\eqref{eq:classical-fiber-decomposition}. Projection onto
$\II_{S\cup\{i\}}$ is $I-\Pi_{S\cup\{i\}}^{\mathrm{obs}}$,
so localization subtracts the finer conditional mean.

\medskip
\noindent\textit{Nesting and spanning.}
The inclusion $\mathcal A_S\subseteq\mathcal A_{S\cup\{i\}}$
gives $\II_{S\cup\{i\}}\subseteq\II_S$, and
$\mathcal A_V=L^2(\mu)$ gives $\II_V=\{0\}$.
When at least two sites remain, full support gives
\[
\bigcap_{i\notin S}\mathcal A_{S\cup\{i\}}=\mathcal A_S.
\]
Indeed, for any $j\notin S$, choose another $i\notin S$.
A function in the intersection depends only on $X_{S\cup\{i\}}$,
so it does not change when $X_j$ alone is varied. Full support
allows every such variation; hence the function depends only on $X_S$.
Taking orthogonal complements in the finite-dimensional ambient space yields
\[
\operatorname{span}_{i\notin S}\II_{S\cup\{i\}}=\II_S.
\]
The increments are also nonzero: for a nonconstant function $h$ on
$\mathcal X_i$, the function
$h(X_i)-\mathbb E_\mu[h(X_i)\mid X_S]$ belongs to
$\mathcal A_{S\cup\{i\}}\cap\mathcal A_S^\perp=\mathcal U_{S,i}$
and is nonzero by full support.

\medskip
\noindent\textit{The heat-bath kernel and gap.}
The heat-bath formula above gives
\eqref{eq:classical-reference-heat-bath}. Since each summand is an
orthogonal projection,
\[
\langle f,(I-P_{\mathrm{HB}})f\rangle_\mu
=\frac1n\sum_i\|(I-\Pi_{V\setminus\{i\}}^{\mathrm{obs}})f\|_\mu^2.
\]
This vanishes exactly when $f\in\mathcal A_{V\setminus\{i\}}$
for every $i$. Such a function is unchanged by varying any single
coordinate, so full support implies that it is constant. Thus the
kernel consists precisely of the constants. Its orthogonal complement
is $\II_\varnothing=L^2_0(\mu)$; restricting to this space removes
only the zero eigenvalue and leaves the positive gap unchanged.
\end{proof}

Thus the residual hierarchy collects all boundary-conditioned
problems in one observable space; it is not a relabeling of a
single link. 

\paragraph{Reference energy on the residual hierarchy.}
To apply descent to this operator, we extend it to an energy family
on the site hierarchy. The following construction uses only the
subspace geometry and gives exact assembly with a unit terminal energy lower bound.

In the complete uniform hierarchy of \cref{sec:weighted-complete-case},
write $S$ for the site-set node (the general index $\tau$) and define, for $m=n-|S|\geq1$,
\begin{equation}\label{eq:weighted-reference-energy}
\widetilde{\LOp}^{\mathrm{ref}}_S
:=\frac1m\sum_{j\notin S}\Pi_{V\setminus\{j\}}
\quad\text{on }\mathcal K,\qquad
\LOp^{\mathrm{ref}}_S
:=\widetilde{\LOp}^{\mathrm{ref}}_S|_{\II_S}.
\end{equation}
Each summand in \eqref{eq:weighted-reference-energy} has range
contained in $\II_S$ and annihilates $\II_S^\perp$, so
$\II_S$ reduces the sum and the restriction is well defined.
For $m\geq2$, a fixed $j\notin S$
occurs in exactly $m-1$ children, with inner coefficient $1/(m-1)$
and outer coefficient $1/m$. Thus
\[
\widetilde{\LOp}^{\mathrm{ref}}_S
=\frac1m\sum_{i\notin S}
\widetilde{\LOp}^{\mathrm{ref}}_{S\cup\{i\}}.
\]
Each child operator is supported on $\II_{S\cup\{i\}}$, so restriction
to $\II_S$ gives exact assembly with the localization maps
$R_a=\Pi_{S\cup\{i\}}|_{\II_S}$. Thus the quadratic forms
$\mathcal E_S^{\mathrm{ref}}(b):=\langle b,\LOp_S^{\mathrm{ref}}b\rangle$
satisfy \eqref{eq:weighted-energy-assembly} with equality.
At a stopping node, $m=1$ and the only summand is $\Pi_S$. Hence
\[
\LOp_S^{\mathrm{ref}}=\id_{\II_S},\qquad
\mathcal E_S^{\mathrm{ref}}(b)=\|b\|_\mu^2
\quad(b\in\II_S),
\]
so \eqref{eq:weighted-terminal-input} holds with $\delta_S=1$ and equality.

For the conditional mean-zero spaces above,
$\Pi_S=I-\Pi_S^{\mathrm{obs}}$, so
\eqref{eq:classical-reference-heat-bath} identifies
$\widetilde{\LOp}^{\mathrm{ref}}_\varnothing=I-P_{\mathrm{HB}}$.
Thus the heat-bath gap can be bounded by applying descent to this
reference family, once the Gram bounds are supplied.

\subsubsection{Binary influence and increment Gram}
We now specialize the observable hierarchy
\eqref{eq:classical-observable-dictionary} to binary spins,
$\mathcal X_i=\{0,1\}$. The removed spaces are
\[
\mathcal U_{S,i}
=\{g(X_S,X_i):\mathbb E_\mu[g(X_S,X_i)\mid X_S]=0\}.
\]
Here $g$ ranges over scalar-valued functions of the spins at sites
$S\cup\{i\}$, subject to the displayed conditional mean-zero condition.
After fixing $X_S=s$, they are one-dimensional spaces of
centered single-variable functions. Write
\[
\mu_{S,s}:=\mu(\,\cdot\mid X_S=s).
\]
Fix $S,s$ with $m=n-|S|\ge2$ remaining sites. Define
\[
\Sigma^{\mu_{S,s}}_{ij}=\operatorname{Cov}_{\mu_{S,s}}(X_i,X_j),\qquad
D^{\mu_{S,s}}=\operatorname{diag}\bigl(\operatorname{Var}_{\mu_{S,s}}(X_i)\bigr).
\]
Here covariance and variance are taken under the conditional law $\mu_{S,s}$:
\[
\operatorname{Cov}_{\mu_{S,s}}(X_i,X_j)
=\mathbb E_{\mu_{S,s}}\bigl[(X_i-\mathbb E_{\mu_{S,s}} X_i)(X_j-\mathbb E_{\mu_{S,s}} X_j)\bigr],
\qquad
\operatorname{Var}_{\mu_{S,s}}(X_i)=\operatorname{Cov}_{\mu_{S,s}}(X_i,X_i).
\]
Thus $\Sigma^{\mu_{S,s}}$ is the covariance matrix of the remaining spins,
and $D^{\mu_{S,s}}$ is the diagonal matrix of their variances.
All diagonal entries of $D^{\mu_{S,s}}$ are positive. Fix the signed
influence convention
\begin{equation}\label{eq:classical-influence}
\Psi^{\mu_{S,s}}_{ij}
=\Pr_{\mu_{S,s}}[X_j=1\mid X_i=1]-\Pr_{\mu_{S,s}}[X_j=1\mid X_i=0]
=\frac{\Sigma^{\mu_{S,s}}_{ij}}{D^{\mu_{S,s}}_{ii}}\quad(i\ne j),
\qquad \Psi^{\mu_{S,s}}_{ii}=0.
\end{equation}

The next theorem identifies conditional influence with increment Gram
geometry, providing the input for the spectral-independence gap bound below.

\begin{theorem}[Binary influence and increment Gram]
\label{prop:classical-si-gram}
For the conditional law $\mu_{S,s}$, define the normalized covariance matrix
\begin{equation}\label{eq:classical-si-gram}
\widehat\Sigma^{\mu_{S,s}}:=(D^{\mu_{S,s}})^{-1/2}\Sigma^{\mu_{S,s}}(D^{\mu_{S,s}})^{-1/2}.
\end{equation}
The same correlation data admit the following descriptions:
\begin{itemize}
\item \emph{Increment Gram.} Let $\bm C^{\mu_{S,s}}$ be the Gram
matrix in $L^2(\mu_{S,s})$ of the normalized centered functions
\[
u_i:=\frac{X_i-\mathbb E_{\mu_{S,s}} X_i}
{\sqrt{\operatorname{Var}_{\mu_{S,s}}(X_i)}}\qquad(i\notin S).
\]
These represent the single-site increment directions at the fixed
assignment $X_S=s$.\footnote{Since $X_i$ is binary, its mean-zero
functions under $\mu_{S,s}$ form a one-dimensional space spanned by the
displayed unit vector $u_i$. Writing the inclusion as $J_i a=a u_i$,
the Gram blocks are
$J_i^*J_j=\langle u_i,u_j\rangle_{L^2(\mu_{S,s})}=C^{\mu_{S,s}}_{ij}$.}
This Gram matrix equals the normalized covariance matrix:
\begin{equation}\label{eq:classical-gram-covariance}
\bm C^{\mu_{S,s}}=\widehat\Sigma^{\mu_{S,s}}.
\end{equation}
\item \emph{Influence.} The signed influence matrix $\Psi^{\mu_{S,s}}$
is similar to $\widehat\Sigma^{\mu_{S,s}}-I$:
\[
\widehat\Sigma^{\mu_{S,s}}-I=(D^{\mu_{S,s}})^{1/2}\Psi^{\mu_{S,s}}(D^{\mu_{S,s}})^{-1/2}.
\]
\end{itemize}
Using $\|\bm C_S\|=m\theta_S$ from
\eqref{eq:weighted-uniform-gram-notation}, taking all assignments
$X_S=s$ into account gives
\begin{equation}\label{eq:classical-si-maximum}
\theta_S
=\frac{1+\max_s\lambda_{\max}
 \bigl(\Psi^{\mu_{S,s}}\bigr)}{m}.
\end{equation}
\end{theorem}
The Gram--influence identity \eqref{eq:classical-si-maximum} is the
relation used in \cref{cor:classical-si-gap} below to turn
spectral-independence bounds into levelwise Gram bounds.
\begin{proof}
\textbf{Increment Gram and conditional covariance.}
Recall from \cref{def:subspace-gram-operator} that
$\bm C_S=\mathsf S_S^*\mathsf S_S$, where the summation map
\[
\mathsf S_S:\bigoplus_{i\notin S}\mathcal U_{S,i}\longrightarrow\II_S,
\qquad (v_i)_{i\notin S}\longmapsto\sum_{i\notin S}v_i
\]
adds the increment vectors in the common ambient space.
For each assignment $s$, use the normalized covariance matrix
$\widehat\Sigma^{\mu_{S,s}}$ defined in \eqref{eq:classical-si-gram}.
We first show that the full increment Gram operator splits into
blocks indexed by the conditioning assignment:
\begin{equation}\label{eq:classical-gram-direct-sum}
\bm C_S\simeq\bigoplus_s\widehat\Sigma^{\mu_{S,s}},
\qquad m\theta_S=\|\bm C_S\|
=\max_s\|\widehat\Sigma^{\mu_{S,s}}\|.
\end{equation}
To identify these blocks, let $p_S(s)=\mu(X_S=s)$.
For each $i\notin S$, the functions
\[
h_{s,i}(x)=
\frac{\mathbf1_{\{x_S=s\}}}{\sqrt{p_S(s)}}\,
\frac{x_i-\mathbb E_{\mu_{S,s}}X_i}
{\sqrt{\operatorname{Var}_{\mu_{S,s}}(X_i)}}
\]
as $s$ varies form an orthonormal basis of $\mathcal U_{S,i}$.
Functions associated with different assignments $s$ have disjoint
supports and are therefore orthogonal. For a fixed $s$, since
$h_{s,i}=\mathbf1_{\{X_S=s\}}u_i/\sqrt{p_S(s)}$, we have
\[
\begin{aligned}
\|h_{s,i}\|_\mu^2
&=\frac1{p_S(s)}\mathbb E_\mu
  \left[\mathbf1_{\{X_S=s\}}|u_i|^2\right]\\
&=\frac{\mu(X_S=s)}{p_S(s)}
  \mathbb E_\mu[|u_i|^2\mid X_S=s]\\
&=\|u_i\|_{L^2(\mu_{S,s})}^2=1.
\end{aligned}
\]
The last equality uses the conditional normalization of $u_i$.
Moreover, because $X_i$ is binary, its conditional
mean-zero functions form a one-dimensional space at each assignment
$s$. Thus these functions span $\mathcal U_{S,i}$.

Let $e_{s,i}$ denote the vector in the external direct sum whose
$i$-th component is $h_{s,i}$ and whose other components are zero.
These vectors form an orthonormal basis of the direct sum.
Since $\mathsf S_S e_{s,i}=h_{s,i}$, the matrix entries of the
increment Gram operator in this basis are
\[
\langle e_{s,i},\bm C_S e_{t,j}\rangle
=\langle \mathsf S_S e_{s,i},\mathsf S_S e_{t,j}\rangle_\mu
=\langle h_{s,i},h_{t,j}\rangle_\mu.
\]
For $s\ne t$, the supports are disjoint, so these entries vanish.
For a fixed assignment $s$, the entries within its block are
\[
\langle h_{s,i},h_{s,j}\rangle_\mu
=\langle u_i,u_j\rangle_{L^2(\mu_{S,s})}
=\frac{\operatorname{Cov}_{\mu_{S,s}}(X_i,X_j)}
{\sqrt{\operatorname{Var}_{\mu_{S,s}}(X_i)
        \operatorname{Var}_{\mu_{S,s}}(X_j)}}
=\widehat\Sigma^{\mu_{S,s}}_{ij}.
\]
Since the middle inner product is $C^{\mu_{S,s}}_{ij}$ by definition,
this also proves \eqref{eq:classical-gram-covariance}.
Ordering the basis vectors $e_{s,i}$ by $s$ therefore gives
the direct-sum decomposition, and taking its operator norm gives the
maximum over assignments. This establishes both identities in
\eqref{eq:classical-gram-direct-sum}.

\medskip\noindent\textbf{Conditional influence.}
For $i\ne j$, put $p_i=\mu_{S,s}(X_i=1)$ and
$q_a=\mu_{S,s}(X_j=1\mid X_i=a)$ for $a\in\{0,1\}$.
Conditioning on the two values of $X_i$ gives
$\mathbb E_{\mu_{S,s}}X_j=p_iq_1+(1-p_i)q_0$ and
$\mathbb E_{\mu_{S,s}}[X_iX_j]=p_iq_1$. Hence
\[
\begin{aligned}
\operatorname{Cov}_{\mu_{S,s}}(X_i,X_j)
&=p_iq_1-p_i\bigl(p_iq_1+(1-p_i)q_0\bigr)\\
&=p_i(1-p_i)(q_1-q_0)
=D^{\mu_{S,s}}_{ii}\Psi^{\mu_{S,s}}_{ij}.
\end{aligned}
\]
This proves \eqref{eq:classical-influence}. Since $\Psi^{\mu_{S,s}}$
has zero diagonal, the matrix identity is
$\Sigma^{\mu_{S,s}}=D^{\mu_{S,s}}(I+\Psi^{\mu_{S,s}})$.
Normalizing by the positive diagonal matrix $D^{\mu_{S,s}}$ gives
\[
\widehat\Sigma^{\mu_{S,s}}-I
=(D^{\mu_{S,s}})^{1/2}\Psi^{\mu_{S,s}}(D^{\mu_{S,s}})^{-1/2}.
\]
Thus the influence matrix has real eigenvalues, and positivity of
$\widehat\Sigma^{\mu_{S,s}}$ implies
\[
\|\widehat\Sigma^{\mu_{S,s}}\|
=\lambda_{\max}(\widehat\Sigma^{\mu_{S,s}})
=1+\lambda_{\max}(\Psi^{\mu_{S,s}}).
\]
Substituting into \eqref{eq:classical-gram-direct-sum} and dividing
by $m$ proves \eqref{eq:classical-si-maximum}.

\end{proof}

The maximum over boundary values $s$ in \eqref{eq:classical-si-maximum},
rather than an average, is what
connects the hierarchy to spectral independence
\cite{anari2020spectral}.

\subsubsection{Recovery of the spectral-independence gap product}
The preceding correspondence turns levelwise spectral-independence
inputs into Gram bounds. Applying descent recovers the classical
product bound of \cite[Theorem~1.3]{anari2020spectral}.

\begin{corollary}[Recovery of the spectral-independence gap bound]
\label{cor:classical-si-gap}
Let $\mu$ have full support on $\{0,1\}^V$, with $n=|V|\ge2$.
Suppose every conditional law $\mu_{S,s}$ with $|S|=k$ satisfies
$\lambda_{\max}(\Psi^{\mu_{S,s}})\le\eta_k<n-k-1$,
for $k=0,\ldots,n-2$, where $\Psi^{\mu_{S,s}}$ is the signed influence
matrix in \eqref{eq:classical-influence}.
Then the single-site heat-bath walk, with one update per step, satisfies
\begin{equation}\label{eq:classical-si-gap-product}
\operatorname{gap}(I-P_{\rm HB})
\ge\frac1n\prod_{k=0}^{n-2}
\left(1-\frac{\eta_k}{n-k-1}\right).
\end{equation}
\end{corollary}
\begin{proof}
By \eqref{eq:classical-si-maximum}, these give levelwise Gram bounds
\[
\theta_S\leq\overline{\theta}_S
:=\frac{1+\eta_k}{n-k}<1,
\qquad |S|=k.
\]
The reference energy family in \eqref{eq:weighted-reference-energy}
has exact assembly and terminal energy lower bounds $\delta_\sigma=1$;
\cref{prop:classical-fibers} identifies its root with the heat-bath
generator on mean-zero functions. Applying
\cref{cor:levelwise-gram-descent}, with
\[
1-\overline{\theta}_S
=\frac{n-k-1}{n-k}
\left(1-\frac{\eta_k}{n-k-1}\right),
\]
and telescoping the rank factors to $1/n$ proves
\eqref{eq:classical-si-gap-product}.
\end{proof}

Alternatively, the influence--link correspondence of
\cite[Theorem~1.11]{anari2020spectral} gives the bound from
\eqref{eq:classical-down-up-gap} with $a_k=\eta_k/(n-k-1)$.
Independent spins give $\eta_k=0$ and the gap $1/n$.

This application uses supplied levelwise Gram bounds; the main descent
theorem (\cref{thm:weighted-link-gap}) additionally derives such bounds
from bottom-level geometry and intermediate link-walk gaps.

\section{Spectral control from the gluing defect}
\label{sec:compatible-assembly}
\label{sec:gluing-response-cost}
The same ROL decomposition supports two ways of obtaining spectral
lower bounds. Recall from \eqref{eq:rol-one-step-lemma} that
\[
\LOp_\tau
=\underbrace{\sum_{a:\tau\to\nu}\omega_aR_a^*\LOp_\nu R_a}
  _{\text{energy inherited from the children}}
+\underbrace{\DOp_\tau}_{\text{gluing defect}},
\qquad \DOp_\tau\succeq0.
\]
ROL supplies this exact decomposition and the nonnegativity of the
defect. The descent in \cref{sec:weighted-descent} discards the defect
and combines child energy bounds with geometric norm retention to
obtain a global gap. Here we instead ask: \emph{can the gluing defect
control directions that are invisible to the lower bound inherited
from the children?}

To answer this question, nonnegativity alone is not enough: we need a
quantitative lower bound on the defect in those missing directions.
Its quadratic form
$\Delta_\tau(b)=\langle b,\DOp_\tau b\rangle_\tau$
is the cost of reconciling local minimizing choices through a common
parent variable. We first state a one-step spectral estimate, then express the defect
in terms of minimizing responses to make its input concrete. These are one-step results; their recursive use is left for
future work.

\subsection{Spectral complementation by the defect}
\label{sec:defect-complementation}
At one node, suppress the parent index and write $\LOp=\AOp+\DOp$,
with $\AOp=\sum_i\omega_iR_i^*\LOp_{\nu_i}R_i$ the energy inherited
from the children. The following estimate shows how a lower bound on
$\DOp$ can complement a lower bound on $\AOp$.

\begin{proposition}[Complementation of two energy lower bounds]
\label{thm:defect-complementation}
Let $\LOp=\AOp+\DOp$ on a finite-dimensional Hilbert space $\II$,
where $\AOp,\DOp\succeq0$. Let $\AOp_0=F_0^*F_0$ be a positive
semidefinite lower bound for $\AOp$ and $B_d:\II\to\mathcal G$ a linear map into a Hilbert
space $\mathcal G$. Assume
\begin{equation}\label{eq:defect-complement-data}
0\preceq\AOp_0\preceq\AOp,\qquad
\DOp\succeq B_d^*B_d,\qquad
\LOp\mathcal N=\{0\},
\end{equation}
for a known subspace $\mathcal N\subseteq\II$. Set
\[
\mathcal U=\ker F_0,\qquad
\mathcal U_0=\mathcal U\cap\mathcal N^\perp,\qquad
\mathcal W=\mathcal U^\perp.
\]
Since $0\preceq\AOp_0\preceq\LOp$ and $\LOp\mathcal N=\{0\}$,
we have $\mathcal N\subseteq\ker F_0=\mathcal U$.
Thus $\mathcal N^\perp=\mathcal U_0\oplus\mathcal W$:
$\mathcal U_0$ consists of the directions in $\mathcal N^\perp$
not controlled by $\AOp_0$, while $\mathcal W$ is controlled by $\AOp_0$.
If $\mathcal U_0$ and $\mathcal W$ are nonzero, suppose that
\begin{equation}\label{eq:defect-complement-input}
\AOp_0\succeq a\operatorname{Proj}_{\mathcal W},\qquad
\|B_du\|\ge\alpha_d\|u\|\ (u\in\mathcal U_0),\qquad
\|B_dv\|\le\beta_d\|v\|\ (v\in\mathcal W),
\end{equation}
where $a,\alpha_d>0$ and $\beta_d\ge0$. Here $a$ bounds $\AOp_0$
from below on $\mathcal W$, while $\alpha_d$ and $\beta_d$ bound
the action of $B_d$ from below on $\mathcal U_0$ and from above on
$\mathcal W$, respectively. Then
\begin{equation}\label{eq:defect-complement-gap}
\LOp\succeq
\frac{a\alpha_d^2}{a+\alpha_d^2+\beta_d^2}\operatorname{Proj}_{\mathcal N^\perp},
\qquad \ker\LOp=\mathcal N.
\end{equation}
If $\mathcal W=0$ and $\mathcal N^\perp\ne0$, the lower bound
on $B_d|_{\mathcal U_0}$ alone gives $\alpha_d^2\operatorname{Proj}_{\mathcal N^\perp}$.
If $\mathcal U_0=0$ and $\mathcal N^\perp\ne0$, the bound on
$\AOp_0$ alone gives $a\operatorname{Proj}_{\mathcal N^\perp}$.
\end{proposition}
\begin{proof}
For $x=u+v\in\mathcal N^\perp$ in the stated orthogonal decomposition,
put $E_0=\|F_0x\|^2$ and $E_d=\|B_dx\|^2$.
Then $\|v\|\le\sqrt{E_0/a}$ and
$\alpha_d\|u\|\le\sqrt{E_d}+\beta_d\sqrt{E_0/a}$. Hence
\[
\|x\|^2\le E_0/a+
\alpha_d^{-2}\bigl(\sqrt{E_d}+\beta_d\sqrt{E_0/a}\bigr)^2
\le\frac{a+\alpha_d^2+\beta_d^2}{a\alpha_d^2}(E_0+E_d).
\]
The last bound follows by bounding the largest eigenvalue of the
positive $2\times2$ coefficient matrix by its trace.
The bounds in \eqref{eq:defect-complement-data} give
$E_0+E_d\le\langle x,\LOp x\rangle$ and the zero action on
$\mathcal N$, proving both conclusions. The two boundary cases follow
directly from the remaining lower bound.
If $\mathcal N^\perp=0$, the operator is zero and there is no
nonzero-space gap assertion.
\end{proof}

The ROL consequence can now be stated without repeating the subspace
definitions and comparison conditions.

\begin{corollary}[Spectral complementation for ROL]
\label{cor:rol-defect-complementation}
Assume the one-step ROL hypotheses of \cref{lem:rol-one-step} at
$\tau$, and put
$\AOp=\sum_{a:\tau\to\nu}\omega_aR_a^*\LOp_\nu R_a$.
If $\AOp$, $\DOp_\tau$, and a known zero space $\mathcal N$ satisfy
the hypotheses of \cref{thm:defect-complementation}, with both
$\mathcal U_0$ and $\mathcal W$ nonzero, then
\[
\operatorname{gap}(\LOp_\tau)
\geq\frac{a\alpha_d^2}{a+\alpha_d^2+\beta_d^2},
\qquad \ker\LOp_\tau=\mathcal N.
\]
The same proposition gives the bounds when either subspace is zero.
\end{corollary}
\begin{proof}
The ROL identity \eqref{eq:rol-one-step-lemma} gives
$\LOp_\tau=\AOp+\DOp_\tau$. Apply \cref{thm:defect-complementation}.
\end{proof}

Only the selected map $B_d$ needs an upper bound on $\mathcal W$,
not the full defect; the parameter $\beta_d$ controls possible
cancellation between $\mathcal U_0$ and $\mathcal W$.
The constants must be supplied independently of the unknown parent gap.
The final estimate is a standard two-form comparison.
The zero space $\mathcal N$ must be verified independently of the
desired gap. The next subsection constructs the defect factor needed
in \cref{cor:rol-defect-complementation} from the ROL response data.

\subsection{The gluing defect as a response cost}
To make the defect input in \cref{cor:rol-defect-complementation}
concrete, we express its energy as the cost of reconciling local
minimizing responses.
Fix a node $\tau$ and enumerate its outgoing arrows as
$a_i:\tau\to\nu_i$, abbreviating
$r_i:=r_{a_i}$, $R_i:=R_{a_i}$, and $\omega_i:=\omega_{a_i}$.
Under the one-step hypotheses, the defect $\Delta_\tau(b)$ defined in
\eqref{eq:one-step-effective-assembly} vanishes
exactly when a parent lift of $b$ makes every localized vector an
energy minimizer. Indeed, choose a parent minimizer $f$ with
$\Gamma_\tau f=b$, so that $\Eff_\tau(b)=\Energy_\tau(f)$.
The defect definition and the full-energy assembly identity
\eqref{eq:rol-one-step-full-assembly} then give
\[
\begin{aligned}
\Delta_\tau(b)
&=\Eff_\tau(b)-\sum_i\omega_i\Eff_{\nu_i}(R_i b)\\
&=\Energy_\tau(f)-\sum_i\omega_i\Eff_{\nu_i}(R_i b)\\
&=\sum_i\omega_i\Energy_{\nu_i}(r_i f)
  -\sum_i\omega_i\Eff_{\nu_i}(R_i b)\\
&=\sum_i\omega_i
\bigl(\Energy_{\nu_i}(r_i f)-\Eff_{\nu_i}(R_i b)\bigr).
\end{aligned}
\]
Each difference compares the energy of the localized parent minimizer
with the minimum attainable independently at that child.
Compatibility makes $r_i f$ a feasible lift of $R_i b$.
All summands are therefore nonnegative, so their sum is zero precisely when
each local energy attains its minimum. Conversely, any common lift
realizing all local minima gives zero defect.

For an explicit formula, choose the coordinates of
\eqref{eq:response-coordinates} at each $\nu_i$, with right inverse $S_i$ of
$\Gamma_{\nu_i}$ and internal space $\YY_i=\ker\Gamma_{\nu_i}$.
Assume its internal cost $\TOp_i$ is positive definite, so that
\cref{prop:minimizing-response} supplies a response $h_i$ and the identity
\begin{equation}\label{eq:localized-response-cost}
\Energy_{\nu_i}(S_i b_i+y_i)
=\Eff_{\nu_i}(b_i)
+\langle y_i-h_i b_i,\TOp_i(y_i-h_i b_i)\rangle_{\YY_i}.
\end{equation}
For $\Gamma_\tau f=b$, compatibility fixes $b_i=R_i b$, while the
induced internal variable is $y_i=r_i f-S_iR_i b\in\YY_i$.

To compare the independently minimizing child responses with those
attainable from a single parent variable, we collect the child internal
variables and their weighted costs into one space.
The superscript $\downarrow$ indicates data collected from the children.
Equip $\YY^{\downarrow}=\bigoplus_i\YY_i$ with the product metric and set
\begin{equation}\label{eq:stacked-response-cost}
h^{\downarrow}b=(h_iR_i b)_i,\qquad
\TOp^{\downarrow}=\bigoplus_i \omega_i\TOp_i\succ0.
\end{equation}
The weights are included in the cost operator, not in the product metric.
To describe the jointly attainable tuples, take a finite-dimensional
constraint Hilbert space $\mathcal C$ and linear maps
$C:\YY^{\downarrow}\to\mathcal C$ and $F:\II_\tau\to\mathcal C$.
The required condition is the exact description
\begin{equation}\label{eq:attainable-response-tuples}
\bigl\{(r_i f-S_iR_i b)_i:\Gamma_\tau f=b\bigr\}
=\{y\in\YY^{\downarrow}:Cy=Fb\}
\qquad(b\in\II_\tau).
\end{equation}
This equality, not merely an inclusion, is required.
The sets are nonempty because $\Gamma_\tau$ is surjective.

\medskip
\begin{lemma}[Response representation of the gluing defect]
\label{prop:gluing-response-cost}
Assume the hypotheses of \cref{lem:rol-one-step}, the
positive definite local costs in \eqref{eq:localized-response-cost}, and
the constraint description \eqref{eq:attainable-response-tuples}. Then
\begin{equation}\label{eq:gluing-response-minimum}
\Delta_\tau(b)=\min_{Cy=Fb}
\langle y-h^{\downarrow}b,
\TOp^{\downarrow}(y-h^{\downarrow}b)\rangle_{\YY^{\downarrow}}.
\end{equation}
Its representing operator is
\begin{equation}\label{eq:gluing-response-operator}
\DOp_\tau=(F-Ch^{\downarrow})^*
(C(\TOp^{\downarrow})^{-1}C^*)^\dagger
(F-Ch^{\downarrow}).
\end{equation}
Here $\dagger$ denotes the Moore--Penrose inverse in the specified
Hilbert metric. In particular, $\Delta_\tau(b)=0$ if and only if
$Ch^{\downarrow}b=Fb$.
\end{lemma}
\medskip
\begin{proof}
For a feasible $f$, substitute \eqref{eq:localized-response-cost} into
the exact assembly \eqref{eq:rol-one-step-full-assembly}. This separates the
fixed sum $\sum_i \omega_i\Eff_{\nu_i}(R_i b)$ from the cost of the induced
tuple $y$. Minimizing and using \eqref{eq:attainable-response-tuples}
gives \eqref{eq:gluing-response-minimum}.

Put $u=(\TOp^{\downarrow})^{1/2}(y-h^{\downarrow}b)$ and
$V=C(\TOp^{\downarrow})^{-1/2}$. The remaining problem is to minimize
$\|u\|^2$ subject to $Vu=(F-Ch^{\downarrow})b$.
Feasibility places the right-hand side in $\ran V$. For $q\in\ran V$,
the minimum-norm solution of $Vu=q$ is $V^*(VV^*)^\dagger q$,
whose squared norm is $\langle q,(VV^*)^\dagger q\rangle_{\mathcal C}$.
Since $VV^*=C(\TOp^{\downarrow})^{-1}C^*$, this proves
\eqref{eq:gluing-response-operator}. Finally, positive definiteness of
the cost implies that the minimum is zero exactly when
$h^{\downarrow}b$ itself is feasible.
\end{proof}

Thus the defect is the squared distance from $h^{\downarrow}b$ to
the attainable affine set in the $\TOp^{\downarrow}$ metric.
The pseudoinverse allows redundant constraints, but the explicit formula
still requires $\TOp_i\succ0$. The one-step lemma and its zero-defect criterion do not require
this additional assumption.

In particular, suppressing the parent index,
\eqref{eq:gluing-response-operator} supplies the factor
\[
B_d=(C(\TOp^{\downarrow})^{-1}C^*)^{\dagger/2}
     (F-Ch^{\downarrow}),
\qquad \DOp=B_d^*B_d.
\]
This gives a concrete choice for the defect input in
\cref{cor:rol-defect-complementation}; bounds on its action on
$\mathcal U_0$ and $\mathcal W$ supply $\alpha_d$ and $\beta_d$.

\paragraph{Toward recursive defect control.}
The results above concern a single localization step; the recursive gap
theorem in \cref{sec:weighted-descent} does not use this quantitative
defect contribution. Incorporating it into a recursive bound would
require tracking the subspaces controlled and not controlled by $\AOp_0$ at each
node, together with quantitative response estimates and verified zero
spaces. Obtaining tractable conditions for this iteration, and bounds
that improve on discarding the defect, remains a direction for future work.

\section*{Acknowledgments}
RH was supported by JST PRESTO (Grant No.~JPMJPR23F9)
and JST ASPIRE (Grant No.~JPMJAP26A4), Japan.

\section*{AI usage}
The author developed this manuscript through iterative interaction with AI systems, using GPT-5.6 Sol and GPT-6 Astra via Codex for brainstorming, developing and checking proofs, drafting and revising the exposition, preparing figures, and obtaining review-style feedback. The author takes responsibility for the mathematical claims and the final presentation.

\bibliographystyle{alpha}
\bibliography{rol}

\end{document}